\documentclass[11pt]{article}
\usepackage{amsmath,amssymb,amsthm,amsfonts} 
\usepackage{mathtools}                      
\numberwithin{equation}{subsection}
\usepackage{color}          
\usepackage{cite}                           
\usepackage[margin=1in]{geometry}           
\usepackage{bm}                             
\usepackage{physics}                        
\usepackage{enumitem}                       
\usepackage{fancyhdr}                       
\usepackage{empheq}                         
\usepackage{authblk}      
\usepackage{bm}                    
\usepackage{graphicx}                       
\usepackage[utf8]{inputenc}                 
\usepackage{caption}                        
\usepackage{subcaption}                     
\usepackage{tikz}     
\usetikzlibrary{arrows.meta,calc}                      
\usepackage{booktabs}                       
\usepackage{fancyhdr}
\fancypagestyle{titlepagefooter}{ 
    \fancyhf{} 
    \fancyfoot[C]{Email: \text{lillian.kleess@baruchmail.cuny.edu}} 
}
\newtheorem{axiom}{Axiom}
\theoremstyle{plain}
\newtheorem{theorem}{Theorem}[section]
\newtheorem{proposition}[theorem]{Proposition}
\newtheorem{lemma}[theorem]{Lemma}
\newtheorem{corollary}[theorem]{Corollary}

\theoremstyle{definition}
\newtheorem{definition}[theorem]{Definition}

\theoremstyle{remark}
\newtheorem{remark}[theorem]{Remark}

\begin{document}

\title{Open‐System Virus Particle Physics: A Path‐Integral Viral Lattice Theory Using Non‐Self‐Adjoint Stochastic PDEs and Fock‐Space Formalism}
\author{Lillian. St. Kleess}
\affil{Baruch College, Department of Natural Sciences, Manhattan, USA}
\date{February 17, 2025}
\maketitle

\begin{abstract}
We develop a comprehensive theoretical biophysics model—grounded in a \emph{path‐integral} perspective and an \emph{$m$‐sectorial} open‐system framework—to describe complex, damped viral phonon dynamics in resource‐limited and noise‐driven environments. By unifying \emph{wave mechanics} (via PDEs with multiplicative noise), \emph{Markov jumps} for occupant or arrangement transitions, and \emph{second‐quantized} (Fock‐space) expansions, our construction accommodates an \emph{unbounded} number of viral lattices in a single ``global wavefunction.'' In doing so, we capture how an entire population—potentially numbering in the millions—may be represented by a single operator‐theoretic state (or ``orbit’’), whose evolution is governed by non‐unitary semigroups with potential equilibrium or non‐equilibrium steady states. This approach admits action functionals over the space of system trajectories, enabling large‐deviation and flux analyses whenever \emph{detailed balance} is broken by sustained resource inputs, as often happens in real infections.

\medskip
\noindent
In earlier work, four axioms established a baseline for \emph{viral lattice theory}: \emph{Metabolic Inertia}, \emph{Deterministic Physical Interactions}, \emph{Self‐Organization into Periodic Lattices}, and \emph{Conservation of Energy}. Here, adopting a non‐equilibrium viewpoint and leveraging topological flux tools from infinite‐dimensional geometry, we show that persistent probability‐current loops arise whenever resource inflows disrupt detailed balance. This discovery motivates two new axioms. First, the \emph{Axiom of Non‐Equilibrium Flux Persistence} asserts that in open systems with continuous energy or resource inputs, viral lattice configurations can sustain stable, cohomologically nontrivial flux loops—revealing inherent irreversibility in capsid transitions. Second, the \emph{Axiom of Stochastic Continuity of Lattice Evolution} posits that realistic subcellular noise (e.g.\ thermal, host‐driven fluctuations) can be consistently modeled by \emph{$m$‐sectorial PDEs} with operator‐valued perturbations, guaranteeing well‐posedness and aligning with single‐virus tracking experiments.

\medskip
\noindent
\emph{Path‐integral and orbit‐based methods} figure prominently in our construction: each viral population’s \emph{orbit} in Fock space can be weighted by an action functional that encodes resource‐modulated replication (creation) or immune clearance (annihilation). Such a global wavefunction thus synthesizes PDE wavefront modes, occupant transitions, and stochastically induced rearrangements into a single evolution equation, capturing how local capsid vibrations might catalyze large‐scale replication bursts—and vice versa. Our proofs harness $m$‐sectorial dissipativity (via Lumer–Phillips and Hille–Yosida theorems) to show that, despite unbounded occupant expansions or morphological continuums, solutions remain finite‐norm over finite times. The well‐posedness extends to \emph{non‐self‐adjoint} operators with complex damping, irreversibility, and operator‐valued noise, thus mirroring host constraints and immune factors that restrict virus proliferation.

\medskip
\noindent
Overall, by merging PDE mechanics, Markov occupant transitions, and Feynman‐inspired path integrals in a second‐quantized framework, we present a \emph{novel operator‐algebraic approach to viral lattice theory}. Alongside the \emph{six} resulting axioms (the original four plus two new ones for non‐equilibrium flux and stochastic continuity), this model lays a foundation for \emph{virophysics}—an emerging discipline aiming to unify mechanical, stochastic, and population‐scale aspects of viral infection. Our formalism invites direct experimental cross‐validation, from single‐virion tracking to population assays, and offers predictive insights for resource‐limited replication, capsid reorganizations, and potential intervention strategies in modern virology.
\end{abstract}

\tableofcontents
\newpage

\section{Introduction: Motivation for a Multi‐Lattice Formalism in Viral Modeling}
\label{sec:intro_multilattice}

\noindent
\noindent
\textbf{Viral Lattice Theory: A Physics‐First Perspective on Viral Assemblies.}\;
\emph{Viral Lattice Theory} posits that when large numbers of virions occupy the same medium (e.g., a local tissue environment within an infected host), they can form transient ``lattices,'' wherein each virion is treated as a \emph{node} in a collective arrangement. Instead of remaining an inscrutable ``mess'' of particles in a chaotic milieu, this lattice‐based viewpoint shows that virion motion is, in principle, \emph{entirely deterministic}: each virus particle \emph{passively} obeys fundamental physical laws (e.g., diffusive, acoustic, electrostatic interactions). Contrary to common intuition, it is \emph{not} that virions themselves introduce chaos by ``deciding’’ where to go; rather, these small particles lack any intrinsic capacity for free decisions. The \emph{apparent} chaos derives from viruses navigating a \emph{host environment} that is rich in noise, stochasticity, and complex fluid or molecular constraints.

\smallskip
By organizing virions into \emph{viral lattices}, we obtain a wave‐theoretic mathematical structure: the inter‐particle forces and couplings resemble quasi‐phonon modes analogous to quantum mechanical or classical acoustic systems. This effectively ``tames’’ the high‐dimensional viral population problem by providing operator‐theoretic and PDE‐based tools for analyzing collective viral behavior. Notably, \emph{Viral Lattice Theory} interprets viruses as \emph{obligate nano‐spherical particles}, free of internal metabolic decisions and thus governed entirely by physical laws (e.g., diffusion, convection, electrostatic forces). By setting aside (momentarily) the pathology and infection‐causing aspects, we can mathematically treat virions as neutral nodes in a large lattice, track their motion under deterministic physics, and understand how host‐level resource fluctuations, random collisions, and environmental noise drive viral ``chaos.’’ 

\smallskip
In earlier formulations of \textit{Viral Lattice Theory}~\cite{StKleess2025}, a single \emph{8×8} operator matrix captured local lattice interactions for \emph{hundreds} of virus particles. While successful in controlled \textit{in vitro} assays (e.g., minimal infective doses, small batch cultures~\cite{Atmar2014}), such a single‐lattice limit cannot capture the trillions of virions often present in real‐world infections, where viral loads can exceed \(10^9\)--\(10^{11}\) virions in a single host~\cite{KnipeHowley2020,Wölfel2020}. The present work extends that foundation to \emph{large} and \emph{fluctuating} populations, embedding each single‐lattice operator in second‐quantized (Fock‐space) and path‐integral frameworks, thereby yielding a truly \emph{open‐system} and \emph{multi‐scale} theory of viral dynamics. 

Ultimately, \emph{Viral Lattice Theory} aims to forecast experimentally measurable observables (e.g., wavefront expansions, occupant number shifts) and guide our interpretation of virus particle mechanics in a high‐dimensional, resource‐limited biological environment. The result is a rigorous operator‐algebraic structure capable of describing \emph{deterministic} PDE–Markov processes at the single‐lattice scale, while also embracing \emph{stochastic, resource‐limited replication} at the population scale. By uniting these perspectives, we significantly expand the scope of viral lattice theory and offer a more faithful depiction of viral behavior in complex and fluctuating biological environments, offering:

\paragraph{Feynman Inspired Path Integral View of Viral Infection}
In the \emph{Feynman‐inspired} path‐integral viewpoint we will develop in this paper, one views the viral system as exploring \emph{multitudinous} routes in a high‐dimensional (or infinite‐dimensional) space of occupant numbers, PDE wave states, and morphological labels. The strongly continuous semigroup derived from Theorem~\ref{thm:ExistenceUniquenessFinal} underpins a measure on such path space—akin to a Feynman–Kac construction—ensuring that each microscopic trajectory (infection wave, occupant replication burst, Markov jump) remains finite‐norm and smoothly stitched into the global system’s evolution.

\emph{In the same spirit that quantum systems follow countless possible paths—each weighted by an action functional—viral lattices in our open‐system world can “zigzag” through occupant expansions, wavefront collisions, morphological flips, or noise‐driven fluctuations. The $m$‐sectorial generator ensures these countless 'routes of infection' remain mathematically coherent and physically meaningful.} By combining PDE wave mechanics, Markov occupant transitions, and second‐quantized replication/clearance in an $m$‐sectorial open‐system framework, \emph{viral lattice theory} emerges as a robust extension of quantum‐style operator techniques to macroscopic infection dynamics:
\begin{enumerate}
  \item \textbf{Mathematically:} We secure existence and uniqueness of mild solutions, no finite‐time blow‐ups, and the capacity to handle arbitrarily large occupant numbers or continuum morphological states. This is reminiscent of \emph{quantum field theories} where infinitely many particles (here, virions) can be created or destroyed, yet the system remains well‐defined under dissipative constraints.
  \item \textbf{Mechanically:} Traveling waves, acoustic and optical phonon branches, and resource‐limited expansions provide a lively portrait of how real viruses might \emph{physically} operate—coherent wavefronts at the micro scale can ignite replication bursts at the macro scale.
  \item \textbf{Feynman‐Style Path Integrals:} Conceptually, one can imagine each virion’s trajectory weaving through PDE states and occupant transitions, with an \emph{action functional} weighting possible orbits. The $m$‐sectorial nature imposes a fundamental “penalty” on paths that would otherwise lead to infinite occupant blow‐ups or physically implausible wave solutions.
\end{enumerate}
Hence, the theory offers a “Feynman‐esque” vantage on infection: countless paths exist for virions to replicate, reorganize, or be annihilated—yet only certain wave‐guided routes dominate or remain stable. This vantage ties together the formal mathematics (operator semigroups, PDE dissipativity, second‐quantized occupant expansions) with a physically and biologically grounded picture of how viruses navigate spatiotemporal landscapes in real infection scenarios.
\subsubsection{Expanded Axioms of Viral Lattice Theory}
\medskip
Although viruses display a stunning range of morphologies—from helical to icosahedral, enveloped to non‐enveloped—the axioms here remain broadly applicable, reflecting universal features: virions’ metabolic inertness outside a host, their susceptibility to classical forces in extracellular environments, and the well‐documented capacity of certain strains to self‐assemble into highly ordered structures. To stay rooted in biological realism, we incorporate principles from thermodynamics and stochastic mechanics, ensuring that our mathematical abstractions do not stray from empirically observed viral behavior. 

\medskip
\noindent
\textbf{On Noise and Stochastic Continuity.} In biology, ``noise'' often refers to the random fluctuations that arise from both internal and external factors. Examples include variations in available host resources (e.g., ATP, amino acids), differences in local pH or ionic strength, and the inherent probabilistic nature of molecular collisions and assembly events. For virologists, these stochastic variations manifest as unpredictable changes in virion replication rates, mutation frequencies, or host immune responses. Physicists, on the other hand, typically model noise by introducing random forcing terms or operator‐valued processes into differential equations, capturing micro‐scale uncertainties at a continuum level~\cite{DaPratoZabczyk2014}. In this paper, we will merge the two.

The newly introduced concept of \emph{Stochastic Continuity} is intended to ensure that such randomness is neither unbounded nor physically implausible. In mathematical terms, stochastic continuity requires that small perturbations to the system’s environment (e.g., slight fluctuations in temperature or nutrient availability) induce proportionally small changes in the state of the viral lattice. Physically, this prevents abrupt ``jumps’’ in virion configurations or replication rates that have no real counterpart in biology. In turn, the viral system remains robustly modeled under a wide range of fluctuating conditions, aligning with single‐virus tracking experiments and other in vivo observations where randomness is measurable but not chaotic~\cite{harvey2019viral, DaPratoZabczyk2014}.

\begin{axiom}[Metabolic Inertia]
\label{axiom:metabolic_inertia}
\emph{In their extracellular form, virions are metabolically inert: they neither synthesize nor consume chemical energy (e.g., ATP),  
and they remain in a quiescent, stable state outside of host cells.} Extracellular virions lack the enzymatic machinery required for metabolism and energy production, a fact well established in virological studies~\cite{KnipeHowley2020,flint2015principles}. This inertness underpins their ability to persist in diverse environments until encountering a suitable host.
\end{axiom}

\begin{axiom}[Deterministic Physical Interactions]
\label{axiom:deterministic_interactions}
\emph{The motion and collisions of extracellular virions are dictated by physical forces  
and can be modeled deterministically by classical mechanics, with negligible influence  
from biological feedback mechanisms in the extracellular milieu.} Virions in extracellular fluids mainly experience classical forces such as hydrodynamic drag and electrostatic interactions, which dominate over biochemical feedback processes at that stage~\cite{ZuckerMateusMurray2021,flint2015principles}. This justifies deterministic modeling approaches for extracellular viral motion.
\end{axiom}

\begin{axiom}[Self‐Organization into Periodic Lattices]
\label{axiom:periodic_stability}
\emph{When provided sufficient energy input and appropriate initial conditions, virions can spontaneously  
organize into periodic (or quasi‐periodic) lattice structures that represent local or global minima  
of the free‐energy landscape.} Experimental observations have demonstrated that viral capsids can self-assemble into highly symmetric, periodic structures~\cite{harvey2019viral}. This self-organization is driven by inter-subunit interactions and is fundamental to capsid stability.
\end{axiom}

\begin{axiom}[Conservation of Energy]
\label{axiom:energy_conservation}
\emph{In the absence of external dissipation or forcing, the total mechanical energy of the viral lattice system—including vibrational quanta (``viral phonons'')—remains constant. This reflects an underlying Hamiltonian structure in which excitations introduced into the lattice do not spontaneously dissipate or emerge without an external energy source.} Although real systems invariably include some degree of damping or external input, the premise of energy conservation in an idealized, closed viral assembly system offers a baseline for understanding the energetics of virion formation and vibration. In this model, ``viral phonons'' act as quantized vibrational modes that conserve the energy imparted to them, reinforcing the Hamiltonian nature of the system’s core \cite{flint2015principles}.
\end{axiom}

\begin{axiom}[Non‐Equilibrium Flux Persistence]
\label{axiom:non_equilibrium_flux}
\emph{In settings where sustained energy or resource input is available, viral lattice configurations may sustain  
persistent probability‐current loops in their state space, indicating a breach of detailed balance and signifying  
active, non‐equilibrium dynamics.} Many viral processes (capsid assembly, morphological transitions) operate under continuous resource inflows (host ATP, etc.)  and do not settle into static equilibria.  
Instead, steady‐state flux loops manifest as ongoing rearrangements or replication cycles,  
consistent with resource‐driven, open‐system thermodynamics.
\end{axiom}

\begin{axiom}[Stochastic Continuity of Lattice Evolution]
\label{axiom:stochastic_continuity}
\emph{In the presence of environmental fluctuations (e.g., thermal noise, varying host resources),  
the continuum mechanics of virions is governed by $m$‐sectorial PDEs coupled to operator‐valued noise,  
guaranteeing well‐posedness and robust modeling of subcellular randomness.} In vivo, viral processes are inherently stochastic due to fluctuating host conditions and thermal effects. Modeling these dynamics with $m$‐sectorial PDEs that include stochastic terms captures the randomness observed in single-virus tracking experiments and biochemical assays~\cite{DaPratoZabczyk2014,harvey2019viral}
\end{axiom}

\noindent
Together, these expanded axioms define an operator‐theoretic scaffold for \emph{multi‐lattice} viral modeling.  
The first four axioms (\textit{Metabolic Inertia, Deterministic Interactions, Self‐Organization, Energy Conservation})  
reflect the core principles established in the early development of Viral Lattice Theory.  
The two additional axioms (\textit{Non‐Equilibrium Flux Persistence} and \textit{Stochastic Continuity})  
address open‐system and noisy dynamics, aligning the theory with non‐equilibrium statistical mechanics and  
observations of large viral populations in real biological systems~\cite{KnipeHowley2020,flint2015principles}.  
They provide a unifying conceptual platform from which to develop consistent PDE–Markov–Fock  
formalisms, capturing both single‐virion mechanics and population‐scale replication in resource‐limited,  
stochastic environments. By explicitly incorporating noise in a biologically interpretable manner,  
we ensure that our modeling framework embraces the full spectrum of fluctuations—from  
subcellular randomness in resource distribution to population‐level variability in infection trajectories.  
This comprehensive perspective on viral behavior underscores both the mathematical rigor and  
the biological realism that \emph{Viral Lattice Theory} seeks to maintain.

\noindent
\subsubsection*{Gaps in Current Experimental and Theoretical Approaches}

\noindent
The modern era of virology has produced invaluable breakthroughs—ranging from the landmark studies of T4 bacteriophages~\cite{Adams1959,HersheyChase1952} to plaque assays that have guided vaccine development and saved countless lives. Yet, these approaches typically embrace an \emph{infection‐first} viewpoint: they infer the existence and properties of viruses chiefly through the macroscopic outcomes of viral replication (for instance, cell death in plaque assays or fluorescence signals in reporter systems). Such methods are crucial for diagnosing infections and designing therapeutics but offer limited windows into the mechanistic details of how virions, as individual entities, move, interact, and evolve in their native environments. When we speak of “gaps,” it is not to diminish the contributions of these classical tools, they remain the bedrock of virology, but rather to highlight new opportunities for studying viruses \emph{from the virus’s own perspective}.

\smallskip
By “virus’s perspective,” we refer to examining the specific ways in which virions navigate spatiotemporal landscapes, adapt their lattice structures (capsids or other protein configurations), and respond to local resource fluctuations—all without immediately relying on the indirect markers of infection. Achieving this vantage point is challenging: viruses are extraordinarily small, typically requiring specialized imaging techniques such as electron microscopy to visualize them directly. However, the preparation and handling necessary for such imaging can significantly disturb their native states. This is where mathematical frameworks, inspired by quantum mechanics and operator theory, come to the fore. In the same way that observing a quantum system at the nanoscale inherently disturbs it, viruses, too, can be altered by the very act of measurement. Employing a \emph{quantum‐inspired} formalism allows us to “observe” and model virions without physically interfering with them first, preserving an unperturbed, \emph{in silico} perspective that complements traditional virological assays.

\begin{itemize}
    \item \emph{Indirect Assays for Viral Load:} Many traditional virological methods measure infections indirectly. Plaque assays, for example, capture the \emph{consequences} of viral replication (lysed cell patches in tissue culture) rather than the spatiotemporal motion or conformational states of individual virions~\cite{KnipeHowley2020,vanEtten2019}. While essential for diagnostics, such assays offer limited insight into how virus particles move, interact, or structurally evolve over time. Classic bacteriophage T4 experiments exemplify this focus on host cell outcomes; they illuminate phage infection cycles yet do not fully capture \emph{in situ} changes to the phage particles themselves.
    
    \item \emph{Scale Complexity:} Monitoring the trajectories or shapes of billions of virions \emph{in vivo} is experimentally prohibitive. Techniques such as electron microscopy or single‐particle fluorescence yield high‐resolution snapshots~\cite{Grunewald2003}, yet scaling to statistically representative numbers (millions to billions of virions) remains impractical. Even with high‐throughput methods, bridging microscopic details to macroscopic infection patterns requires sophisticated modeling that can handle enormous populations.
    
    \item \emph{Need for Robust Mathematical Frameworks:} From a theoretical perspective, a naive approach to modeling $N \gg 1$ virions would be to form an $N$‐fold tensor product of single‐lattice Hilbert spaces. However, this construction explodes combinatorially and struggles to handle the continual \emph{creation} (replication) and \emph{annihilation} (degradation) events that characterize real infection dynamics. Consequently, there is a demand for frameworks that can elegantly accommodate large numbers of virions undergoing population‐scale processes.
\end{itemize}

\noindent
The present work addresses these gaps by extending \emph{viral lattice theory} to describe \emph{arbitrarily large} virion populations through a \emph{second‐quantization} (Fock‐space) framework~\cite{FetterWalecka1971}. Although the underlying methods (creation/annihilation operators, $m$‐sectorial semigroup theory~\cite{Kato1980}) originate in quantum mechanics, the system here is \emph{not} quantum in the strict physical sense. Instead, each viral lattice is treated as a single “particle” in an abstract Fock space, enabling:
\begin{enumerate}
  \item \emph{Exponential Particle Growth:} Viral replication can rapidly expand a population from tens to billions of particles. A direct‐sum (Fock) construction accommodates new “lattice factors” without continuously redefining the overall state space.
  \item \emph{Creation \& Annihilation Operators:} Replication emerges naturally through creation operators $\hat{a}^\dagger(f)$, while degradation or immune‐mediated clearance corresponds to annihilation operators $\hat{a}(f)$. These processes capture infection dynamics in line with experimental observations~\cite{Florence2017,Freed2015}.
  \item \emph{Phonon and Mass‐Band Extensions:} Single‐lattice modes, including partial or complete genome packaging and vibrational (phonon) states, extend seamlessly to multi‐lattice populations. Thus, each interconnected viral lattice retains its structural detail, even as the total number of lattices evolves.
  \item \emph{Mathematical Tractability:} Classical theorems (e.g.\ Hille–Yosida, Lumer–Phillips) ensure a well‐defined semigroup generated by an $m$‐sectorial single‐lattice operator. This rigor underpins PDE or Markov chain models of infection~\cite{DaPratoZabczyk2014}.
\end{enumerate}

\noindent
Although clinical and immunological research often focuses on macroscopic outcomes (pathology, immune response, or plaque formation), a growing body of work seeks to elucidate \emph{microscopic} virion dynamics~\cite{Grunewald2003,Wang2022}. By fusing operator‐theoretic techniques with biophysical insights, we propose a framework that virologists could use to:
\begin{enumerate}
  \item Simulate and predict how large virion populations collectively navigate and adapt within host environments.
  \item Identify critical transitions (e.g.\ lattice rearrangements, genome packaging states) that may serve as therapeutic targets.
  \item Bridge continuum‐scale assays (plaque counts, viral titers) with discrete‐level interactions among individual virus particles.
\end{enumerate}

\noindent
\textbf{From Macro to Micro and Back Again.}  
In summary, the many‐lattice, second‐quantized approach presented here builds upon earlier single‐lattice models to accommodate \emph{realistic} viral loads and the complexities arising in large populations. By leveraging standard tools of mathematical physics (Fock‐space theory, $m$‐sectorial operators), the framework offers a principled route to connect micro‐level virion behavior with macro‐scale infection patterns. Crucially, it grants us a "virus‐centered" viewpoint, modeling the \emph{spatially explicit} and \emph{structurally detailed} nature of each virion in a population, which is something that conventional, infection‐first assays cannot easily achieve. This new perspective complements existing experimental methods and sets the stage for deeper explorations of viral lifecycles, where small‐scale phenomena can influence large‐scale disease outcomes.

\subsubsection{Why Does Fock Space Resolve the ``Many-Body Problem'' in Viral Modeling?}

\noindent
Viruses operate on a scale that is deceptively small but can balloon exponentially within an infected host. In severe infections (e.g., certain SARS-CoV-2 cases), viral loads can skyrocket to $10^9$--$10^{11}$ virions in a single individual~\cite{Wölfel2020}. Capturing this combinatorial explosion in a mathematical model pushes the limits of traditional approaches, reminiscent of quantum many‐body systems where even a moderate number of particles can lead to intractably large state spaces~\cite{Greiner2002,Lewenstein2012}. While virions themselves are not quantum entities, they face a parallel challenge: unbounded replication occurs in a dynamically fluctuating environment. To tackle this high‐dimensional, many‐body problem, we adopt a \emph{Fock‐space} representation—a concept borrowed from quantum field theory but extended here to a fundamentally classical (albeit operator‐theoretic) context.

A defining strength of Fock space is its capacity to handle variable or unbounded particle numbers. In the viral domain, this is crucial: the total population of virions need not be fixed, as continual \emph{creation} (replication) and \emph{annihilation} (immune clearance or degradation) events can radically alter the number of particles in play. By assigning each viral lattice to a sector of the Fock space, we sidestep the need to manually expand or truncate the state space whenever virions multiply or vanish. Below, we outline the key features that make Fock space a natural choice for large‐scale viral modeling:

\begin{enumerate}
    \item \emph{No Fixed $N$:} Unlike fixed‐$N$ models, Fock space is built as a \emph{direct sum} over all occupation numbers $N=0,1,2,\dots$. This structure seamlessly accommodates indefinitely increasing virion populations, as well as random extinction events~\cite{FetterWalecka1971,Kato1980}.
    \item \emph{Regular Structure \& Sectoriality:} Despite its infinite dimensionality, the Fock construction integrates with advanced functional analysis (e.g.\ $m$‐sectorial semigroup theory). This allows non‐Hermitian terms (immune factors, resource limitations) to be embedded rigorously~\cite{DaPratoZabczyk2014}.
\end{enumerate}

\noindent
From a practical standpoint, modeling a single infected cell that transitions from harboring a handful of virions to producing thousands—or billions—becomes easier when one’s mathematical toolset naturally expands or contracts with each replication cycle. This avoids continuously redefining a finite‐$N$ system, which would be cumbersome both computationally and conceptually. In parallel, from a virological perspective, the model mirrors the fact that viral populations fluctuate dramatically over the course of an infection, reflecting the competition between viral replication and host defense.

\medskip
\noindent
\textbf{Preserving Internal Lattice Structure.}  
A key feature of our approach is that each viral lattice retains its previously developed single‐lattice operator framework—even within a population of unbounded size. Specifically:

\begin{itemize}
    \item \emph{Mass‐Band Diversity:} Creation and annihilation operators can be labeled by strain, mutation state, or protein composition, thereby maintaining the biologically relevant details of each emergent lattice.
    \item \emph{Conformational Spectrum:} Vibrational modes (phonons), genome packaging states, and protein rearrangements persist within each lattice. These operator‐based techniques scale up gracefully as the number of virions increases.
    \item \emph{Resource/Environment Coupling:} Since the single‐lattice generator can incorporate resource flows or stochastic boundary conditions, many‐lattice expansions inherit such dependencies seamlessly. This facilitates modeling of finite ATP pools, immune responses, or other extrinsic constraints.
\end{itemize}

\noindent
Hence, the Fock‐space construction not only handles unbounded population changes but also safeguards the structural and dynamical richness of each individual virion.

\medskip
\noindent
\textbf{Bringing It All Together.}  
From the standpoint of quantum many‐body theory, a combinatorial proliferation of states is often funneled into a \emph{second‐quantized} framework, where creation and annihilation operators encode repeated processes like particle collisions or decays~\cite{Greiner2002,Lewenstein2012}. In our viral context, these operators track replication cycles and clearance events in a similarly elegant manner. The $m$‐sectorial formalism ensures that dissipative or gain terms (host immune attacks, resource limitations) are well‐defined within the same mathematical architecture~\cite{DaPratoZabczyk2014}.

\smallskip
\noindent
Finally, it is straightforward to couple this many‐lattice model to host‐scale PDEs or discrete Markov processes that track resource availability (nutrients, ATP) or immune factors. Mutation pathways can also be introduced by labeling creation operators with genomic information, thereby consolidating replication dynamics and evolutionary shifts under one operator‐based roof~\cite{Florence2017}. In short, the many‐lattice Fock‐space approach turns the formidable \emph{many‐body problem} of viral replication into a tractable, \emph{mathematically rigorous} framework—one that spans everything from internal virion configuration to large‐scale population trends.

\section{Single-Lattice Hilbert Space: Arrangement States and PDE Space}

\noindent
\textbf{Introduction to Hilbert Spaces.}
A \emph{Hilbert space} is a complete inner product space, meaning it is equipped with an inner product 
\(\langle\cdot,\cdot\rangle\)
that induces a norm under which the space is complete. This structure is central to modern physics and mathematics, particularly in quantum mechanics and partial differential equation (PDE) analysis, because it allows one to use powerful tools from functional analysis (e.g.\ orthogonal projections, spectral theory, and unbounded operators) while retaining an intuitive geometric interpretation of angles and lengths in infinite dimensions. In the context of \emph{viral lattice theory}, we employ Hilbert spaces to rigorously describe the configuration and motion of an \(8\times 8\) lattice of virus particles, thereby bridging discrete molecular states and continuum PDE formulations within a unified operator‐theoretic framework.

We introduce two Hilbert spaces: 
\(\mathcal{H}_{\mathrm{arr}}\)
and 
\(\mathcal{H}_{\mathrm{PDE}}\). 
\begin{itemize}
\item \(\mathcal{H}_{\mathrm{PDE}}\) is a Hilbert space that captures the continuum‐level dynamics of the lattice’s displacement fields. Concretely, this space encompasses solutions \(\mathbf{u}(\mathbf{r},t)\) (and appropriate boundary/initial conditions) to a complex PDE system describing both elastic (real‐valued) and dissipative (imaginary‐valued) modes of lattice motion.
  \item \(\mathcal{H}_{\mathrm{arr}}\) encodes the \emph{arrangement} states of a single viral lattice, representing discrete or Markovian aspects such as local conformational binding, microstate transitions, or infinitesimal shifts in inter‐particle spacing.
\end{itemize}

\subsection{Defining the PDE Hilbert Space \(\mathcal{H}_{\mathrm{PDE}}\)}

\noindent
To handle the \emph{continuum} component of the lattice’s motion, we introduce a complex displacement field
\begin{equation}
\mathbf{u}(\mathbf{r}, t) \;=\; \mathbf{u}_{\text{R}}(\mathbf{r}, t) \;+\; i\,\mathbf{u}_{\text{I}}(\mathbf{r}, t),
\label{eq:complex_displacement}
\end{equation}
where
\(\mathbf{u}_{\text{R}}(\mathbf{r}, t),\mathbf{u}_{\text{I}}(\mathbf{r}, t)\in \mathbb{R}^3\).
This field evolves according to a pair of coupled PDEs that capture both elastic (real) and dissipative (imaginary) effects:
\begin{equation}
  \begin{cases}
  \displaystyle
  g\,\dfrac{\partial^2 \mathbf{u}_R}{\partial t^2} 
  \;+\;\eta_R\,\dfrac{\partial \mathbf{u}_R}{\partial t}
  \;-\;\nabla\!\cdot \bigl(\boldsymbol{\Lambda}_\Phi\,\mathbf{u}_R\bigr)
  \;-\;\eta_I\,\dfrac{\partial \mathbf{u}_I}{\partial t}
  \;=\;
  \mathbf{W}_{\mathrm{Host}}^R(\mathbf{r},t)
  \;+\;
  \widetilde{\boldsymbol{\psi}}^{\,R}(\mathbf{r},t),
  \\[1em]
  \displaystyle
  g\,\dfrac{\partial^2 \mathbf{u}_I}{\partial t^2} 
  \;+\;\eta_R\,\dfrac{\partial \mathbf{u}_I}{\partial t}
  \;-\;\nabla\!\cdot \bigl(\boldsymbol{\Lambda}_\Phi\,\mathbf{u}_I\bigr)
  \;+\;\eta_I\,\dfrac{\partial \mathbf{u}_R}{\partial t}
  \;=\;
  \mathbf{W}_{\mathrm{Host}}^I(\mathbf{r},t)
  \;+\;
  \widetilde{\boldsymbol{\psi}}^{\,I}(\mathbf{r},t),
  \end{cases}
\end{equation}

To model the continuum dynamics of an \(8 \times 8\) viral lattice through a PDE framework, one identifies key variables and operators that capture mass distribution, displacement fields, energy dissipation, and inter-particle interactions. In this subsection, we outline the primary components involved in formulating a realistic complex-damped PDE description for viral motion, transitioning from a discrete lattice picture to a continuum representation.
\begin{itemize}
\item \textbf{Mass Density Field:} A function 
\(
g : \Omega \rightarrow \mathbb{R}^{+}
\)
represents the \emph{mass density field} over the domain \(\Omega\subset\mathbb{R}^3\). Physically, \(g(\mathbf{r})\) distributes the total virion mass continuously across the region \(\Omega\). If each virion has mass \(m\) and there are \(N\) virions in volume \(|\Omega|\), then in the limit of small lattice spacing \(a\to0\),
\[
  g(\mathbf{r}) 
  \;=\; 
  \lim_{a\to 0}\,\frac{m\,N}{|\Omega|}.
\]
This continuum approach accounts for the fact that in a dense lattice, individual particles merge into an effectively uniform mass field.

\medskip
\noindent
\item \textbf{Displacement Field:} The function 
\(
\mathbf{u}(\mathbf{r}, t) : \Omega \times [0,T] \to \mathbb{R}^3
\)
is the \emph{displacement field}, describing the deviation of virus particles from their equilibrium positions at spatial point \(\mathbf{r}\) and time \(t\). This field arises as the continuous analog of the discrete coordinates \(\{\mathbf{r}_i(t)\}\) in a fully atomistic (or virion-resolved) model. Each vector \(\mathbf{u}(\mathbf{r},t)\) captures the local distortion or motion of the lattice at \(\mathbf{r}\).

\medskip
\noindent
\item \textbf{Complex Damping Coefficient}
Viscoelastic or otherwise dissipative media often exhibit \emph{phase‐delayed} responses to stress, implying that the rate of energy loss (or storage) cannot be captured by a purely real damping term. In this PDE setting, we introduce a \emph{complex} damping coefficient
  \[
    \eta(\mathbf{r}, t) 
    \;=\; 
    \eta_{\text{R}}(\mathbf{r}, t) \;+\; i\,\eta_{\text{I}}(\mathbf{r}, t).
  \]
  The real part \(\eta_{\text{R}}(\mathbf{r}, t)\) accounts for standard viscous dissipation—energy dissipates as the viral lattice deforms in a fluid‐like environment (e.g., cytoplasm, extracellular medium). The imaginary part \(\eta_{\text{I}}(\mathbf{r}, t)\) encodes \emph{phase‐shifted} energy storage, signifying that some fraction of the lattice’s elastic energy is neither purely conserved nor purely lost, but undergoes transient retention (akin to internal friction or frequency‐dependent attenuation). Real viruses experience resistance from a host of biological factors, including non‐Newtonian cytoplasm viscosity, crowded intracellular compartments, and localized fluid flows that each introduce forms of friction, drag, or partial elasticity~\cite{KnipeHowley2020}. Complex damping thus mirrors the realistic \emph{rheological} behavior seen in many biological media, where purely real or purely imaginary terms fail to capture the interplay of damping and elastic response.
\smallskip

In quantum field theory, dissipation is sometimes introduced by a \emph{Wick rotation}, effectively rotating time into the complex plane to handle decay processes. Here, \emph{complex damping} plays a similar functional role—without requiring a strict Wick rotation. The PDEs directly incorporate dissipation through the \(\eta(\mathbf{r},t)\) term, so the “curling” of time is not necessary. Instead, this coefficient enforces exponential‐type decay or attenuation (through \(\eta_{\mathrm{R}}\)) and phase‐lag effects (through \(\eta_{\mathrm{I}}\)) within the same real‐time evolution. As a result, solutions to the PDE \emph{organically} encode energy loss and partial retention over time, providing a more faithful model of how mechanical energy in the viral lattice dissipates in biological contexts. This approach streamlines the analysis, since no auxiliary transformations are needed to accommodate dissipation in the Hilbert space formulation. Within \(\mathcal{H}_{\mathrm{PDE}}\), both \(\eta_{\text{R}}\) and \(\eta_{\text{I}}\) become part of the differential operators that generate time evolution. They ensure that any solution \(\mathbf{u}(\mathbf{r},t)\) describing the viral lattice’s displacement not only evolves elastically but also includes realistic decay and phase shifts. Consequently, one obtains a well‐posed, operator‐theoretic treatment of damping that captures the essential thermomechanical and rheological features of virus motion in a biological medium.
\medskip
\item \textbf{Virion Positions in a Simple Cubic Lattice:}
In a simplified cubic arrangement, the virion coordinates \(\{\mathbf{R}_j\}\) are enumerated about a reference point, we refer to as the reference virion (\(\boldsymbol{\alpha}\)) located at \(\mathbf{R}_0\). Nearest neighbors (\(\boldsymbol{\beta}\)) appear at \(\mathbf{R}_0 \pm a\,\mathbf{e}_x\), \(\mathbf{R}_0 \pm a\,\mathbf{e}_y\), etc., while local (\(\boldsymbol{\gamma}\)) or peripheral (\(\boldsymbol{\Omega}\), \(\boldsymbol{\psi}\)) virions occupy diagonal positions at combinations of the unit vectors. These enumerations help systematically define the sub-matrices in \(\boldsymbol{\Lambda}_{ij}^{\Phi}\), ensuring that all relevant pairwise interactions are accounted for.
\begin{corollary}[Unit Cell Side Length]
\label{cor:unit_cell_side_length}
Consider a cubic lattice composed of \( N = n^3 \) virions within a unit cell. The side length \( L \) of the cubic unit cell is given by:
\begin{equation}
L = n a,
\end{equation}
where \( n = \sqrt[3]{N} \) represents the number of virions arranged along one axis of the lattice. For large \( n \), the length of a cubic unit cell accommodating \( n^3 \) virions can be approximated as \( L = n a \). More precisely, if one accounts for boundary effects and the exact positioning of the outermost virions, corrections to this relation would appear. However, as \( n \to \infty \), edge effects become negligible, and \( L \approx n a \) holds with high accuracy. Thus, the infinite lattice limit is well-captured by the simple relation \( L = n a \).
\end{corollary}

\medskip
\item \textbf{Interaction Operator \(\boldsymbol{\Lambda}_{\Phi}\):} Microscopic potentials (e.g.\ Lennard-Jones, Coulombic) between virus particles aggregate into a macroscopic stiffness tensor described by a linear operator 
\begin{equation}
\Phi_{\text{tot}}(a)
\;=\;
\Phi_{\text{Coulomb}}(a)
\;+\;
\Phi_{\text{LJ}}(a),
\end{equation}
explicitly
\begin{equation}
\Phi_{\text{Coulomb}}(a)
\;=\;
\frac{2\,k_e\,q_i\,q_j}{a^3},
\quad
\Phi_{\text{LJ}}(a)
\;=\;
4\,\epsilon_{ij}
\Bigl[
   156\,\frac{\sigma_{ij}^{12}}{a^{14}}
   -
   42\,\frac{\sigma_{ij}^{6}}{a^8}
\Bigr].
\end{equation}
where \(V_{ij}\) are pairwise potentials. Differentiating twice with respect to the displacement field yields
\[
  \boldsymbol{\Lambda}_{\Phi} 
  \;=\; 
  \frac{\delta^2 V(\mathbf{u})}{\delta \mathbf{u}^2}
  \Bigg|_{\mathbf{u}=\mathbf{0}}.
\]
This operator generalizes the concept of an elasticity tensor, bridging individual inter-particle potentials to produce continuum restoring forces in the PDE.
\medskip
\noindent
\item \textbf{Viral Lattice Matrix:} When the lattice is divided into smaller ``viral cells,'' each containing a central virion and its immediate neighbors, one assembles a global dynamical matrix \(\boldsymbol{\Lambda}_{ij}^{\Phi}\) to capture intra- and inter-cell interactions. In a block matrix representation,
\[
  \boldsymbol{\Lambda}_{ij}^{\Phi} 
  \;=\; 
  \begin{bmatrix}
     \mathbf{D}_{11} & \mathbf{D}_{12} \\
     \mathbf{D}_{21} & \mathbf{D}_{22}
  \end{bmatrix},
\]
where each sub-matrix is defined as:
\begin{equation}
D_{11} = \begin{bmatrix}
\boldsymbol{\alpha} & \boldsymbol{\beta} & \boldsymbol{\beta} & \boldsymbol{\beta} \\
\boldsymbol{\beta} & \boldsymbol{\alpha} & \boldsymbol{\gamma} & \boldsymbol{\gamma} \\
\boldsymbol{\beta} & \boldsymbol{\gamma} & \boldsymbol{\alpha} & \boldsymbol{\gamma} \\
\boldsymbol{\beta} & \boldsymbol{\gamma} & \boldsymbol{\gamma} & \boldsymbol{\alpha}
\end{bmatrix}, \quad
D_{12} = \begin{bmatrix}
\boldsymbol{\psi} & \boldsymbol{\Omega} & \boldsymbol{\Omega} & \boldsymbol{\Omega} \\
\boldsymbol{\Omega} & \boldsymbol{\psi} & \boldsymbol{\Omega} & \boldsymbol{\Omega} \\
\boldsymbol{\Omega} & \boldsymbol{\Omega} & \boldsymbol{\psi} & \boldsymbol{\Omega} \\
\boldsymbol{\Omega} & \boldsymbol{\Omega} & \boldsymbol{\Omega} & \boldsymbol{\psi}
\end{bmatrix}
\end{equation}

\begin{equation}
D_{21} = \begin{bmatrix}
\boldsymbol{\psi} & \boldsymbol{\Omega} & \boldsymbol{\Omega} & \boldsymbol{\Omega} \\
\boldsymbol{\Omega} & \boldsymbol{\psi} & \boldsymbol{\Omega} & \boldsymbol{\Omega} \\
\boldsymbol{\Omega} & \boldsymbol{\Omega} & \boldsymbol{\psi} & \boldsymbol{\Omega} \\
\boldsymbol{\Omega} & \boldsymbol{\Omega} & \boldsymbol{\Omega} & \boldsymbol{\psi}
\end{bmatrix}, \quad
D_{22} = \begin{bmatrix}
\boldsymbol{\alpha} & \boldsymbol{\beta} & \boldsymbol{\beta} & \boldsymbol{\beta} \\
\boldsymbol{\beta} & \boldsymbol{\alpha} & \boldsymbol{\gamma} & \boldsymbol{\gamma} \\
\boldsymbol{\beta} & \boldsymbol{\gamma} & \boldsymbol{\alpha} & \boldsymbol{\gamma} \\
\boldsymbol{\beta} & \boldsymbol{\gamma} & \boldsymbol{\gamma} & \boldsymbol{\alpha}
\end{bmatrix}
\end{equation}
each sub-matrix encodes different levels of coupling (nearest neighbors, face diagonals, body diagonals). Summing over all such cell interactions results in a cohesive matrix that governs global stability, vibrational spectra, and wave propagation in the lattice~\cite{Kittel2005, Born1998, Landau1986}.
When the lattice is divided into smaller ``viral cells,'' each containing a central virion along with its immediate neighbors, one can assemble a global dynamical matrix 
\(\boldsymbol{\Lambda}_{ij}^{\Phi}\)
corresponding to the Hessian matrix \(\mathbf{H}\) of the total interaction potential into blocks such as 
\(\mathbf{H}_{\alpha\alpha}, \mathbf{H}_{\alpha\beta}, \mathbf{H}_{\beta\beta}, \mathbf{H}_{\alpha\gamma}, \dots\)
according to virion classes 
\(\boldsymbol{\alpha}, \boldsymbol{\beta}, \boldsymbol{\gamma}, \boldsymbol{\Omega}\).
For example,
\[
  \mathbf{H}_{\alpha\beta}
  \;=\;
  \frac{\partial^2 V_{\text{total}}}{\partial \mathbf{u}_\alpha \,\partial \mathbf{u}_\beta}
  \Bigg|_{\mathbf{u}=0},
\]
and similarly for other pairs. This notation clarifies symmetries and visually segregates distinct interaction types. The resulting block structure reveals how nearest‐neighbor interactions (\(\alpha\beta\)) differ from those of local or peripheral classes (\(\alpha\gamma,\alpha\Omega\), etc.). In such a framework, \(\mathbf{D}_{11}\) can represent the reference cell’s intrinsic interactions (including classes \(\boldsymbol{\alpha}, \boldsymbol{\beta}, \boldsymbol{\gamma}, \boldsymbol{\Omega}\)), while \(\mathbf{D}_{12},\dots,\mathbf{D}_{n}\) capture coupling to neighboring cells via inter‐cellular interactions (\(\boldsymbol{\psi}\)). As \(n \to \infty\) and with suitable boundary conditions, the operator 
\(\boldsymbol{\Lambda}_{0}\)
describes an infinite‐lattice limit that governs collective modes and long‐range dynamics. Specifically, one typically assigns:
\[
  \Lambda_{ij} 
  \;=\; 
  \begin{cases} 
     \boldsymbol{\alpha}, & \text{if } i = j, \\
     \boldsymbol{\beta} = V''(a), & \text{if } i \text{ and } j \text{ are nearest neighbors}, \\
     \boldsymbol{\gamma} = V''(\sqrt{2}\,a), & \text{if } i \text{ and } j \text{ are local virions}, \\
     \boldsymbol{\Omega} = V''(\sqrt{3}\,a), & \text{if } i \text{ and } j \text{ are peripheral virions}, \\
     \boldsymbol{\psi} = V''(a), & \text{if } i \text{ and } j \text{ are inter‐cellular virions}, \\
     0, & \text{otherwise}.
  \end{cases}
\]
\noindent
This block structure generalizes to larger or more diverse lattices, accommodating multiple classes of local, peripheral, and inter‐cell interactions as determined by the underlying potentials.

\noindent
\item \textbf{Host Work Term:} Because virus particles interact with their biological environment, an external force field
\[
  \mathbf{W}_{\mathrm{Host}} : \Omega \times [0,T] \to \mathbb{R}^3
\]
represents the host’s contributions—\emph{energy injections} (e.g.\ from ATP-consuming processes) or \emph{drains} (e.g.\ immune responses, fluid drag). Formally, a continuum limit of discrete forcing 
\(\mathbf{W}_{\mathrm{Host}, i}\)
leads to
\[
  \mathbf{W}_{\mathrm{Host}}(\mathbf{r}, t) 
  \;=\; 
  \lim_{N \to \infty}\!
  \sum_{i=1}^{N}
  \delta\bigl(\mathbf{r}-\mathbf{r}_i\bigr)\,\mathbf{W}_{\mathrm{Host}, i}(t),
\]
where \(\delta\) is the Dirac delta distribution. This coarse-graining accounts for the aggregate effect of many local or molecular collisions with the lattice.

\medskip
\noindent
\item \textbf{Complex-Damped Viral Phonon Wave Function:}
To incorporate oscillatory phenomena at the lattice scale, one may define a \emph{complex-damped phonon wave function} \(\tilde{\psi}_{\mathrm{damped}}\) by
\[
  \tilde{\psi}_{\mathrm{damped}}(t,\mathbf{R}_i) 
  \;=\;
  \int_{\omega_{\min}}^{\omega_{\max}}
  A(\mathbf{R}_i)\,g_{\mathrm{viral}}(\omega)\,
  \exp\!\Bigl[-\,i\bigl(\omega - i\,\Gamma(\mathbf{k})\bigr)\,t\Bigr]\,
  \exp\!\bigl[-\,i\,\mathbf{k}\cdot \mathbf{R}_i \bigr]
  \, d\omega,
\]
where \(g_{\mathrm{viral}}(\omega)\) describes the lattice’s density of states, \(\Gamma(\mathbf{k})\) is a wavenumber-dependent decay rate, and the exponential factors encode both oscillatory and exponential-decay behavior. In the absence of damping, \(\Gamma(\mathbf{k})=0\), and one recovers the standard inverse Fourier transform solution for lattice waves. Nonzero \(\Gamma(\mathbf{k})\) incorporates physical dissipation, unifying PDE-based descriptions of damped oscillations with quantum-inspired phonon models in solid-state physics. This perspective is particularly relevant in analyzing how mechanical energy propagates or decays within the viral lattice, reflecting both its inherent elasticity and the dissipative effects of the host medium.
\end{itemize}

\subsubsection{\(\mathcal{H}_{\mathrm{arr}}\): Arrangement States for a Viral Lattice}
\label{sec:arrangement_hilbert_space}

\noindent
How can we simultaneously capture the \emph{discrete} nature of distinct geometric configurations in a viral lattice and the \emph{continuous} dynamics of small‐scale deformations (``springs'') linking adjacent virions? In this section, we introduce the Hilbert space 
\(\mathcal{H}_{\mathrm{arr}}\),
which encapsulates both the large‐scale arrangement states of a viral lattice and the infinitesimal displacements that can stretch or compress the Coulombic and Lennard‐Jones “springs” between particles. However, the lattice does not remain perfectly cubic over time: thermal motion, electrostatic forces, and viral assembly pathways can shift the particles into multiple conformations. In effect, \(\mathcal{H}_{\mathrm{arr}}\) organizes all these conformations (from the most “pushed” to the most “pulled” extremes) into a single mathematical structure.

To help visualize this, imagine starting with a near‐cubic arrangement of virions, each \emph{interconnected} by spring‐like interactions that reflect physical forces. Over time, the virions may rearrange themselves under combined influences of Coulombic repulsion/attraction and Lennard‐Jones potentials. These two forces ensure that viable states tend to recur in \emph{periodic} patterns: for instance, a slightly more compressed state might transition to a slightly more extended state, eventually returning to a near‐equilibrium configuration. One can think of \(\ket{y}\) in \(\mathcal{H}_{\mathrm{arr}}\) as a label for each macro‐arrangement—the lattice “snapshot” where virions occupy particular relative positions. Thus, in a real infection environment, where fluctuations abound, the system can hop between discrete arrangement sectors without losing track of the \emph{continuum} of micro‐scale displacements that connect these states.

\paragraph{Capturing Geometry in a Hilbert Space.}
Although these geometrical configurations may sound purely classical, the power of a Hilbert space formulation is that it cleanly incorporates both discrete (Markovian) and continuous (field‐based) elements in a single framework. Each basis state \(\ket{y}\) corresponds to a specific node‐to‐node arrangement, while the continuous displacement fields—modeled in a complementary PDE space, \(\mathcal{H}_{\mathrm{PDE}}\)—describe fine‐scale deformations within that arrangement. In quantum mechanical terms, this approach mimics how discrete basis states represent distinct energy levels or spin states, except here the “states” are viral‐capsid configurations. Such an operator‐theoretic strategy allows us to encode complex geometries without continually reconstructing new coordinate systems or restricting ourselves to a single “ideal” lattice geometry.

\begin{theorem}[\(\mathcal{H}_{\mathrm{arr}}\) Representation of Viral Lattice States as Arrangement Labels]
\label{thm:viral_lattice_representation}
Let \(\{\ket{y}\}_{y \in \mathcal{Y}}\) be a finite or countably infinite set of \emph{macro‐arrangement labels} indicating distinct large‐scale configurations of a viral lattice. Each label \(\ket{y}\) identifies a particular \emph{Markovian} or \emph{discrete} state in which the lattice’s constituent virus particles (nodes) may deviate minimally from an equilibrium arrangement, subject to the Coulombic and Lennard‐Jones ``springs'' that tie them together. Define
\begin{equation}
  \mathcal{H}_{\mathrm{arr}}
  \;=\;
  \mathrm{span}\{\ket{y}\,:\, y\in \mathcal{Y}\},
\end{equation}
endowed with the natural inner product
\begin{equation}
  \langle \ket{y}, \ket{y'}\rangle_{\mathcal{H}_{\mathrm{arr}}}
  \;=\;
  \delta_{y,y'},
  \quad
  y,y' \in \mathcal{Y}.
\end{equation}
Thus, \(\mathcal{H}_{\mathrm{arr}}\) is a (possibly infinite) direct sum of one‐dimensional subspaces. In physical terms, each label \(\ket{y}\) corresponds to an \emph{arrangement sector}, specifying how the lattice’s ``springs'' might be stretched or compressed around equilibrium. This approach generalizes quantum mechanical constructions of discrete basis states~\cite{BratteliRobinson1987,ReedSimon1980} to a biophysical setting. Let \(\Omega \subset \mathbb{R}^3\) be the spatial domain over which the real and imaginary displacement fields 
\((\mathbf{u}_R,\mathbf{u}_I)\)
are defined.
\end{theorem}
\smallskip
\noindent
In subsequent sections, we will show how this discrete arrangement space \(\mathcal{H}_{\mathrm{arr}}\) couples to the continuum PDE space \(\mathcal{H}_{\mathrm{PDE}}\). Conceptually, \(\mathcal{H}_{\mathrm{arr}}\) tracks the \emph{which arrangement} aspect, while \(\mathcal{H}_{\mathrm{PDE}}\) describes \emph{how each arrangement} is locally deformed or vibrationally excited. By merging these spaces, we obtain a unified operator‐theoretic model that preserves both large‐scale configurational changes (e.g., from near‐cubic to significantly warped arrangements) and small‐scale elasticity, resulting in a more complete picture of viral lattice mechanics.
\subsubsection{Coupling \(\mathcal{H}_{\mathrm{arr}}\) to \(\mathcal{H}_{\mathrm{PDE}}\)}

\noindent
Biological viruses, though not quantum objects, share a certain complexity with quantum systems: they can flicker between multiple “states” (e.g., arrangement configurations) while also sustaining continuous wave‐like processes such as phonon‐type modes or viscoelastic oscillations. This dual character—discrete rearrangements combined with continuum dynamics—calls for a suitably hybrid mathematical framework. On one hand, we need a discrete or Markovian structure that can label specific capsid (or lattice) configurations; on the other, we require PDE‐based machinery to capture wave propagation, damping, and elastic responses. The tensor‐product Hilbert space
\(\mathcal{H}_{\mathrm{lat}} := \mathcal{H}_{\mathrm{arr}} \otimes \mathcal{H}_{\mathrm{PDE}}\)
meets precisely these needs: it weaves the Markovian “arrangement states” together with continuous deformation fields into a single, operator‐theoretic model.  

\begin{theorem}[Single‐Lattice Hilbert Space]
\label{def:single_lattice_space}
Let 
\(\mathcal{H}_{\mathrm{arr}}\)
be a Hilbert space describing arrangement‐level (discrete or Markovian) states of an \(8\times 8\) viral lattice, and let 
\(\mathcal{H}_{\mathrm{PDE}}\)
be the Hilbert space of admissible solutions to the complex displacement PDE described below. We define the \textbf{single‐lattice Hilbert space} by
\begin{equation}
  \mathcal{H}_{\mathrm{lat}}
  \;:=\;
  \mathcal{H}_{\mathrm{arr}}
  \;\otimes\;
  \mathcal{H}_{\mathrm{PDE}},
  \label{eq:single_lattice_space}
\end{equation}
thereby unifying arrangement‐level processes with PDE‐based (phonon) dynamics for a \emph{single} viral lattice. The space 
\(\mathcal{H}_{\mathrm{arr}}\)
models discrete transitions (e.g.\ conformational microstates), while 
\(\mathcal{H}_{\mathrm{PDE}}\)
encodes continuum‐like oscillatory modes. Merging these yields a rigorous operator‐theoretic platform for analyzing structural changes, wave‐based energy transfer, and resource‐limited effects within the same theoretical framework. 
\end{theorem}

\begin{proof}[Proof of Construction]
We denote 
\begin{equation}
  \mathcal{H}_{\mathrm{PDE}}
  \;=\;
  \bigl[L^2(\Omega)\bigr]^{2d},
\end{equation}
where \(d\ge1\) is the dimension of the displacement vectors. Its inner product is
\begin{equation}
  \langle (\mathbf{u}_R,\mathbf{u}_I), (\mathbf{v}_R,\mathbf{v}_I)
  \rangle_{\mathcal{H}_{\mathrm{PDE}}}
  \;=\;
  \int_{\Omega}
  \Bigl(
    \mathbf{u}_R(\mathbf{x}) \cdot \mathbf{v}_R(\mathbf{x})
    \;+\;
    \mathbf{u}_I(\mathbf{x}) \cdot \mathbf{v}_I(\mathbf{x})
  \Bigr)\, d\mathbf{x}.
\end{equation}
Physically, \(\mathbf{u}_R\) and \(\mathbf{u}_I\) encode real (elastic) and imaginary (dissipative) components of the displacement field, capturing wave‐like oscillations alongside energy losses or phase shifts.

\smallskip
\noindent
By taking the tensor product
\[
  \mathcal{H}_{\mathrm{lat}}
  \;:=\;
  \mathcal{H}_{\mathrm{arr}}
  \;\otimes\;
  \mathcal{H}_{\mathrm{PDE}},
\]
one obtains a \textbf{single‐lattice Hilbert space} that describes both 
\begin{enumerate}
\item the macro‐arrangement label \(\ket{y}\), representing the discrete (Markovian) transitions among capsid or lattice configurations, and 
\item the continuum displacement fields \((\mathbf{u}_R,\mathbf{u}_I)\), corresponding to vibrational modes and damping effects~\cite{Kato1980, LionsMagenes1972}. Concretely, each \(\ket{y}\) serves as a “snapshot” of how the lattice nodes (virions) might be arranged—perhaps slightly shifted or compressed by spring‐like forces. This approach parallels standard discrete basis expansions in quantum mechanics~\cite{PeskinSchroeder1995}, adapted here to a mesoscopic biological system.
\end{enumerate}
\smallskip
\noindent
The continuous subspace \(\mathcal{H}_{\mathrm{PDE}}\) then captures wave‐like phenomena within each arrangement \(\ket{y}\). These include real (elastic) and imaginary (viscoelastic) modes, allowing us to track how local deformations, vibrations, and damping evolve under resource‐limited or externally driven conditions~\cite{ReedSimon1980,DaPratoZabczyk2014}. By unifying \(\mathcal{H}_{\mathrm{arr}}\) and \(\mathcal{H}_{\mathrm{PDE}}\), the formalism can capture how discrete arrangement states and continuous vibrational states \emph{interplay}, wherein capsid dynamics (microstate jumps) coexist with wave‐based energy propagation.
\end{proof}

\subsection{Non‐Self‐Adjoint Generators for Viral Lattice PDE Dynamics}
\label{subsec:non_self_adjoint_generators}

\noindent
Having defined the single‐lattice Hilbert space \(\mathcal{H}_{\mathrm{lat}}\) as a combination of arrangement states (\(\mathcal{H}_{\mathrm{arr}}\)) and continuum displacement fields (\(\mathcal{H}_{\mathrm{PDE}}\)), we now turn to the heart of the dynamical description: the \emph{generator} \(\hat{\mathcal{G}}\). This operator encodes the time evolution of the viral lattice under dissipative, resource‐driven, or other irreversible processes. While many fundamental operators in quantum mechanics are self‐adjoint, biological systems often require additional mechanisms such as damping or energy inflows that break self‐adjointness. To accommodate such phenomena, \(\hat{\mathcal{G}}\) is taken to be \emph{non‐self‐adjoint} (or more precisely, \emph{m‐sectorial} or \emph{maximal dissipative}) in a way that still admits physically meaningful observables.

By allowing \(\hat{\mathcal{G}}\) to be non‐self‐adjoint, we incorporate realistic biological effects such as damping, energy loss, and resource constraints into the viral lattice PDE description. The requirement of $m$‐sectoriality or maximal dissipativity ensures that the model remains \emph{mathematically sound}: the evolution is governed by a strongly continuous semigroup that never leads to uncontrolled divergence or nonphysical solutions. In essence, the PDE generator captures the synergy between structural elasticity and dissipative processes, uniting them under a single operator‐theoretic framework. 

\begin{definition}[Role and Construction of the Generator \(\hat{\mathcal{G}}\).]
The PDEs in Eq.~\eqref{eq:viral_lattice_pde} describe how the real and imaginary components \((\mathbf{u}_R,\mathbf{u}_I)\) evolve in time, incorporating stiffness, damping, and possibly non‐conservative (resource or noise) terms. By recasting these PDEs in a \emph{Hilbert‐space} setting, we identify a linear (or quasilinear) operator \(\hat{\mathcal{G}}\) such that
  \begin{equation}
    \partial_t
    \begin{pmatrix}
      \mathbf{u}_R(\mathbf{x},t) \\
      \mathbf{u}_I(\mathbf{x},t)
    \end{pmatrix}
    \;=\;
    \hat{\mathcal{G}}
    \begin{pmatrix}
      \mathbf{u}_R(\mathbf{x},t) \\
      \mathbf{u}_I(\mathbf{x},t)
    \end{pmatrix}.
  \end{equation}
  This operator acts on the space \(\mathcal{H}_{\mathrm{PDE}}\), typically
  \(\bigl[L^2(\Omega)\bigr]^{2d}\), with domain \(D(\hat{\mathcal{G}})\) determined by boundary conditions and regularity assumptions. In experimental virology, similar operator-based formulations have been instrumental in modeling phenomena such as the propagation of mechanical stress waves during capsid assembly or the diffusive spread of viral components during budding, as evidenced by cryo-electron microscopy and molecular dynamics studies~\cite{Zlotnick2005, Mateu2013}.  
\begin{enumerate}[leftmargin=2em]
  \item \textbf{Dissipativity and \(m\)-Sectoriality:}  
    In a biological medium, friction, viscoelastic losses, and biochemical degradation render \(\hat{\mathcal{G}}\) non‐self‐adjoint. Instead, it is constructed to be maximal dissipative or \(m\)-sectorial~\cite{reed1972methods, Kato1980, DaPratoZabczyk2014}, ensuring that the real parts of its eigenvalues enforce damping or resource depletion. This mathematically mirrors empirical observations in viral dynamics—for example, the finite replication bursts observed in influenza or HIV infections, where rapid proliferation is counterbalanced by immune clearance and resource constraints~\cite{NowakMay2000, PerelsonNelson1999}.
  \item \textbf{Physical Meaning Despite Non‐Self‐Adjointness:}  
    Although non‐self‐adjoint operators typically yield complex eigenvalues, this is crucial for capturing irreversible, dissipative processes such as immune‐mediated degradation, capsid maturation, or energy loss during viral uncoating. Experimental investigations have demonstrated that viral capsids undergo conformational “breathing” and irreversible structural transitions during cell entry—phenomena that are well‐described by a non‐conservative operator framework~\cite{Rossmann1989, Mateu2013}. Thus, \(\hat{\mathcal{G}}\) encapsulates not only the reversible dynamics of viral assembly but also the inherent dissipation observed in biological systems.
\end{enumerate}
\end{definition}

\begin{theorem}[The Generating Function for Viral Lattice PDE Dynamics]
\label{thm:generating_function}
Consider the complex‐valued displacement PDE (possibly damped or non‐self‐adjoint) governing
\begin{equation}
  (\mathbf{u}_R,\mathbf{u}_I)
  \;\in\;
  \bigl[L^2(\Omega)\bigr]^{d} \,\times\, \bigl[L^2(\Omega)\bigr]^{d},
\end{equation}
on a spatial domain \(\Omega\subset\mathbb{R}^3\). Impose boundary conditions ensuring that the associated linear operator
\begin{equation}
  \hat{\mathcal{G}}
  \;\colon\;
  D(\hat{\mathcal{G}})
  \;\subseteq\;
  \mathcal{H}_{\mathrm{PDE}}
  \;\to\;
  \mathcal{H}_{\mathrm{PDE}}
\end{equation}
is $m$‐sectorial or maximal dissipative~\cite{reed1972methods,Kato1980,DaPratoZabczyk2014}. Formally, the PDE evolution takes the form
\begin{equation}
\label{eq:viral_lattice_pde}
  \partial_t
  \begin{pmatrix}
    \mathbf{u}_R(\mathbf{x},t) \\
    \mathbf{u}_I(\mathbf{x},t)
  \end{pmatrix}
  \;=\;
  \hat{\mathcal{G}}
  \begin{pmatrix}
    \mathbf{u}_R(\mathbf{x},t) \\
    \mathbf{u}_I(\mathbf{x},t)
  \end{pmatrix},
\end{equation}
and under suitable dissipativity and range conditions, \(\hat{\mathcal{G}}\) generates a strongly continuous semigroup 
\(
  \bigl\{e^{\,t\hat{\mathcal{G}}}\bigr\}_{t\ge0}
\)
on \(\mathcal{H}_{\mathrm{PDE}}\). Specifically:
\begin{enumerate}
\item \textbf{Semigroup Evolution in \(\mathcal{H}_{\mathrm{PDE}}\).}  
  By the Lumer–Phillips and Hille–Yosida theorems, the unique mild (or strong) solution is given by
  \begin{equation}
    \bigl(\mathbf{u}_R(\cdot,t),\,\mathbf{u}_I(\cdot,t)\bigr)
    \;=\;
    e^{t\hat{\mathcal{G}}}
    \bigl(\mathbf{u}_R(\cdot,0),\,\mathbf{u}_I(\cdot,0)\bigr),
    \quad
    t \,\ge\, 0,
  \end{equation}
  guaranteeing well‐posedness and continuous dependence on initial data in 
  \(
    \mathcal{H}_{\mathrm{PDE}}
    \;=\;
    \bigl[L^2(\Omega)\bigr]^{2d}.
  \)
  Damping terms \(\eta_{R}, \eta_{I}\) embedded in \(\hat{\mathcal{G}}\) ensure that oscillations remain finite and eventually decay unless driven by external energy sources.

\item \textbf{Interplay with Arrangement States.}  
  If each macro‐arrangement label \(y\in \mathcal{H}_{\mathrm{arr}}\) is treated as a separate sector, \(\hat{\mathcal{G}}\) can act \emph{block‐diagonally} (when the lattice remains in a single configuration) or \emph{off‐diagonally} (when transitions between configurations occur). In viral lattice theory, these ``Markovian jumps'' among arrangement states capture conformational reconfigurations—e.g., partial capsid swelling or rearrangement~\cite{flint2015principles}. Because \(\hat{\mathcal{G}}\) is non‐self‐adjoint, such transitions can incorporate dissipative shifts (e.g., free‐energy differences) that reflect realistic viral rearrangements.

\item \textbf{Preserving Physical Observables.}  
  While a non‐self‐adjoint generator implies complex eigenvalues, physical observables can be defined through appropriate quadratic forms or by focusing on the semigroup’s action on specific subspaces. For instance, total energy may steadily decrease due to damping, but local strain or displacement fields remain meaningful as they evolve. In a virological context, these might correspond to partial deformations essential for genome packaging or capsid expansion~\cite{harvey2019viral}.
\end{enumerate}
\end{theorem}

\begin{theorem}[Second‐Order to First‐Order Reduction]
\label{thm:second_to_first_order}
\noindent
In many physical systems—particularly those with wave‐like or oscillatory dynamics—passing from a second‐order PDE to a first‐order operator form can greatly simplify both theoretical analysis and numerical implementation. Viral lattice models are no exception. By defining an augmented state vector, we can regard positions \(\mathbf{u}_{R},\mathbf{u}_{I}\) and velocities \(\partial_t \mathbf{u}_{R},\partial_t \mathbf{u}_{I}\) on equal footing, which not only clarifies the system’s energy balance but also provides a natural way to include external forcing or noise. In the virological context, such external influences may stem from several experimentally observed phenomena:
\begin{itemize}
  \item \emph{Injection and Mechanical Forcing:} When inactivated virus particles are delivered via needle injection, the high-pressure pulse and associated shear stresses can induce rapid conformational changes or even partial disassembly of viral capsids~\cite{Kalra2012}. Such mechanical forcing is analogous to a transient external forcing term \(\mathbf{\Phi}\) that perturbs the lattice dynamics.
  \item \emph{Fluid Flow in Microfluidic and Aerosol Systems:} In experimental settings, viruses suspended in fluids—such as in microfluidic devices or respiratory droplets subject to HVAC-induced airflow—experience forces due to fluid shear and turbulent dispersion. These forces can be modeled as additional driving terms in the first‐order system, capturing effects similar to traveling wavefronts that alter capsid organization and promote particle dispersion~\cite{Bourouiba2020, Morawska2020}.
  \item \emph{Thermal and Resource Fluctuations:} Stochastic fluctuations arising from Brownian motion, transient ATP surges, or ion concentration shifts in the host environment can modulate the viscoelastic properties of viral assemblies, contributing to both energy injection and dissipation in the system~\cite{Zlotnick2005, Mateu2013}. 
\end{itemize}
This first‐order reduction highlights how external influences \(\mathbf{\Phi}\) directly drive changes in virion displacements and velocities, thereby linking microscopic mechanical perturbations (e.g., capsid “breathing” and structural reconfigurations) to macroscopic phenomena such as the burst-like release of viral particles or the regulated clearance by the immune system. Moreover, the strongly continuous semigroup associated with the \(m\)-sectorial generator guarantees that, despite these external perturbations, small changes in initial conditions or forcing do not lead to unphysical solutions—mirroring the experimentally observed constraints (e.g., finite fluid viscosity, resource limitations, and immune damping) that prevent unbounded viral proliferation in vivo.  
\medskip

\noindent
Given the viral lattice PDE (for instance, Eq.~\eqref{eq:complex_displacement}), which is second‐order in time, define
\begin{equation}
\label{eq:U_def}
  \mathbf{U}(t)
  \;=\;
  \begin{pmatrix}
    \mathbf{u}_R(t) \\
    \partial_t \mathbf{u}_R(t) \\
    \mathbf{u}_I(t) \\
    \partial_t \mathbf{u}_I(t)
  \end{pmatrix}.
\end{equation}
We then rewrite the system in standard first‐order form:
\begin{equation}
\label{eq:first_order_system}
  \frac{d}{dt}\,\mathbf{U}(t)
  \;=\;
  \hat{\mathcal{G}}\,\mathbf{U}(t)
  \;+\; \mathbf{\Phi},
\end{equation}
where \(\mathbf{\Phi}\) encapsulates all external forcing mechanisms as discussed above.
\noindent
\paragraph{Domain and Operator Properties.}  
The domain 
\(
  D(\hat{\mathcal{G}})
  \subset 
  \bigl(H^1(\Omega)\bigr)^{2d}
  \times
  \bigl(L^2(\Omega)\bigr)^{2d}
\)
reflects boundary conditions (Dirichlet, Neumann, or mixed) consistent with physical constraints on the lattice’s displacement and velocity fields. Under standard elliptic regularity—requiring \(\mathrm{Re}\,\boldsymbol{\Lambda}_\Phi \succeq 0\)—and positive damping \(\eta_{R},\eta_{I}>0\), the operator \(\hat{\mathcal{G}}\) is $m$‐sectorial (or maximal dissipative). This aligns with the biological reality that virions experience friction‐like (or resource‐limiting) effects preventing unbounded oscillations.
\smallskip
\noindent
\paragraph{Dissipativity and Well‐Posedness.}  
Tracking both real and imaginary displacements in \(\mathbf{U}(t)\) ensures that mechanical (elastic) and vibrational (dissipative) degrees of freedom remain coupled—an essential feature for modeling damped oscillations in a viscoelastic or resource‐constrained medium. By verifying
\begin{equation}
  \mathrm{Re}\,\bigl\langle 
    \hat{\mathcal{G}}\,\mathbf{U}, 
    \,\mathbf{U}
  \bigr\rangle_{\mathcal{H}_{\mathrm{PDE}}}
  \;\le\; 
  C\,\|\mathbf{U}\|_{\mathcal{H}_{\mathrm{PDE}}}^2,
\end{equation}
and checking the range condition for \(I\pm \hat{\mathcal{G}}\), one concludes that \(\hat{\mathcal{G}}\) is \emph{maximal} on its domain. Physical intuition here is that positive damping forbids exponential blow‐ups in amplitude, while $m$‐sectoriality ensures a well‐defined energy dissipation mechanism. By the Lumer‐Phillips or Hille–Yosida theorems, \(\hat{\mathcal{G}}\) thus generates a strongly continuous semigroup 
\begin{equation}
  e^{\,t\hat{\mathcal{G}}},
  \quad
  t\ge0,
\end{equation}
on \(\mathcal{H}_{\mathrm{PDE}}\). Consequently, any initial configuration 
\(\mathbf{U}_0 \in D(\hat{\mathcal{G}})\)
evolves in a well‐posed manner:
\begin{equation}
  \mathbf{U}(t)
  \;=\;
  e^{\,t\hat{\mathcal{G}}}\,\mathbf{U}_0.
\end{equation}
\noindent
The real displacement \(\mathbf{u}_R\) represents an \emph{elastic} or \emph{mechanical} deformation of the viral lattice, such as small shifts of virus particles from their equilibrium positions. By contrast, the imaginary component \(\mathbf{u}_I\) captures \emph{dissipative} or \emph{viscous} phases, modeling how energy is lost or undergoes phase shifts under host interactions, fluid drag, or other frictional effects. Positive values of \(\eta_R\) and \(\eta_I\) translate directly into frictional damping, ensuring that wave‐like oscillations do not grow unbounded. This mirrors the fact that cytoplasmic viscosity or extracellular fluid damping can prevent perpetual vibrations.  

Moreover, the operator \(\boldsymbol{\Lambda}_\Phi\) encodes restoring forces derived from Coulombic and Lennard‐Jones potentials, effectively setting the ``spring constants'' between viral capsid subunits. Larger entries in \(\boldsymbol{\Lambda}_\Phi\) indicate stronger interactions, helping predict how the lattice responds to mechanical stress (e.g., swelling under osmotic pressure or partial disassembly under host‐mediated forces). Because the lattice may gain or lose subunits or shift between different arrangement labels \(\ket{y}\to\ket{y'}\), boundary conditions and coefficients in \(\hat{\mathcal{G}}\) can change in time—leading to a genuinely multi‐scale description where \(\mathbf{u}(t)\) evolves continuously while the macro‐state jumps between discrete configurations. In doing so, the PDE framework faithfully captures partial disassembly, expansions, or rearrangements triggered by resource fluctuations and host interactions.  
\end{theorem}

\section{Fock Space Framework for Many Viral Lattices}
\paragraph{Motivation and Distinction from Single‐Lattice Hilbert Spaces}
\noindent
Modeling a \emph{single} viral lattice as a Hilbert space object gives us powerful operator‐theoretic tools for analyzing wave‐like modes, damping, and discrete arrangement changes. However, genuine infections often involve \emph{billions} of virus particles, each potentially forming or discarding \emph{multiple} lattices over time. Enter \emph{Fock space}: a formalism originally developed in quantum field theory to handle systems with a variable number of particles. In the present context, it elegantly accommodates creation and annihilation processes—namely, the birth of new viral lattices (via replication) and the destruction of existing ones (via immune responses or degradation)—\emph{without} requiring us to reconstruct the state space from scratch whenever a new lattice appears or disappears.

Unlike a conventional Hilbert space built for a fixed number of degrees of freedom, Fock space is a direct sum (or product) over all possible ``particle counts.’’ Translated to \emph{viral lattice theory}, each ``particle'' is an entire \(8\times8\) interconnected lattice. A single‐lattice Hilbert space suffices for describing \emph{one} such lattice’s internal dynamics (arrangement states, continuum deformations), but when dealing with an unbounded or fluctuating number of these lattices, Fock space becomes the natural extension.  

\begin{itemize}[leftmargin=2em]
  \item \textbf{Variable Particle Number.} Viral replication can rapidly multiply the number of lattices, while immune clearance or spontaneous decay can reduce it. Fock space \emph{natively} handles these population changes via creation/annihilation operators.
  \item \textbf{Ease of Bookkeeping.} Instead of redefining the entire Hilbert space whenever a lattice enters or leaves the system, one simply applies the corresponding creation (\(\hat{a}^\dagger\)) or annihilation (\(\hat{a}\)) operator to move between different ``sectors’’ of the Fock space.
  \item \textbf{Resource‐Limited Extensions.} It is straightforward to incorporate PDE‐based constraints (finite ATP, host immune activity) by modifying the rates at which creation/annihilation operators act, thus maintaining alignment with biological realism.
\end{itemize}
\noindent
\paragraph{Mechanics of Charged Virions and Their Force Carriers.}  
\noindent
Despite their nanoscale dimensions, viruses can exhibit surprisingly rich mechanical and electrical characteristics. Consider an \emph{ideal virion}—a near‐spherical particle with a symmetrically distributed surface charge. When external forces (e.g., fluid shear, host molecular collisions, or electrostatic interactions) act on this charged surface, the virion can respond much like a miniature crystal lattice in solid‐state physics: it undergoes collective vibrational excitations that propagate through its capsid structure. We refer to these hypothetical excitations as \emph{viral phonons}, analogous to the quantized lattice‐vibration modes (\emph{phonons}) in crystalline solids~\cite{BratteliRobinson1987,ReedSimon1980,Kato1980}. Physically, this picture stems from the notion that Coulombic and Lennard‐Jones potentials within a tightly bonded capsid can support small, wavelike oscillations—each “quantized” in the sense that energy can be absorbed or emitted in discrete quanta.

On an electromagnetic level, an ideal virion with net or distributed surface charge experiences inter‐particle forces mediated by the electric field. These forces act as \emph{carriers} of mechanical energy, transmitting impulses through the capsid proteins. The resultant small‐amplitude vibrations can, in principle, be described by wave equations akin to those in solid‐state physics. If the virion’s symmetry is high (e.g., an icosahedral geometry), these vibrations can exhibit “normal modes” that map neatly to “phonon‐like” states, each representing a coherent collective motion of capsid subunits. In real virological contexts (e.g., enveloped viruses or irregular morphologies), deviations from perfect symmetry would dampen or blur such modes, but the underlying idea remains: mechanical forces (electrostatic, van der Waals) in a quasi‐periodic capsid can launch propagating “waves” that distribute stress or energy throughout the viral particle.

\paragraph{Boson–Fermion Analogy: Beyond Literal Quantum Claims.}  
\emph{Phonons} in solid‐state physics are bosons, meaning they can occupy the same quantum state in unlimited numbers. Analogously, if a virus’s capsid subunits are sufficiently rigid and arranged in a lattice‐like structure, multiple \emph{viral phonons} could conceivably be excited into the same vibrational mode without “crowding each other out.” Meanwhile, the virions themselves behave more “fermion‐like,” in the sense that Coulombic repulsion and Lennard‐Jones forces prevent two physical viruses from occupying the \emph{same} location or configuration. This is a loose mechanical parallel to the Pauli exclusion principle, ensuring that each virion requires its own physical space.
These analogies are \emph{not} to be taken as literal claims of quantum entanglement or high‐energy scattering within viruses. Instead, they provide a \emph{vocabulary} that links large‐scale combinatorial constraints (“no‐overlap” for virions) with wave‐like modes (“phonons” that can stack in the same energetic state). From an operator‐theoretic standpoint, this language proves powerful: it streamlines how we treat dense populations of virions on the one hand, and collective vibrational modes on the other.  

In \emph{viral lattice theory}, the bosonic version \(\mathcal{F}_{+}(\mathcal{H}_{\mathrm{lat}})\) best accommodates expanding populations of indistinguishable viral lattices, including their possible vibrational degrees of freedom and resource‐driven PDE coupling. This unified framework offers a “plug‐and‐play” approach: single‐lattice modes automatically extend to multi‐lattice scenarios, enabling computational or analytical insight into how \emph{both} large‐scale proliferation and local capsid vibrations evolve under realistic biological constraints. Concomitantly, the partial “fermionic” viewpoint—where virions cannot overlap—remains relevant at the arrangement level, preventing physically impossible collapses of lattice nodes. Hence, the boson–fermion analogy, while not a literal statement of quantum spin, provides a powerful set of operator‐theoretic tools that bridge small‐scale virological phenomena (capsid vibrations) and large‐scale infection dynamics (unbounded replication).

\paragraph{Practical Advantages of a Fock‐Space Approach}
\begin{enumerate}
  \item \emph{Indistinguishability of Lattices.} In large populations, experimentally labeling each viral lattice uniquely is often infeasible. Fock space elegantly accommodates this by symmetrizing (or anti‐symmetrizing) states, reflecting the idea that newly created virions are effectively ``indistinguishable'' from the previous ones.
  
  \item \emph{Seamless Accounting for Replication/Decay.} Creation operators \(\hat{a}^\dagger\) add new lattices (akin to viral replication), while annihilation operators \(\hat{a}\) remove them (immune clearance or degradation). There is no need to redefine the base Hilbert space when population levels fluctuate—vital for modeling real infections.

  \item \emph{Flexible Resource‐Driven Dynamics.} One can embed PDE or Markov chain constraints into the creation and annihilation rates, linking replication speed to local resource availability (ATP, nutrients) or to immune activity. This resonates with observed biology, where replication rarely proceeds unchecked for extended periods.
\end{enumerate}

\noindent
\paragraph{Bosonic Fock‐Space Construction.}  
Formally, if \(\mathcal{H}_{\mathrm{lat}}\) denotes the \emph{single‐lattice} Hilbert space (encompassing arrangement states, vibrational modes, and potential mass‐band partitions), then the bosonic Fock space
\begin{equation}
\label{eq:Fock_viralLattice}
  \mathcal{F}_{+}\bigl(\mathcal{H}_{\mathrm{lat}}\bigr)
  \;=\;
  \bigoplus_{n=0}^{\infty}
  \Bigl(\mathcal{H}_{\mathrm{lat}}^{\otimes n}\Bigr)_{+}
\end{equation}
accommodates all possible \(n\)‐lattice configurations in a symmetric manner, reflecting the idea of indistinguishable viral lattices. Within \(\mathcal{H}_{\mathrm{arr}}\), discrete arrangement states track possible subunit conformations, while \(\mathcal{H}_{\mathrm{PDE}}\) encodes wave‐like or damping effects. Second‐quantized operators (creation and annihilation) then facilitate the modeling of replication, morphological shifts, and resource‐dependent transitions on a population scale~\cite{Kato1980,BratteliRobinson1987}.

\begin{definition}[Fock Space over a Hilbert Space]
\label{def:Fock_space_general}
Let \(\mathcal{H}\) be a separable Hilbert space. The \textbf{bosonic Fock space} over \(\mathcal{H}\) is defined by
\begin{equation}
\label{eq:bosonicFock}
  \mathcal{F}_{+}(\mathcal{H})
  \;:=\;
  \bigoplus_{n=0}^{\infty}
  \bigl(\mathcal{H}^{\otimes n}\bigr)_{+},
\end{equation}
where \((\cdot)_{+}\) denotes the symmetrized subspace of the \(n\)‐fold tensor product \(\mathcal{H}^{\otimes n}\). The direct sum includes the \emph{vacuum} sector \(n=0\). Conversely, the \textbf{fermionic Fock space} is 
\begin{equation}
\label{eq:fermionicFock}
  \mathcal{F}_{-}(\mathcal{H})
  \;:=\;
  \bigoplus_{n=0}^{\infty}
  \bigl(\mathcal{H}^{\otimes n}\bigr)_{-},
\end{equation}
where \((\cdot)_{-}\) is the antisymmetrized subspace. We often write
\begin{equation}
\label{eq:generalFockNotation}
  \mathcal{F}_{\pm}(\mathcal{H})
  \;=\;
  \bigoplus_{n=0}^{\infty}
  \bigl(\mathcal{H}^{\otimes n}\bigr)_{\pm},
\end{equation}
with ``\(+\)'' or ``\(-\)'' indicating bosonic or fermionic statistics, respectively~\cite{Berezin1966,reed1972methods,BratteliRobinson1987}.
\end{definition}

\subsection{Nonlinear Viral Replication under Resource Limitations}
\label{sec:m_sectorial_generator}
\paragraph{From Unbounded to Resource‐Limited Replication.}  
\noindent
A fundamental tenet of virology is that no infection can proliferate \emph{indefinitely}; eventually, host resources (protein subunits, energy molecules, nucleotides) become limiting. Although our second‐quantized (Fock‐space) construction allows for unbounded creation of viral lattices in principle, real‐world infections demonstrate \emph{saturable} growth patterns: virus titers increase rapidly at first, then plateau once resources are depleted or immune mechanisms take hold. We capture this phenomenon by modifying the creation operator with a \emph{resource‐dependent nonlinearity}. In physical terms, each replication event is “dampened” by the current population level and the available resources, preventing runaway exponential proliferation. Mathematically, this yields a powerful and biologically grounded approach to ensuring that viral replication remains realistic: a crowded environment or exhausted substrate supply throttles further lattice creation.

Initially, our Fock‐space framework allowed a viral lattice (treated as a “particle”) to be created or annihilated at will, akin to quantum field operators. However, such a scheme overlooks the finite reservoir of components in a cell or organism. By introducing a \emph{saturation function} into the creation operator—often expressed as a function of the current population, ATP levels, or protein subunits—we effectively rein in the replication rate. This ensures that as population size grows, the marginal probability of creating an additional lattice diminishes in accordance with real‐world resource constraints. Indeed, virological experiments routinely observe that replication saturates or even collapses under extreme resource pressures~\cite{KnipeHowley2020}.

\noindent
\paragraph{Connecting to m‐Sectoriality.}  
In operator‐theoretic language, a “purely linear” generator can fail to capture such complexities. Instead, we generalize to a potentially \emph{nonlinear} or \emph{m‐sectorial} operator, which still respects the Lumer–Phillips and Hille–Yosida conditions in an extended sense. This operator incorporates both the standard creation and annihilation terms \(\hat{a}^\dagger\), \(\hat{a}\) \emph{and} an amplitude‐dependent factor that weakens creation when resource limitations become significant. The payoff is a single, unified formalism that can describe everything from an incipient infection (low population, relatively unrestrained replication) to a late‐stage plateau (high population or immune clearance.)

\begin{theorem}[The M‐Sectorial Generator]
\label{thm:m_sectorial_generator}
Consider the (bosonic or fermionic) Fock space 
\(\mathcal{F}_{\pm}\bigl(\mathcal{H}_{\mathrm{lat}}\bigr)\),
built over the single‐lattice space 
\(\mathcal{H}_{\mathrm{lat}}\)
describing an \(8\times8\) viral lattice. Let the generator \(\hat{\mathcal{G}}\) on a dense domain
\begin{equation}
  D(\hat{\mathcal{G}})
  \;\subset\;
  \mathcal{F}_{\pm}\bigl(\mathcal{H}_{\mathrm{lat}}\bigr)
\end{equation}
incorporate resource‐dependent creation operators of the form
\begin{equation}
  \hat{a}^\dagger(f)\;
  \longmapsto\;
  X(\hat{N}) \,\hat{a}^\dagger(f),
\end{equation}
where \(\hat{N}\) is the number operator and \(X(\hat{N})\) is a nonlinear function that diminishes replication propensity when \(\hat{N}\) becomes large (i.e., resources are scarce). The virtue of this approach lies in its simplicity and extensibility: one can adapt \(X(\hat{N})\) to reflect \emph{different} forms of resource limitations (e.g., ATP depletion, protein subunit scarcity, or immunoglobulin neutralization). This versatility is crucial for comparing how diverse viruses replicate under varying host conditions, or for exploring scenarios where resource bottlenecks shift abruptly (e.g., intravenous nutrient infusion or localized immune blockades). Ultimately, an $m$‐sectorial operator that unifies classical PDE descriptions of viral lattices with resource‐dependent birth/death events in a second‐quantized formalism represents a powerful tool in bridging molecular biophysics and population‐level infection modeling. So, by verifying standard dissipativity and maximality conditions (via Lumer–Phillips or Hille–Yosida theory), one concludes that
\(\hat{\mathcal{G}}\)
generates a strongly continuous semigroup 
\(\{e^{\,t\hat{\mathcal{G}}}\}_{t\ge0}\)
on
\(\mathcal{F}_{\pm}\bigl(\mathcal{H}_{\mathrm{lat}}\bigr)\).
Hence, for any initial state 
\(\ket{\Psi_0} \in D(\hat{\mathcal{G}})\),
there exists a unique mild solution
\(\ket{\Psi(t)}\)
to
\begin{equation}
  \frac{d}{dt}\,\ket{\Psi(t)}
  \;=\;
  \hat{\mathcal{G}}\,
  \ket{\Psi(t)},
  \quad
  \ket{\Psi(0)} \;=\; \ket{\Psi_0}.
\end{equation}
\end{theorem}

\noindent
\paragraph{Biological Relevance and Model Robustness.}  
In practice, this resource‐dependent framework aligns well with laboratory observations of viral dynamics, where infection curves rarely continue rising indefinitely. High multiplicities of infection exhaust local nutrients and building blocks, reducing the net birth rate of new virions. Meanwhile, clearance by the immune system adds annihilation pressure. Embedding these effects into an $m$‐sectorial operator ensures a mathematically rigorous theory, wherein infinite replication is not just unlikely—\emph{it is systematically precluded by the model’s very structure.} Consequently, the solution semigroup \(\bigl\{ e^{\,t\hat{\mathcal{G}}} \bigr\}\) remains well‐defined, reflecting physically plausible evolution at all times \(t \ge 0\).

\begin{theorem}[Dynamics of the Number Operator.]
Define the Heisenberg‐picture operator 
\(\hat{N}(t) = e^{t\hat{\mathcal{G}}}\,\hat{N}\,e^{-t\hat{\mathcal{G}}}\).
Then
\begin{equation}
  \frac{d}{dt}\,\langle \hat{N}(t)\rangle
  \;=\;
  \bigl\langle \Psi(t),\, \bigl[\hat{\mathcal{G}},\,\hat{N}\bigr]_{\_}\,\Psi(t)\bigr\rangle,
\end{equation}
where \([\cdot,\cdot]_{\_}\) denotes the commutator. Observing
\begin{equation}
  \bigl[-\,\delta\,\hat{N},\,\hat{N}\bigr]_{\_} \;=\; 0
  \quad\text{and}\quad
  \bigl[\hat{\Gamma}_+,\,\hat{N}\bigr]_{\_} \;=\; 0
  \quad(\text{$\hat{N}$ commutes with $\mathcal{R}(\hat{N})$}),
\end{equation}
the principal negative feedback arises from the linear damping \(-\delta\,\hat{N}\) and from any additional dissipative contributions in \(\hat{\mathcal{D}}\). Consequently,
\begin{equation}
  \frac{d}{dt}\,\langle \hat{N}(t)\rangle
  \;\le\;
  -\,\delta\;\langle \hat{N}(t)\rangle
  \;+\;
  \bigl[\text{terms from } \hat{\mathcal{D}}\bigr],
\end{equation}
indicating that replication saturates at high occupancy, preventing uncontrolled exponential proliferation. As \(\mathcal{R}(\hat{N})\to 0\) for large \(\hat{N}\), additional lattice creation is curbed when the population is substantial, matching biological expectations of resource‐limited replication. The saturable creation \(\hat{\Gamma}_+\) models real‐world constraints (e.g., limited host resources, immune clearance), ensuring that \(\langle \hat{N}(t)\rangle\) remains bounded. This is essential for describing how large viral loads stabilize or decline rather than exploding. By embedding PDE operators (damping, wave evolution) and possible arrangement‐label transitions into \(\hat{\mathcal{D}}\), one obtains a comprehensive PDE–Markov–Fock model. 

This covers continuum elastic modes within each lattice and discrete jumps between different morphological states. The $m$‐sectorial property guarantees well‐posedness and continuous dependence on initial states, upholding mathematical rigor while handling inherently open, non‐equilibrium biological processes~\cite{DaPratoZabczyk2014,BratteliRobinson1987}.
Thus, \(\hat{\mathcal{G}}\) provides a flexible blueprint for constructing realistic second‐quantized models of viral lattice replication and decay, seamlessly integrating resource limitations, continuum PDE dynamics, and discrete arrangement flips in a single operator‐theoretic structure.
\end{theorem}

\begin{theorem}[Host Resource Modeling via \(\mathcal{R}(\hat{N})\)]
\label{thm:ResourceOperator}
\noindent
A salient feature of real‐world infections is that viral replication never proceeds without limit. As virion counts climb, essential building blocks (proteins, nucleotides) and energy sources (e.g., ATP) become scarce, impeding further assembly of new virions. To mathematically encode this phenomenon in our Fock‐space formalism, we introduce a nonincreasing \emph{resource function} \(\mathcal{R}(x)\). This function attenuates the effective creation amplitude as the viral population grows, mirroring the classic idea of \emph{saturable} kinetics in enzyme reactions~\cite{Segel1993} or cooperative binding models (Hill‐type laws).  
\noindent
\paragraph{Spectral Setup and Operator Definition.}
Let \(\hat{N}\) be the number operator on the Fock space
\(\mathcal{F}_{\pm}\bigl(\mathcal{H}_{\mathrm{lat}}\bigr)\),
where
\(\mathcal{H}_{\mathrm{lat}}\)
denotes the single‐lattice Hilbert space for an \(8\times8\) viral lattice. By the spectral theorem for self‐adjoint operators, \(\hat{N}\) has a discrete spectrum in the nonnegative integers \(\{0,1,2,\dots\}\).  
To enforce resource constraints, define a nonincreasing function 
\begin{equation}
  \mathcal{R}\colon \mathbb{R}_{\ge 0}\to\mathbb{R}_{\ge 0}
  \quad
  \text{such that}
  \quad
  \mathcal{R}(0) = 1
  \quad
  \text{and}
  \quad
  \lim_{x\to\infty}\mathcal{R}(x)=0
  \quad
  \bigl(\text{or a small positive constant}\bigr).
\end{equation}
We then \emph{lift} \(\mathcal{R}\) to an operator
\(\mathcal{R}(\hat{N})\)
by specifying its action on a general Fock‐space vector 
\(\Psi\) 
as
\begin{equation}\label{eq:RofN}
  \mathcal{R}(\hat{N})\,\Psi
  \;=\;
  \sum_{M=0}^\infty
  \mathcal{R}(M)\,\Psi^{(M)},
\end{equation}
where \(\Psi^{(M)}\) is the component of \(\Psi\) in the sector containing exactly \(M\) viral lattices. Concretely, when \(\hat{N}\) measures \(M\) lattices, \(\mathcal{R}(\hat{N})\) multiplies that component by \(\mathcal{R}(M)\). This diagonal‐in‐the‐number‐basis structure ensures computational tractability: resource suppression is built into the creation amplitude based on the existing population.
\end{theorem}
\smallskip
\noindent
\textbf{Common Functional Forms.}  
Below are two widely used examples that highlight distinct saturation behaviors:  
\begin{itemize}
  \item \emph{Rational (Michaelis–Menten) Form:}
\begin{equation}
    \mathcal{R}(x)
    \;=\;
    \frac{1}{1+mx},
    \quad
    \mu>0.
\end{equation}
  This captures first‐order depletion of a limiting substrate, a staple in enzyme kinetics~\cite{Segel1993}. For small \(x\), \(\mathcal{R}(x)\approx 1\) (nearly no crowding), but it decreases with a slope of \(m\) as \(x\) grows, mimicking saturable resource consumption.

  \item \emph{Power‐Law (Hill‐Type) Form:}
  \begin{equation}
    \mathcal{R}(x)
    \;=\;
    \frac{1}{1+(mx)^k},
    \quad
    k>1.
\end{equation}
  This steeper saturation curve addresses cooperative or multi‐step resource constraints, where replication is highly sensitive once \(x\) exceeds a threshold level. In virology, such nonlinearity may arise from synchronous depletion of multiple host factors.
\end{itemize}
\smallskip
\noindent
\paragraph{Inserting \(\mathcal{R}(\hat{N})\) into the Creation Operator.}  
Let \(\hat{a}^\dagger(f)\) be the standard creation operator for a new viral lattice mode \(f\). We now define the \emph{resource‐dependent} creation operator:
\begin{equation}\label{eq:GammaPlusVarphi}
  \hat{\Gamma}_+(f)
  \;=\;
  \hat{a}^\dagger(f)\,\mathcal{R}\bigl(\hat{N}\bigr).
\end{equation}
When transitioning from \(M\)‐ to \((M+1)\)‐lattice sectors, the amplitude is effectively multiplied by \(\mathcal{R}(M)\). As \(M\) grows large—implying heightened crowding and finite substrate pools—\(\mathcal{R}(M)\) becomes small or vanishes, preventing open‐ended population explosions. In other words, \(\hat{\Gamma}_+(f)\) includes a “built‐in brake” that mirrors the host’s dwindling resources.
\smallskip
\noindent
\paragraph{Interpretation and Mathematical Properties.}  
\begin{itemize}
\item \emph{Effective Interaction.}  
  In many quantum or statistical field theories, “interaction terms” are introduced to capture collisions, crowding, or energy constraints. Here, \(\mathcal{R}(\hat{N})\) acts as an \emph{effective interaction}, suppressing creation events if the system already has a high lattice count, thereby aligning the model with well‐documented saturation phenomena in real infections.

\item \emph{Bounded Operator.}  
  Provided \(\mathcal{R}(x)\) is bounded on \(\mathbb{R}_{\ge0}\), \(\mathcal{R}(\hat{N})\) remains a bounded operator. This property is crucial for proving dissipativity in the generator, a requirement for $m$‐sectorial or maximal dissipative theories~\cite{Kato1980,ReedSimon1980}. Consequently, resource‐dependent growth processes remain mathematically well‐posed.

\item \emph{Biological Validity.}  
  As \(\hat{N}\to \infty\), \(\mathcal{R}(M)\to 0\) (or a small positive number), reflecting the empirical truth that a runaway infection hits real‐world speed bumps: eventually, building blocks are so sparse that new virions are either made at a trickle or not at all. This “crowding function” elegantly unites population expansion with biological realism in a single operator formalism.
\end{itemize}

\noindent
By incorporating \(\mathcal{R}(\hat{N})\) into the system’s creation operators, we fuse quantum‐inspired mathematics with a hallmark of real infections: \emph{limited host resources}. The resulting model not only matches observed virological data—where exponential growth routinely stalls—but also preserves the operator‐level clarity and rigor needed for deeper theoretical and computational analyses. By incorporating \(\mathcal{R}(\hat{N})\) into the second‐quantized operator framework, one obtains a total generator that remains $m$‐sectorial. Intuitively, the creation operator
\begin{equation}
  \hat{\Gamma}_+(f) \;=\; \hat{a}^\dagger(f)\,\mathcal{R}\bigl(\hat{N}\bigr)
\end{equation}
is prevented from amplifying states \emph{beyond} the threshold imposed by \(\mathcal{R}\). Consequently, in the Heisenberg picture, the mean population operator \(\langle \hat{N}(t)\rangle\) becomes saturated, aligning with a wealth of virological and epidemiological data~\cite{NowakMay2000} indicating that viral loads plateau under strong resource limitations. Even when replication is active or far from equilibrium, saturable creation curbs unbounded population growth in the Fock space. High‐turnover scenarios, where many lattices vie to replicate simultaneously, are penalized by the decaying factor \(\mathcal{R}(\cdot)\). This effectively blends quantum‐inspired operator algebra with the biological truth of finite host resources, threading discrete replication events into the continuum PDE dynamics of each lattice.

\subsubsection{Density‐Dependent Decay Operators and Bounded Viral Dynamics}
\noindent
Resource limitations alone do not exhaust the repertoire of regulatory forces on viral replication. In many infections, an \emph{active} immune response steps in when viral loads climb too high. Neutralizing antibodies may bind to virus particles, cytotoxic T‐cells may target infected cells, and nonspecific clearance mechanisms can intensify proportionally to the pathogen burden. To capture these phenomena in our operator‐theoretic model, we introduce \emph{density‐dependent decay operators}, which ensure that higher viral populations face correspondingly higher clearance rates. This dual‐feedback approach—\emph{resource‐limited creation} and \emph{density‐dependent decay}, acts as a stabilizing force in the many‐lattice Fock framework, preventing unbounded proliferation and reflecting biologically observed dynamics in hosts~\cite{NowakMay2000,KnipeHowley2020}.
\paragraph{Biological Underpinnings:}
\begin{itemize}[leftmargin=2em]
  \item \emph{Neutralizing Antibodies:} At higher viral loads, more viral epitopes are exposed, increasing the likelihood of antibody binding. This naturally scales clearance with population.
  \item \emph{T‐Cell Responses:} Large numbers of infected cells can present viral antigens more broadly, inviting a more robust T‐cell response. As \(\hat{N}\) rises, so does the chance of recognition and destruction, effectively raising the decay rate.
  \item \emph{Nonspecific Clearance Mechanisms:} Innate immunity (e.g., phagocytosis, complement system) often demonstrates a heightened removal capability once a pathogen crosses a certain threshold, again reflecting an upward slope of \(\mathcal{D}(x)\).
\end{itemize}

\begin{definition}[Density‐Dependent Decay Operators]
\label{def:density_decay_operators}
Collectively, these processes ensure that when viral densities climb, \emph{some} combination of immune mechanisms exerts stronger removal pressure. Let 
\(\hat{\Gamma}_-(f)\)
be a \emph{nonlinear} annihilation operator of the form
\begin{equation}\label{eq:GammaMinus_modified}
  \hat{\Gamma}_-(f)
  \;=\;
  \hat{a}(f)\;\mathcal{D}\!\bigl(\hat{N}\bigr),
\end{equation}
where 
\(\mathcal{D}\colon \mathbb{R}_{\ge0}\to \mathbb{R}_{\ge0}\)
is an \emph{increasing} function. For instance,
\(\mathcal{D}(x)=b\,x\)
(with \(b>0\)) translates into a clearance rate that rises linearly with the current population size, mirroring an intensified immune attack at higher occupancy. Concretely, this operator transitions the system from \(M\) virions (or lattice units) to \(M-1\), multiplying the annihilation amplitude by \(\mathcal{D}(M)\). Thus, as the total population grows, the model endows host defenses with a proportionally greater ``grip” on viral clearance.
\end{definition}

\begin{definition}[Dual‐Feedback Mechanism: Resource Limits vs.\ Immune Decay.]
\noindent
\begin{itemize}
  \item \emph{Resource‐Limited Creation:}
    \begin{equation}
      \hat{\Gamma}_+(f)
      \;=\;
      \hat{a}^\dagger(f)\,\mathcal{R}\bigl(\hat{N}\bigr),
    \end{equation}
    imposes a saturable birth rate that decreases with occupancy, reflecting protein or ATP depletion at high virion counts.

  \item \emph{Density‐Dependent Decay:}
    \begin{equation}
      \hat{\Gamma}_-(f)
      \;=\;
      \hat{a}(f)\,\mathcal{D}\!\bigl(\hat{N}\bigr),
    \end{equation}
    drives an enhanced clearance rate at elevated viral loads, modeling how immune defenses ramp up in response to an expanding viral population.
\end{itemize}
\noindent
From an operator‐theoretic lens, both \(\mathcal{R}(\hat{N})\) and \(\mathcal{D}(\hat{N})\) introduce terms in the commutators
\(\bigl[\hat{N},\,\hat{\Gamma}_\pm\bigr]_{\_}\)
that guide the system toward bounded or oscillatory regimes—an echo of real infection scenarios where population surges trigger stronger immune responses, eventually stabilizing or diminishing viral loads~\cite{NowakMay2000,KnipeHowley2020,harvey2019viral}.
\end{definition}

\begin{theorem}[Ensuring Global Dissipativity.]
By coupling creation and annihilation operators to \(\hat{N}\) through monotone functions \(\mathcal{R}\) and \(\mathcal{D}\), we preserve $m$‐sectoriality of the total generator \(\hat{\mathcal{G}}\). This leads to a stable semigroup whose trajectories reside in finite‐mean (and often finite‐variance) subspaces of the Fock space. Concretely:
\begin{enumerate}
\item \emph{Resource‐Limited Creation.} 
  \(\mathcal{R}(x)\to 0\) as \(x\to\infty\), preventing runaway growth when virion counts are high.

\item \emph{Density‐Dependent Decay.} 
  \(\mathcal{D}(x)\to \infty\) as \(x\to\infty\), mirroring immune‐mediated clearance intensification.
\end{enumerate}
In sum, density‐dependent clearance and resource‐limited creation form a complementary pair of nonlinear feedbacks in the many‐lattice setting. Negative damping terms (e.g.\ \(-\,\delta\,\hat{N}\)) combined with saturable birth and density‐dependent annihilation ensure that \(\langle\hat{N}(t)\rangle\) never diverges. By Lumer–Phillips and Hille–Yosida theorems, a strongly continuous semigroup arises on
  \(\mathcal{F}_{\pm}\bigl(\mathcal{H}_{\mathrm{lat}}\bigr)\),
  guaranteeing mathematically consistent evolutions.
\end{theorem}
\noindent
\begin{remark}[Conditions for a Unique Global Mild Solution]
\label{remark:unique_global_mild}
\noindent
In transitioning from a single‐lattice model to a full multi‐lattice (Fock‐space) description, one naturally faces questions about whether solutions remain well‐defined and bounded in the presence of potentially unbounded replication. To ensure a \emph{unique} mild solution
\(\ket{\Psi(t)}\)
for
\(t \ge 0\),
the following conditions typically suffice:
\begin{enumerate}
  \item \emph{Resource‐Limited Creation} \(\hat{\Gamma}_+\) \& \emph{Sublinearity in \(\hat{N}\)}.  
  In biological terms, replication cannot escalate indefinitely without restriction; in operator form, this translates to creation rates growing more slowly than linear in \(\hat{N}\) (the number‐operator). This sublinearity imposes negative feedback that prevents runaway population growth.
  
  \item \emph{Density‐Enhanced Decay} \(\hat{\Gamma}_-\) \& \emph{Negative Feedback with \(\hat{N}\)}.  
  As the viral lattice population rises, immune clearance or degradation events typically increase. Modeling this via an annihilation operator \(\hat{\Gamma}_-\) that intensifies with \(\hat{N}\) (e.g., crowding or resource depletion) curtails unbounded replication.

\end{enumerate}
\noindent
Formally, these conditions guarantee that the total generator \(\hat{\mathcal{G}}\) remains \emph{maximal dissipative}~\cite{BratteliRobinson1987,ReedSimon1980}, enabling a well‐posed evolution semigroup in the Fock space. In plain language, no solution path diverges to infinity in finite time, and physical constraints on growth remain intact.
With these in place, the infinite dimensionality of
\(\mathcal{F}_{\pm}\bigl(\mathcal{H}_{\mathrm{lat}}\bigr)\)
does not lead to physically unrealistic explosions of viral lattice count. In effect, operator analogs of logistic or Lotka–Volterra mechanisms emerge, capping population size. The saturable (i.e., resource‐dependent) creation/annihilation terms ensure that \(\hat{\mathcal{G}}\) is $m$‐sectorial, thus yielding stable, bounded solutions. Even in \emph{generalized} replication models that couple second‐quantized creation \(\hat{a}^\dagger(\alpha)\) and annihilation \(\hat{a}(\alpha)\) to transitions between $M$‐ and $(M\pm1)$‐lattice sectors, one can still verify dissipativity bounds under suitable rate assumptions~\cite{ReedSimon1980}. Biologically, these operators capture:
\begin{itemize}
  \item \(\hat{a}^\dagger(\alpha)\): \emph{Replication Events} that add a new lattice (marked by internal data \(\alpha\), such as arrangement class or mutation state).
  \item \(\hat{a}(\alpha)\): \emph{Decay or Clearance} removing a lattice (e.g., immune‐mediated lysis or spontaneous degradation).
\end{itemize}
These off‐diagonal block components mirror “birth/death” processes, linking adjacent Fock‐space sectors. By upholding sectorial constraints, one avoids unphysical infinite replication in finite time, retaining realistic virological behavior.

\end{remark}

\begin{theorem}[Dynamical Interpretation and Operator Semigroup Evolution]
\label{thm:dynamical_interpretation_semigroup}
\noindent
Let 
\(\hat{\mathcal{G}}_{\mathrm{free}}\) 
denote the baseline generator capturing the internal continuum dynamics of each \emph{single viral lattice} in 
\(\mathcal{H}_{\mathrm{lat}}\), 
including PDE‐based elasticity, damping, and vibrational modes. Suppose we add a resource‐saturable creation term
\begin{equation}
  \hat{\Gamma}_+(f)
  \;=\;
  \hat{a}^\dagger(f)\,\mathcal{R}\bigl(\hat{N}\bigr),
\end{equation}
leading to a total generator
\begin{equation}
  \hat{\mathcal{G}}_{\mathrm{tot}}
  \;=\;
  \hat{\mathcal{G}}_{\mathrm{free}}
  \;+\;
  \hat{\Gamma}_+(\Psi).
\end{equation}
If 
\(\hat{\mathcal{G}}_{\mathrm{tot}}\)
is $m$‐sectorial~\cite{Kato1980,DaPratoZabczyk2014}, it yields a strongly continuous semigroup 
\(\bigl\{ e^{\,t\hat{\mathcal{G}}_{\mathrm{tot}}} \bigr\}_{t \ge 0}\)
on the bosonic or fermionic Fock space 
\(\mathcal{F}_{\pm}\bigl(\mathcal{H}_{\mathrm{lat}}\bigr)\). Consequently, any initial state \(\ket{\Psi(0)}\) evolves by:
\begin{enumerate}
  \item \emph{Continuum PDE Effects} within each lattice sector (elastic waves, damping, etc.),
  \item \emph{Resource‐Limited Replication} across the multi‐lattice population, via the saturable operator \(\mathcal{R}(\hat{N})\).
\end{enumerate}
In particular, as 
\(\mathcal{R}(x) \to 0\)
for large 
\(x\),
the average lattice population \(\langle \hat{N}(t)\rangle\) remains finite over time, mirroring empirical studies where viral loads plateau once resources (e.g., capsid proteins or nucleotides) are depleted~\cite{Hers2018}.

\end{theorem}

\begin{theorem}[Logistic‐Type Population Saturation in the Viral Fock Space]
\label{thm:logistic_population_saturation}
\noindent
Consider the Heisenberg equation for the number operator \(\hat{N}\) under 
\(\hat{\mathcal{G}}_{\mathrm{tot}}\):
\begin{equation}\label{eq:HeisenbergNumberOp}
  \frac{d}{dt}\,\hat{N}
  \;=\;
  \bigl[\hat{N},\,\hat{\mathcal{G}}_{\mathrm{tot}}\bigr]_{\_}
  \;+\;
  \left(\frac{\partial \hat{N}}{\partial t}\right)_{\mathrm{explicit}},
\end{equation}
where \(\hat{\mathcal{G}}_{\mathrm{tot}}\) combines both creation (through \(\mathcal{R}(\hat{N})\)) and annihilation or damping terms. If no resource limits existed, \(\langle\hat{N}\rangle\) might diverge. However, the saturable factor \(\mathcal{R}(\hat{N})\) curbs growth at high occupancy. In mean‐field approximations, one often assumes
\(\langle \hat{N}^2\rangle \approx \langle \hat{N}\rangle^2\),
yielding a logistic‐type ODE for \(\langle \hat{N}\rangle\):
\begin{equation}\label{eq:meanFieldODE_updated}
  \frac{d}{dt}\,\langle \hat{N}\rangle
  \;\approx\;
  m_1\,\langle \hat{N}\rangle
  \;-\;
  m_2\,\langle \hat{N}\rangle^2,
\end{equation}
with coefficients \(m_1,m_2\) derived from expansions of \(\mathcal{R}(\hat{N})\) and annihilation rates. The corresponding “sigmoidal” growth trajectory recapitulates a classical logistic regime~\cite{NowakMay2000}, consistent with real infections that show exponential increases at low viral loads but plateau in response to finite resources or immune pressure. Because \(\mathcal{R}(x)\) diminishes at large \(x\), \(\hat{\mathcal{G}}_{\mathrm{tot}}\) remains dissipative, ensuring global existence of the semigroup via Hille–Yosida or Lumer–Phillips theorems~\cite{DaPratoZabczyk2014,Pazy1983}. Biologically, this self‐limiting replication is crucial for accurate modeling: while viruses can proliferate explosively at first, they rarely expand without bound. The operator form of saturable creation precisely captures this effect, preserving mathematical stability and biological plausibility in multi‐lattice viral simulations.
\end{theorem}
\subsection{The Many‐Lattice Schr\"odinger‐Like Equation and Its Inherited Sectorial Structure}
\noindent
Having established the single‐lattice framework and the construction of a Fock space to accommodate an \emph{unbounded} number of viral lattices, we now consider the full dynamical evolution generated by a suitable operator \(\hat{\mathcal{G}}_{\mathrm{tot}}\). In certain conservative cases, this leads to a Schr\"odinger‐type equation, while more general dissipative or $m$‐sectorial operators govern non‐equilibrium processes. The following theorem outlines how multiple lattice sectors come together in a single operator formalism, guaranteeing well‐posedness and ensuring no blow‐up in population despite the potential for unbounded replication.

\begin{theorem}[The Many‐Lattice Schr\"odinger‐Like Equation]
\label{thm:many_lattice_Schrodinger}
Consider an $m$‐sectorial (or maximal dissipative) operator 
\(\hat{\mathcal{G}}\)
on the \emph{single viral lattice} Hilbert space 
\(\mathcal{H}_{\mathrm{lat}}\).
When second quantized over bosonic or fermionic Fock space 
\(\mathcal{F}_{\pm}(\mathcal{H}_{\mathrm{lat}})\),
the resulting many‐lattice dynamics satisfy a Schr\"odinger‐like equation (or a more general semigroup evolution if
\(\hat{\mathcal{G}}\)
is non‐self‐adjoint):
\begin{equation}
  \partial_{t} \Psi_{\mathrm{tot}}(t)
  \;=\;
  \hat{\mathcal{G}}_{\mathrm{tot}}
  \,\Psi_{\mathrm{tot}}(t),
  \quad
  \Psi_{\mathrm{tot}}(t)
  \;\in\;
  \mathcal{F}_{\pm}\bigl(\mathcal{H}_{\mathrm{lat}}\bigr).
\end{equation}
\noindent
In the bosonic (symmetric) setting, \(\Psi_{\mathrm{tot}}(t)\) can be viewed as an \(L^2\) function over $M$ indistinguishable lattices (each itself described by the PDE or arrangement degrees of freedom). In the fermionic (antisymmetric) scenario, one imposes an exclusion principle; though less common biologically for a system of virus particles, it could be introduced to reflect certain steric or biochemical constraints. If the single‐lattice operator \(\hat{\mathcal{G}}\) is $m$‐sectorial, then \(\hat{\mathcal{G}}^{(M)}\) is also $m$‐sectorial on the $M$‐lattice sector, and the direct sum
\(\hat{\mathcal{G}}_{\mathrm{tot}}\)
maintains the required sectorial bounds \cite{BratteliRobinson1987}. Consequently,
\(\hat{\mathcal{G}}_{\mathrm{tot}}\)
generates a strongly continuous semigroup \(e^{\,t\hat{\mathcal{G}}_{\mathrm{tot}}}\) on
\(\mathcal{F}_{\pm}(\mathcal{H}_{\mathrm{lat}})\),
ensuring a well‐posed many‐lattice evolution problem.

\end{theorem}
\begin{proof}[Proof of Sectorial Extension to Fock Space]
\label{prop:Gtot_sectorial}

Let \(\hat{\mathcal{G}}\) be an $m$‐sectorial (or maximal dissipative) operator on the single viral lattice Hilbert space \(\mathcal{H}_{\mathrm{lat}}\). We claim that the second‐quantized operator 
\(\hat{\mathcal{G}}_{\mathrm{tot}}\) (cf.\ Theorem~\ref{thm:dynamical_interpretation_semigroup})
extends this $m$‐sectorial property to the full Fock space
\begin{equation}
  \mathcal{F}_{\pm}\bigl(\mathcal{H}_{\mathrm{lat}}\bigr),
\end{equation}
thereby generating a strongly continuous semigroup 
\begin{equation}
  \bigl\{\,e^{-\,i\,t\,\hat{\mathcal{G}}_{\mathrm{tot}}}\bigr\}_{t\ge 0}.
\end{equation}
Because \(\hat{\mathcal{G}}\) can include damping or resource‐limiting effects, this semigroup is generally \emph{non‐unitary}, reflecting dissipative dynamics. Since \(\hat{\mathcal{G}}\) is dissipative on \(\mathcal{H}_{\mathrm{lat}}\), there exists \(\gamma \ge 0\) such that
\begin{equation}
  \mathrm{Re}\,\bigl\langle \hat{\mathcal{G}}\,\psi,\;\psi \bigr\rangle 
  \;\le\;
  -\,\gamma\,\|\psi\|^2
  \quad
  \text{for all } \psi \in D(\hat{\mathcal{G}}).
\end{equation}
(Here, the parameter \(\gamma\) is unrelated to any interaction potential for local virus particles; it denotes a lower bound on the operator’s real part.) On the \(M\)\,-lattice sector 
\(\bigl(\mathcal{H}_{\mathrm{lat}}^{\otimes M}\bigr)_{\pm}\),
the induced operator is
\begin{equation}
  \hat{\mathcal{G}}^{(M)} 
  \;=\;
  \sum_{j=1}^M
  I \,\otimes\, \cdots \,\otimes\, 
  \hat{\mathcal{G}} 
  \,\otimes\, \cdots \,\otimes\, I,
\end{equation}
so for any 
\(\Phi \in D\bigl(\hat{\mathcal{G}}^{(M)}\bigr)\),
\begin{equation}
  \mathrm{Re}\,\bigl\langle \hat{\mathcal{G}}^{(M)}\Phi,\;\Phi \bigr\rangle
  \;=\;
  \sum_{j=1}^M
  \mathrm{Re}\,
  \Bigl\langle 
    \hat{\mathcal{G}}\,\Phi_j,\,
    \Phi_j
  \Bigr\rangle,
\end{equation}
where each \(\Phi_j\) is the component of \(\Phi\) on which \(\hat{\mathcal{G}}\) acts in the \(j\)th tensor factor. By the dissipativity of \(\hat{\mathcal{G}}\),
\begin{equation}
  \mathrm{Re}\,\bigl\langle 
    \hat{\mathcal{G}}\,\Phi_j,\,
    \Phi_j
  \bigr\rangle 
  \;\le\;
  -\,\gamma\,\|\Phi_j\|^2,
\end{equation}
so
\begin{equation}
  \mathrm{Re}\,\bigl\langle \hat{\mathcal{G}}^{(M)}\Phi,\;\Phi \bigr\rangle
  \;\le\;
  -\,\gamma 
  \sum_{j=1}^M \|\Phi_j\|^2
  \;\le\;
  -\,\gamma\,\|\Phi\|^2.
\end{equation}
Hence, \(\hat{\mathcal{G}}^{(M)}\) is dissipative on each \(M\)\,-lattice sector, and the same inequality extends to the full Fock space
\(\bigoplus_{M=0}^\infty \bigl(\mathcal{H}_{\mathrm{lat}}^{\otimes M}\bigr)_{\pm}\).

Next, to show \(\hat{\mathcal{G}}_{\mathrm{tot}} = \bigoplus_{M=0}^\infty \hat{\mathcal{G}}^{(M)}\) is \emph{maximal} (or $m$‐sectorial), one verifies that \((I - \hat{\mathcal{G}}_{\mathrm{tot}})\) is surjective (or has dense range) in 
\(\mathcal{F}_{\pm}\bigl(\mathcal{H}_{\mathrm{lat}}\bigr)\). 
Classical results on tensor‐product operators (see, for example, \cite[Chapter X]{ReedSimon1980} or \cite[Vol.\,2]{BratteliRobinson1987}) assert that if \(\hat{\mathcal{G}}\) is maximal dissipative on \(\mathcal{H}_{\mathrm{lat}}\), then 
\(\hat{\mathcal{G}}_{\mathrm{tot}}\)
extends this property to the direct sum of symmetrized (or antisymmetrized) \(M\)‐fold tensor products. Concretely, to solve
\begin{equation}
  (I - \hat{\mathcal{G}}_{\mathrm{tot}})\,\Phi 
  \;=\;
  \Theta,
\end{equation}
one addresses each \(M\)\,-lattice sector independently, using the resolvent \((I - \hat{\mathcal{G}})^{-1}\) on the single‐lattice factor. By the Lumer–Phillips theorem, an $m$‐sectorial operator on a Hilbert space generates a strongly continuous semigroup of contractions (or quasi‐contractions). Therefore,
\begin{equation}
  e^{-\,i\,t\,\hat{\mathcal{G}}_{\mathrm{tot}}}
  \quad
  (t\ge0)
\end{equation}
is well‐defined on 
\(\mathcal{F}_{\pm}\bigl(\mathcal{H}_{\mathrm{lat}}\bigr)\),
but is generally \emph{non‐unitary} if 
\(\mathrm{Re}\,\langle \hat{\mathcal{G}}\psi,\psi\rangle <0\). 
This confirms that \(\hat{\mathcal{G}}_{\mathrm{tot}}\) generates a well‐posed semigroup for \emph{many‐lattice} evolution in Fock space. Because
\(\mathcal{F}_{\pm}(\mathcal{H}_{\mathrm{lat}})\)
admits \emph{arbitrarily many} viral lattices, \(\hat{\mathcal{G}}\) describes the elastic/damped behavior of a \emph{single viral lattice}. In more elaborate infection models, replication or decay processes may vary the number of lattices over time, creating off‐diagonal couplings between adjacent Fock sectors. Nonetheless, so long as dissipativity conditions hold (e.g.\ resource limitations, damping), one avoids finite‐time blow‐up and secures stable dynamics. As a result, the operator \(\hat{\mathcal{G}}_{\mathrm{tot}}\) is $m$‐sectorial on 
\(\mathcal{F}_{\pm}(\mathcal{H}_{\mathrm{lat}})\)
and generates a strongly continuous (non‐unitary) semigroup
\begin{equation}
  \{e^{-\,i\,t\,\hat{\mathcal{G}}_{\mathrm{tot}}}\}_{t\ge0}.
\end{equation}
A many‐lattice state
\begin{equation}
  \ket{\Psi(t)}
  \;\in\;
  \mathcal{F}_{\pm}\bigl(\mathcal{H}_{\mathrm{lat}}\bigr)
\end{equation}
thus evolves according to this semigroup, capturing dissipative or resource‐driven processes for an unbounded viral population. The PDE, arrangement labels, and sectorial operator framework together yield a rich description of large‐scale, multi‐lattice virological phenomena under realistic environmental and host constraints.
\end{proof}

\begin{theorem}[Well‐Posedness in Many‐Lattice Fock Space]
\label{thm:many_lattice_well_posed}
Let
\(\hat{\mathcal{G}}_{\mathrm{tot}}\)
be $m$‐sectorial on
\(\mathcal{F}_{\pm}(\mathcal{H}_{\mathrm{lat}})\).
Then for each initial condition 
\(\ket{\Psi_0}\in D(\hat{\mathcal{G}}_{\mathrm{tot}})\)
and for each forcing term
\(\ket{\Phi(t)}\)
in an appropriate Bochner‐integrable class, there is a unique \emph{mild solution}
\begin{equation}
\label{eq:Fock_mild_solution}
  \ket{\Psi(t)}
  \;=\;
  e^{-\,i\,t\,\hat{\mathcal{G}}_{\mathrm{tot}}}\,\ket{\Psi_0}
  \;+\;
  \int_{0}^{t}
  e^{-\,i\,(t-s)\,\hat{\mathcal{G}}_{\mathrm{tot}}}
  \,\ket{\Phi(s)}\,ds,
\end{equation}
such that 
\(\ket{\Psi(t)}\in \mathcal{F}_{\pm}(\mathcal{H}_{\mathrm{lat}})\) 
for all \(t\ge 0\). Moreover, if \(\ket{\Phi}\equiv 0\) (i.e.\ no external forcing) and the dissipativity 
constant of 
\(\hat{\mathcal{G}}_{\mathrm{tot}}\)
is strictly positive (\(\gamma>0\)), then 
\(\|\ket{\Psi(t)}\|\)
is non‐increasing in time, reflecting the decay or ``leakage'' typical of dissipative processes.
\end{theorem}
\begin{proof}[Sketch of Proof]
From the Lumer–Phillips theorem (or Hille–Yosida criteria for $m$‐sectorial operators), one knows that an $m$‐sectorial operator on a Hilbert space generates a strongly continuous contraction (or quasi‐contraction) semigroup \cite{Pazy1983, ReedSimon1980}. Consequently,
\begin{equation}
  e^{-\,i\,t\,\hat{\mathcal{G}}_{\mathrm{tot}}}
  \quad
  (t\ge 0)
\end{equation}
is well‐defined on
\(\mathcal{F}_{\pm}\bigl(\mathcal{H}_{\mathrm{lat}}\bigr)\),
and satisfies the usual properties of existence, uniqueness, and continuous dependence on initial data in the mild solution sense \eqref{eq:Fock_mild_solution}. If 
\(\hat{\Phi}(t)\)
is Bochner‐integrable, one incorporates it via the standard variation of parameters approach, ensuring a unique solution in 
\(
  C\bigl([0,T],\mathcal{F}_{\pm}(\mathcal{H}_{\mathrm{lat}})\bigr).
\)
When \(\gamma>0\), the semigroup is strictly contractive in norm. In particular,
\begin{equation}
  \mathrm{Re}\,\bigl\langle 
    \hat{\mathcal{G}}_{\mathrm{tot}} \,\Psi,\,\Psi 
  \bigr\rangle
  \;\le\;
  -\,\gamma\,\|\Psi\|^2,
\end{equation}
which implies \(\|\Psi(t)\|^2\) decays over time. Physically, this formalizes damped or resource‐limited dynamics for an \emph{unbounded} number of viral lattices, guaranteeing no finite‐time blow‐up. Thus, the system remains stable and well‐posed, aligning with biological realism where even a large population of viral lattices eventually dissipates or is constrained by finite host resources~\cite{KnipeHowley2020,harvey2019viral}.
\end{proof}
\subsection{Markovian Jumps in Fock Space}

\noindent
Modeling large populations of virions necessitates capturing both \emph{continuous} dynamics (e.g., elastic deformations, wave‐like oscillations) and \emph{discrete} events (e.g., replication, degradation, conformational flips). In virology, this latter class of phenomena is often governed by chance: a virion might spontaneously undergo a mutation, a capsid rearrangement, or be cleared by the immune system. Such randomness naturally lends itself to \emph{Markovian dynamics}, where the probability of moving to a new ``state'' depends only on the present state, not on past history. Within an operator‐theoretic or Fock‐space framework, these transitions manifest as \emph{Markovian jumps}, bridging deterministic PDE evolution with probabilistic, event‐driven changes.

\begin{definition}[Markovian Jumps in Fock Space]
\label{def:MarkovianJumps}
To capture realistic infection scenarios, we incorporate \emph{Markovian jumps} that drive transitions between different occupancy sectors \((n \to n\pm 1)\) and/or among internal arrangement states. Concretely:
\begin{enumerate}
  \item \emph{Replication/Creation Events:}  
    Operators \(\hat{a}^\dagger(\varphi)\) increase \(\hat{N}\), modeling new lattice formation when resources (proteins, nucleotides, ATP) are available.

  \item \emph{Degradation/Annihilation Events:}  
    Operators \(\hat{a}(\varphi)\) decrease \(\hat{N}\), reflecting clearance by host immunity or the gradual breakdown of a lattice over time.

  \item \emph{Arrangement Transitions:}  
    Operators acting on the discrete arrangement space \(\mathcal{H}_{\mathrm{arr}}\) (e.g., \(\hat{H}_{\mathrm{arr}}\)) switch among conformational or mass‐band classes within each viral lattice. These capture morphological or biochemical changes in the capsid and associated proteins.
\end{enumerate}

\noindent
When these Markov jumps couple to individual arrangement states (encoded by \(\mathcal{H}_{\mathrm{arr}}\) within each ``\emph{lattice sector}'' of \(\mathcal{H}_{\mathrm{lat}}\)), the system can simultaneously shift its occupancy number \emph{and} its internal configuration. This yields an operator‐algebraic scheme that fuses:
\begin{enumerate}
  \item \emph{Deterministic PDE Components:} Wave‐like and viscoelastic dynamics within each lattice (the continuum part of \(\mathcal{H}_{\mathrm{lat}}\)).
  \item \emph{Discrete, Probabilistic Jump Processes:} Conformational rearrangements and birth/death (replication or decay) at the population scale.
\end{enumerate}
Such a framework is essential in \emph{viral lattice theory}, where we envision an \emph{unbounded} number of viral lattices—each treated as a single interconnected entity with rich internal degrees of freedom. Mathematically, we embed \(\mathcal{H}_{\mathrm{lat}}\) (the single‐lattice Hilbert space) into a Fock space that accommodates arbitrarily many lattices:
\begin{equation}
  \mathcal{F}\bigl(\mathcal{H}_{\mathrm{lat}}\bigr)
  \;=\;
  \bigoplus_{n=0}^{\infty} \bigl(\mathcal{H}_{\mathrm{lat}}^{\otimes n}\bigr)_{\mathrm{sym/asym}},
\end{equation}
where symmetrization or antisymmetrization corresponds to the effective ``particle statistics'' (bosonic or fermionic). In practice, bosonic statistics are often most relevant for large, indistinguishable populations of viral lattices, while fermionic constraints may be used to enforce no‐overlap conditions or other exclusion principles in specialized contexts.
\end{definition}

\medskip
\noindent
\textbf{Why Markovian Jumps Matter for Viral Systems.}
\begin{itemize}
    \item \emph{Capturing Stochasticity in Infections.}  
    Real infections feature abrupt events (e.g., capsid assembly failures, immune detection, genome mutations). By modeling these as Markovian jumps, one straightforwardly incorporates random transitions without forfeiting the rigors of PDE evolution in \(\mathcal{H}_{\mathrm{lat}}\).

    \item \emph{Linking Occupancy Changes to Internal States.}  
    In large populations, replication events (\(\hat{a}^\dagger\)) or annihilation events (\(\hat{a}\)) can hinge on a lattice’s internal conformation. For instance, a capsid might need to achieve a specific arrangement before it can bud off from the host cell. Markovian transitions allow such coupling between “arrangement state” and “replication probability.”

    \item \emph{Multi‐Lattice Evolution in Fock Space.}  
    Once replicated, newly formed lattices enter the next Fock sector. Likewise, cleared lattices depart the system. This built‐in mechanism for birth/death processes means one need not redefine the base Hilbert space for each population change—a central advantage of second‐quantized (Fock‐based) modeling. 
\end{itemize}

\noindent
Thus, Markovian jumps serve as a versatile tool, blending stochastic event modeling (common in epidemiology or birth‐death processes) with the continuous structural and dynamical richness of \emph{viral lattice theory}. By placing both under one operator‐theoretic roof, we obtain a coherent, mathematically sound picture of how virions replicate, degrade, and morph—be it at the level of a single capsid or across vast populations in vivo.

A natural next step after introducing Markovian jumps in a Fock‐space setting (Section~\ref{def:MarkovianJumps}) is to formalize how we \emph{count} viral lattices and track the \emph{probabilistic flows} that describe replication, clearance, and morphological changes in detail. In classical birth‐death processes, one typically defines a ``population size'' variable; here, we do likewise via an operator‐based perspective.

\begin{definition}[The Number Operator]
\label{def:NumberOperator}
The \textbf{number operator} 
\(\hat{N}\)
tracks how many ``viral‐lattice quanta'' (i.e., identical lattices) exist in the system. Formally, 
\begin{equation}
  p_{n}
  \;=\; 
  \mathbb{P}\bigl(\hat{N} = n\bigr)
\end{equation}
denotes the steady‐state probability of being in the \(n\)‐lattice sector of the Fock space. In this operator‐algebraic framework, \emph{Markovian jumps} can shift the system between different occupancy states (\(n \to n \pm 1\)) or among arrangement sectors within each occupancy. Biologically, these transitions represent events such as:
\begin{itemize}
  \item \emph{Replication} (increasing \(\hat{N}\)) if host resources permit new lattice formation,
  \item \emph{Degradation/clearance} (decreasing \(\hat{N}\)) via immune attack or spontaneous lattice breakdown,
  \item \emph{Morphological reconfigurations} of capsid or envelope, switching an existing lattice to a different arrangement sector.
\end{itemize}
\end{definition}

\begin{definition}[Markovian Jump Rates]
\label{def:markovRates}
Let 
\(\nu_{n\to n+1}\)
and 
\(\nu_{n+1\to n}\)
denote net transition rates between the \(n\)‐lattice sector and the \((n+1)\)‐lattice sector. These account for the stochastic processes by which new lattices appear or existing ones disappear. In steady state, the probability \emph{flux} from \(n\) to \(n+1\) is
\begin{equation}
  J_{(n\to n+1)}
  \;=\;
  p_{n}\,\nu_{n\to n+1}
  \;-\;
  p_{n+1}\,\nu_{n+1\to n}.
\end{equation}
A nonzero 
\(J_{(n\to n+1)}\)
indicates a persistent flow of probability mass between these two Fock sectors, signifying ongoing replication or clearance. In virological terms, this might correspond to a net gain in viral lattices (when creation outpaces annihilation), or a net loss (when clearance dominates), depending on host resource availability and immune activity.
\end{definition}
\begin{figure}[h!]
\centering
\begin{tikzpicture}[>=stealth,shorten >=1pt,shorten <=1pt,
  thick,font=\small,scale=1.4]  
  \node (n) at (0,0)   {$n$};
  \node (nplus) at (4,0) {$n+1$};

  \draw[->]
    (n) -- (nplus)
    node[midway,above]
    {$\nu_{n\to n+1}$};

  \draw[->]
    (nplus) to[bend left=30]
    node[above]
    {$\nu_{n+1\to n}$}
    (n);

  \draw[->,thick,dotted]
    (2,0.7) -- (2,1.2)
    node[midway,right] {$J_{(n\to n+1)}$};
\end{tikzpicture}
\caption{A schematic of the Markovian jump rates between adjacent lattice‐number sectors. 
A system in sector \(n\) can transition to \(n+1\) at rate \(\nu_{n\to n+1}\), or in the reverse direction at rate \(\nu_{n+1\to n}\). 
The \emph{flux} \(J_{(n\to n+1)}\) accounts for the net probability current flowing from \(n\) to \(n+1\), defined by 
\(\;p_{n}\,\nu_{n\to n+1} \;-\; p_{n+1}\,\nu_{n+1\to n}\). 
In steady‐state, a nonzero \(J_{(n\to n+1)}\) indicates persistent circulation among these number sectors.}
\label{fig:markovRatesDiagram}
\end{figure}
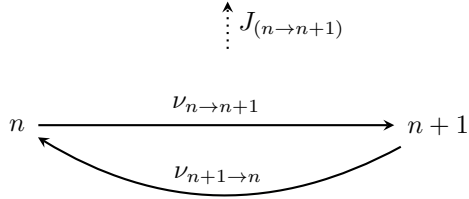
\smallskip
\noindent
\begin{definition}[Probability Currents and Viral Load Evolution]
\label{def:ProbCurrentViralLoad}
Let \(P(\mathbf{n},t)\) solve a Master equation describing the probability that the system occupies configuration \(\mathbf{n}\in y\) at time \(t\). Here, \(\mathbf{n}\) may label both occupancy numbers (i.e., how many lattices exist) and \emph{internal} arrangement states (conformational or otherwise). The \textbf{probability current} from \(\mathbf{n}\) to \(\mathbf{n}'\) is
\begin{equation}\label{eq:ProbCurrent}
  J_{\mathbf{n}\to \mathbf{n}'}(t) 
  \;=\; 
  P(\mathbf{n},t)\,T_{\mathbf{n}\to \mathbf{n}'} 
  \;-\;
  P(\mathbf{n}',t)\,T_{\mathbf{n}'\to \mathbf{n}},
\end{equation}
where \(T_{\mathbf{n}\to \mathbf{n}'}\) is the Markov jump rate for transitions \(\mathbf{n}\to \mathbf{n}'\). These currents track how the system’s probability measure \(\mu\) \emph{flows} through the high‐dimensional configuration space, be it from lower occupancy to higher occupancy or between different lattice conformations.

\smallskip
\noindent
\paragraph{Biological Significance of Non‐Equilibrium Flows.}  
Under open‐system conditions (i.e., resource or immune interactions), persistent cycles in these probability currents indicate \emph{non‐equilibrium steady states}, where viral loads or lattice conformations cycle indefinitely rather than settling to a single equilibrium. For instance, repeated replication–clearance cycles might keep a certain lattice arrangement in flux, even in a steady‐state distribution. This phenomenon aligns with real‐world infections, where viruses can continue to replicate and mutate under host pressures, thereby preventing a static equilibrium from emerging. Observing or simulating these probability currents offers virologists insights into critical transitions (e.g., thresholds at which replication overwhelms immune response) that may guide therapeutic strategies. The interplay of occupancy number, arrangement states, and probabilistic flows in configuration space yields a unified operator‐based model that captures both \emph{continuous} PDE dynamics (e.g., within each lattice) and \emph{discrete} Markov jumps (e.g., new lattice formation or immune removal). By quantifying these probability currents, we gain a deeper understanding of how large viral populations behave under real‐life constraints, bridging the gap between quantum‐inspired operator formalisms and the empirical patterns seen in virological experiments.
\end{definition}

\noindent
\section{Path Integrals, Orbits, and Trajectory Expansions in Fock Space}
\label{subsubsec:Orbits_and_Path_Integrals}

\noindent
Many problems in viral lattice theory, especially those involving large populations and multiple Markov transitions, can benefit from a \emph{path integral} viewpoint. Rather than tackling each state evolution in isolation, we consider the entire \emph{trajectory} of the system as it moves across occupancy sectors (via creation/annihilation events) and arrangement states (through conformational shifts). This approach is reminiscent of Feynman’s path integrals in quantum mechanics, where one sums over all possible paths to glean statistical and dynamical information. Here, each “path” weaves through a high‐dimensional space, \(\mathcal{F}(\mathcal{H}_{\mathrm{lat}})\), involving both continuous PDE evolution within each lattice and discrete Markov jumps that shift occupancy or rearrangements. In this extended configuration space, a system state \(\Psi \in \mathcal{F}\!\bigl(\mathcal{H}_{\mathrm{lat}}\bigr)\) can traverse either:
\begin{enumerate}
 \item \textit{Orbits:}  
    In the language of dynamical systems, an \emph{orbit} is the trajectory (or solution curve) produced by evolving an initial state \(\Psi_0 \in \mathcal{F}(\mathcal{H}_{\mathrm{lat}})\) forward in time under the combined PDE–Markov–Fock generator. Concretely, one starts with a particular arrangement and occupant configuration, then follows how the PDE modes (e.g.\ displacement fields) evolve continuously while discrete creation or annihilation jumps change occupant number, and Markovian events switch arrangement labels. Purely mathematically, the orbit is the image of the mild solution \(t \mapsto \Psi(t)\) in the configuration space. Biologically, a viral lattice (or population of lattices) “experiences” an orbit as it transitions smoothly (via PDE flows) and abruptly (via occupant or capsid rearrangements) through a succession of states, reflecting real infection dynamics over time.

  \item \textit{Loops:}  
    A \emph{loop} (or closed orbit) is an orbit that eventually returns to its original state in the Hilbert–Fock space—i.e., \(\Psi(T) = \Psi(0)\) for some finite \(T>0\). Mathematically, this represents a \emph{periodic} or \emph{cyclic} solution under the PDE–Markov–Fock evolution. From a virological standpoint, loops can occur if the system revisits the same arrangement or occupant distribution multiple times (e.g.\ repeated cycles of replication and clearance). In non‐equilibrium physics, such loops often indicate a persistent flow of resources (ATP, amino acids, etc.) driving the viral population to circle through the same sequence of states. A nonzero flux integral around the loop then implies a \emph{non‐equilibrium steady state} (NESS), embodying active circulation among occupancy or arrangement modes.
\end{enumerate}
\noindent

\noindent
\textbf{Piecewise Paths and Non‐Equilibrium Cycles.}
\noindent
Each resulting trajectory is “piecewise” in nature, blending deterministic PDE motion with probabilistic jumps. From a topological lens, the presence of loops or cycles—in which the system keeps re‐entering certain configurations—often signals a \emph{non‐equilibrium} regime. Resource inflows and ongoing host–virus interactions can drive these repeated transitions, giving rise to steady fluxes that do not die out over time. Practically, this may manifest as a virus population repeatedly swelling and shrinking while undergoing specific capsid rearrangements, a behavior well captured by path integrals and orbit/loop analysis in the high‐dimensional Fock‐space description. Hence, from both a rigorous mathematical perspective (where we track all PDE, Markov, and occupant variables in \(\mathcal{F}(\mathcal{H}_{\mathrm{lat}})\)) and an intuitive virological viewpoint, the concepts of \emph{orbits} (open‐ended evolution) and \emph{loops} (recurrent or periodic trajectories) provide a powerful lens for understanding how viral lattices move through their configuration space—either as once‐through evolutionary paths or as cycling states driven by non‐equilibrium forces.

\begin{itemize}
    \item \emph{Continuous Segments:}  
    While in a specific arrangement sector \(y\) and occupant number \(n\), the lattice’s PDE degrees of freedom evolve smoothly (e.g., elastic vibrations, damping).  
    \item \emph{Discrete Jumps:}  
    A sudden replication event might bump the system from \((y,n)\) to \((y',\,n+1)\), or an immune‐mediated clearance could drive \((y,n)\to(y',\,n-1)\).  
\end{itemize}
\medskip
\noindent
\textbf{Intuitive Value of Path Integrals.}  
Just as Feynman’s original path integrals reveal \emph{dominant} contributions to particle trajectories, the viral lattice path integral framework pinpoints the most probable routes through arrangement and occupancy space. For instance, among all possible ways a virion might replicate or degrade, the path integral summation weighs each route by its Markov transition rate and PDE “action,” guiding us to the orbits that dominate population dynamics or morphological transitions. This offers:
\begin{itemize}
    \item \emph{Statistical Averages and Fluctuations:}  
    One can compute mean viral load, variation in replication times, or probabilities of specific conformational changes by summing over all possible state‐space paths.  
    \item \emph{Ensemble Interpretations:}  
    Instead of focusing on a single solution trajectory, we look at the entire “ensemble” of plausible evolutions, each with its own weighting. This is crucial for capturing the stochastic nature of viral replication.  
    \item \emph{Geometric and Topological Insights:}  
    Loops or cycles in path space highlight where the system recurrently visits certain states—key to understanding NESS and sustained oscillatory behavior in resource‐driven environments.
\end{itemize}
Despite the high‐dimensional complexity, the theory remains \emph{computationally and analytically manageable} thanks to rigorous functional‐analytic tools:
\begin{enumerate}
  \item \emph{Sectorial/Maximal Dissipative Generators:}  
  Each single‐lattice operator (encompassing PDE evolution plus arrangement transitions) can be made $m$‐sectorial if dissipativity and Lipschitz‐type constraints are met~\cite{DaPratoZabczyk1992,Pazy1983}. This prevents blow‐up in the PDE sector and ensures Markov jumps are well‐defined in infinite dimensions.

  \item \emph{Markov Creation–Annihilation Structure:}  
  Shifting occupant numbers by \(\pm1\) via creation/annihilation operators preserves overall sectoriality when resource limitations and decay rates are sublinear in \(\hat{N}\). This allows replication to be “unbounded” in principle, yet realistically capped by finite resources.

  \item \emph{Semigroup Formalism:}  
  Under Lumer–Phillips or Hille–Yosida conditions, these $m$‐sectorial operators generate strongly continuous semigroups in the Fock space. As a result, solutions depend continuously on initial states, enabling spectral analysis and a well‐posed evolution of occupant distributions over time.
\end{enumerate}

\noindent
\textbf{Biological and Computational Payoffs.}  
By combining Markovian jumps (for random conformational changes, replication, or immune clearance) with continuous PDE modes (for mechanical or wave‐based processes), we \emph{mirror} the realistic complexity of virus–host interactions. In practice, one can:
\begin{itemize}
  \item \emph{Identify Most Likely Replication Pathways:}  
  Evaluate which sequences of occupant‐state transitions dominate under given resource constraints.  
  \item \emph{Quantify Morphological State Fluctuations:}  
  Track how partial genome packaging or capsid subunit rearrangements vary over time, leveraging the path integral’s ability to sum over discrete conformational jumps.  
  \item \emph{Explore Non‐Equilibrium Steady States:}  
  Detect persistent loops in path space (e.g., repeated birth‐death cycles) that reflect real‐world infections where the virus never fully clears nor explodes, but remains in a dynamic steady‐state driven by resource inflows.
\end{itemize}

\noindent
In short, the path integral perspective grants a “trajectory‐based” lens on viral lattice behavior that is both mathematically comprehensive and intuitively graspable. By focusing on orbits, loops, and fluxes in an infinite‐dimensional configuration space, one can visualize the swirling currents of replication and clearance—much like studying flow lines on a phase portrait, but lifted to a quantum‐inspired operator setting. This viewpoint enriches our capacity to analyze how small‐scale PDE deformations and discrete Markov events conspire to shape large‐scale viral population outcomes, uniting rigorous physics with the intricate realities of virological systems.

\begin{figure}[h!]
\centering
\begin{tikzpicture}[>=stealth,shorten >=1pt,shorten <=1pt,
  thick,font=\small,scale=1.4]

  \node (y1n1) at (0,0)   {$(y_{1},n_{1})$};
  \node (y2n1) at (3,0)   {$(y_{2},n_{1})$};
  \node (y2n2) at (1.5,2) {$(y_{2},n_{2})$};
  
  \node (y1n2) at (-1.5,2) {$(y_{1},n_{2})$};

  \draw[->] (y1n1) to[bend left=15] 
    node[above] {$\nu_{y_1 \to y_2}$}
    (y2n1);

  \draw[->] (y2n1) to[bend left=15]
    node[right] {$\nu_{n_1 \to n_2}$}
    (y2n2);

  \draw[->] (y2n2) to[bend left=15]
    node[left] {$\nu_{y_2 \to y_1}$}
    (y1n1);

  \draw[->] (y1n1) to[bend left=25]
    node[left] {$\nu_{n_1 \to n_2}$}
    (y1n2);

  \draw[->] (y1n2) to[bend left=25]
    node[left] {$\nu_{y_1 \to y_2}$}
    (y2n2);

  \draw[->,dashed] (y2n1) to[bend left=15]
    node[below] {$\nu_{y_2 \to y_1}$}
    (y1n1);

  \draw[->,dashed] (y2n2) to[bend left=15]
    node[right] {$\nu_{n_2 \to n_1}$}
    (y2n1);

  \draw[->,dashed] (y1n1) to[bend left=15]
    node[left] {}
    (y2n2);

  \draw[->,dashed] (y1n2) to[bend left=25]
    node[right] {$\nu_{n_2 \to n_1}$}
    (y1n1);

  \draw[->,dashed] (y2n2) to[bend left=25]
    node[right] {$\nu_{y_2 \to y_1}$}
    (y1n2);

\end{tikzpicture}
\caption{\label{fig:closedLoopDiagram}
\textbf{Feynman‐like Diagram of Typical Loops in Fock Space.}
This schematic shows how the system can traverse different arrangement states 
(\(y_1 \leftrightarrow y_2\)) at fixed occupant number (\(n_1\)) and also undergo 
birth/death transitions (\(n_1 \leftrightarrow n_2\)) to reach new population sectors 
(e.g., \((y_{2},n_{2})\) or \((y_{1},n_{2})\)). Solid arrows illustrate primary transitions, 
while dashed arrows indicate possible reverse paths. Together, these loops form \emph{closed orbits} 
in the state space, reminiscent of Feynman diagrams for particle scattering. 
A nonzero \emph{net flux} around any loop signals persistent probability current 
(e.g., ongoing replication–clearance cycles), thus indicating a non‐equilibrium steady state. 
Such circulations can be interpreted as repeated patterns of conformational change 
followed by occupant‐number shifts, or vice versa, providing insight into how 
\emph{realistic} viral lattice systems might sustain repeated replication events 
before cycling back to earlier configurations.}
\end{figure}
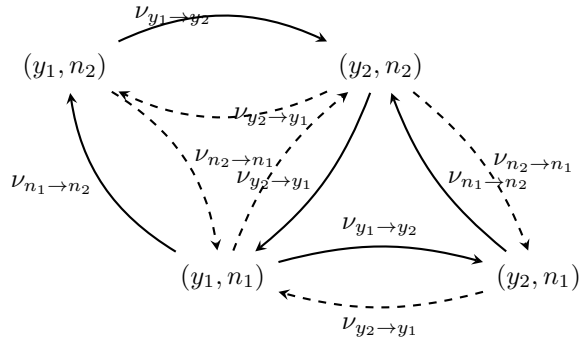

\begin{definition}[Closed Loops in Fock-Space Configurations]
\label{def:closedLoopFock}
Consider a single viral lattice, modeled as an $8\times 8$ array of virion nodes.  
Within each lattice, the \emph{arrangement space} 
\(\mathcal{H}_{\mathrm{arr}} \ni y\)
captures discrete morphological or conformational indices (e.g.\ occupant classes, capsomere states), 
while the \emph{PDE space} 
\(\mathcal{H}_{\mathrm{PDE}}\)
encodes continuous fields describing intra-lattice interactions (e.g.\ mechanical displacement, resource density). 
When we extend to \emph{multiple} such lattices, we embed the single-lattice state space 
\(\mathcal{H}_{\mathrm{lat}} = \mathcal{H}_{\mathrm{arr}} \otimes \mathcal{H}_{\mathrm{PDE}}\) 
into a Fock space, 
\(\mathcal{F}\bigl(\mathcal{H}_{\mathrm{lat}}\bigr)\), 
to accommodate the creation and annihilation of entire lattices. In this Fock-space setting, a \emph{closed loop} is a cyclic sequence of the form:
\begin{equation}
  (y_1,\,n_1)
  \;\longrightarrow\;
  (y_2,\,n_1)
  \;\longrightarrow\;
  (y_2,\,n_2)
  \;\longrightarrow\;
  \cdots
  \;\longrightarrow\;
  (y_1,\,n_1),
\end{equation}
where \(y_i\in \mathcal{H}_{\mathrm{arr}}\) denotes a particular arrangement (e.g.\ a conformational state), and \(n_i\) is the occupant number (the total number of such lattices). This loop involves:
\begin{enumerate}[label=(\roman*)]
  \item \emph{Transitions in arrangement indices} 
  (\(y_i \to y_{i+1}\)), modeled by a discrete Markov operator that might encode changes in lattice morphology or capsomere conformation.
  \item \emph{Transitions in occupant number} 
  (\(n_i \to n_{i+1}\)), driven by creation or annihilation operators representing replication (creation of a new lattice) or clearance (annihilation of an existing lattice).
\end{enumerate}
\end{definition}
\begin{definition}[Infinite Fock Space and Recurrent Loops]
\label{def:InfiniteFock}
In an \emph{infinite} Fock space, infinitely many loops (or cycles) can intertwine a viral lattice’s internal morphologies with population‐scale replication and clearance. Conceptually, each viral lattice is treated as a “particle” in a second‐quantized (Fock) framework, thereby accommodating even very large populations—up to millions of virions arranged into thousands of lattices—through a \emph{single global wavefunction}. Formally, the configuration space for any number of lattices is
\begin{equation}
  \mathcal{F}\bigl(\mathcal{H}_{\mathrm{lat}}\bigr),
\end{equation}
the bosonic (or fermionic) Fock‐space expansion of a single‐lattice Hilbert space 
\(\mathcal{H}_{\mathrm{lat}}\). 
In this infinite‐dimensional setting, \emph{closed loops} and \emph{orbits} represent recurrent replication–clearance cycles or repeated morphological transitions, with each loop encoding a Markovian “cycle” that couples intra‐lattice (arrangement) dynamics to inter‐lattice (population) processes.
\end{definition}
\subsubsection{Mathematical Form of the Global Wavefunction in Fock Space}
\begin{theorem}[Mathematical Form of the Global Wavefunction in Fock Space]
\label{thm:GlobalWavefunctionFock}
\noindent
The bosonic (or generalized) Fock space,
\begin{equation}
   \mathcal{F}\bigl(\mathcal{H}_{\mathrm{lat}}\bigr)
   \;=\;
   \bigoplus_{N=0}^{\infty}\,
   \Bigl[\mathcal{H}_{\mathrm{lat}}^{\otimes N}\Bigr]_{\pm},
\end{equation}
accommodates configurations ranging from \(N=0\) (no viral lattices) to arbitrarily many occupant lattices. A \emph{global wavefunction} 
\begin{equation}
   \ket{\Psi(t)}
   \;\in\;
   \mathcal{F}\bigl(\mathcal{H}_{\mathrm{lat}}\bigr)
\end{equation}
then encodes all microscopic (PDE) degrees of freedom \emph{and} the macro‐level occupant count at time \(t\). Concretely, we can write
\begin{equation}
\label{eq:FockStateExpansion}
\ket{\Psi(t)}
\;=\;
\sum_{N=0}^{\infty}
\;\;\sum_{(y_1,\dots,y_N)\in \mathcal{Y}^N}
\;\;
\int_{(\Omega)^N}
\Psi_{N}\!\Bigl(y_1,\dots,y_N;\,f_1,\dots,f_N;\,t\Bigr)\;
\bigl|y_1, f_1\bigr\rangle
\otimes
\cdots
\otimes
\bigl|y_N, f_N\bigr\rangle
\,d f_1 \dots d f_N,
\end{equation}
where:
\begin{itemize}
  \item \(y_i \in \mathcal{Y}\) label discrete ``arrangement states'' (Markovian or occupant‐conformational subspaces) for each lattice,
  \item \(f_i\) represents continuous PDE variables (e.g., displacement fields) on \(\Omega \subset \mathbb{R}^3\), 
  \item \(\Psi_N\) is a complex‐valued amplitude (or probability amplitude) on these configuration spaces, satisfying appropriate normalization or integrability conditions.
\end{itemize}
\noindent
Hence, \eqref{eq:FockStateExpansion} represents a \emph{global wavefunction} bridging micro‐level mechanical details and macro‐level replication events in a mathematically rigorous way. The resulting mild solutions of the hybrid PDE–Markov–Fock equations thus offer a \emph{comprehensive picture} of infection dynamics that can, in principle, track tens of millions of virions \emph{without} sacrificing analytical precision. 
\end{theorem}
\begin{theorem}[Global Wavefunction Evolution.]
A global state
\(\ket{\Psi(t)} \in \mathcal{F}\bigl(\mathcal{H}_{\mathrm{lat}}\bigr)\)
solves
\begin{equation}
  \frac{d}{dt}\,\ket{\Psi(t)}
  \;=\;
  \mathcal{G}_{\mathrm{Fock}}\>\ket{\Psi(t)},
  \quad
  \ket{\Psi(0)}\;=\;\ket{\Psi_0}.
\end{equation}
In the amplitude representation (cf.\ \eqref{eq:FockStateExpansion}), 
\(\Psi_N(\dots,t)\)
obeys:
\begin{equation}
  \begin{cases}
    \text{(i) An $m$‐sectorial PDE generator } \hat{\mathcal{G}}_{\mathrm{PDE}}
    \text{ acting on each continuous variable } \varphi_i,\\
    \text{(ii) Discrete Markov transitions for each arrangement state } y_i,\\
    \text{(iii) Creation/annihilation terms linking $N$‐ and $(N\pm1)$‐lattice sectors.}
  \end{cases}
\end{equation}
\noindent
Consequently, the system simultaneously incorporates \emph{intra‐lattice PDE mechanics}, \emph{arrangement changes}, and \emph{population‐level birth/death events}.

\paragraph{Reconciling Fine‐Grained PDE Physics with Macroscopic Occupant Changes.}
\begin{itemize}
  \item \emph{Population‐Scale Consistency:}  
    Even if \(\ket{\Psi(t)}\) spans occupant numbers from \(N=0\) to large values (e.g.\ \(10^6\)–\(10^9\)), this single wavefunction remains finite‐norm under $m$‐sectorial and bounded‐perturbation assumptions, thus precluding nonphysical blow‐ups in finite time.
  \item \emph{Detailed Intra‐Lattice Dynamics:}  
    Each individual \(8\times8\) lattice retains the continuum PDE description (e.g.\ wave‐like or diffusive processes, subunit interactions), ensuring no loss of fine‐grained capsid‐physics detail. Moreover, Markov transitions in arrangement labels \(y_i\) can capture morphological shifts or partial (dis)assembly events.
  \item \emph{Emergent Resource Constraints:}  
    If creation (replication) rates saturate at high occupant counts, the model enforces biologically realistic limitations on viral proliferation. Immune clearance (via annihilation) can similarly prevent indefinite population growth. These features ensure alignment with experimental observations of viral loads leveling off under host immune pressure.
  \item \emph{Unified Operator‐Theoretic View:}  
    The operator \(\mathrm{d}\Gamma(\hat{\mathcal{G}}_{\mathrm{lat}})\) evolves \emph{within} each \(N\)‐lattice sector, while \(\widehat{c}^\dagger\) and \(\widehat{c}\) shift the system \emph{between} sectors, effectively stitching together micro‐scale physics (single‐lattice PDEs) and macro‐scale replication/degradation processes.
\end{itemize}
By encoding occupant‐number fluctuations via creation/annihilation and maintaining an $m$‐sectorial PDE description for each lattice’s internal modes, this hybrid formalism yields a \emph{global} wavefunction \(\ket{\Psi(t)}\) that remains both mathematically tractable and biologically relevant. Even large populations of viral lattices—numbering in the millions—can be captured in a single operator‐theoretic solution, where continuum PDE detail is preserved for \emph{every} lattice, and the occupant count evolves through realistic replication/clearance dynamics.
\end{theorem}
\medskip
\noindent
\textbf{Interpretation and Advantages:}
\begin{enumerate}[label=\textbf{(\arabic*)}]
  \item \emph{Comprehensive State Description:}  
    The wavefunction \(\ket{\Psi(t)}\) gathers \emph{all} dynamic aspects: occupant number, arrangement states, and PDE modes (for each lattice). No separate ``bookkeeping’’ is needed to track replication, clearance, or rearrangements—they emerge naturally from creation/annihilation and Markov transitions in the operator formalism.
  
  \item \emph{Scalable to Large Populations:}  
    Even if millions of \(8\times8\) lattices co‐exist, the direct‐sum/tensor‐product structure of Fock space systematically handles new occupant factors. This prevents repeated ad hoc expansions of a fixed‐\(N\) model, maintaining mathematical and simulational tractability.
  
  \item \emph{Connection to Orbits and Mild Solutions:}  
    At each fixed \(t\), \(\ket{\Psi(t)}\) is a snapshot of the full viral population’s microscopic state; over time, the path \(t \mapsto \ket{\Psi(t)}\) yields the orbit \(\Theta(t)\) in \(\mathcal{F}\bigl(\mathcal{H}_{\mathrm{lat}}\bigr)\). Dissipative PDE segments and bounded creation rates ensure no uncontrolled blow‐ups occur, making the system \emph{physically plausible} (e.g., it is not feasible for occupant numbers to become infinite in finite time).
  
  \item \emph{Novelty of the Approach:}  
    By combining second quantization (for occupant dynamics) with an $m$‐sectorial PDE description of each lattice’s internal degrees of freedom, we achieve a single operator‐algebraic model. This unifies continuum mechanics, discrete Markov transitions, and population‐scale replication within one framework—something rarely accomplished in classical virology modeling.
\end{enumerate}

\subsection{Path Integral and Probability Weights with Action Functionals}
\label{subsec:path_integrals_and_actions}

\noindent
In the preceding sections, we introduced \emph{Markovian jumps} to capture discrete replication, degradation, and conformational switches within a multi‐lattice (Fock‐space) framework. These jumps dovetail with \emph{continuous} PDE dynamics, describing wave‐like displacements or damping within each viral lattice. To further unify these perspectives—and to handle the \emph{stochastic} nature of real infection scenarios—we now turn to \emph{path integrals}. Drawing inspiration from Feynman’s reformulation of quantum mechanics, path integrals offer a way to represent system evolution not merely by forces or direct PDE solutions, but by \emph{summing over} (or integrating over) all possible trajectories the system could follow, weighting each by an \emph{action functional}.  

\smallskip
\noindent
\begin{definition}[Action Functionals in Viral Lattice Theory.]
In classical mechanics, the action functional encodes the net “cost” or “phase” associated with a trajectory in configuration space. In the present \emph{viral lattice} context, our action functional can integrate both deterministic PDE forces (e.g., elastic restoring forces, damping) and stochastic jump processes (e.g., replication or decay events). By switching to this action‐based perspective, we can tackle complex problems—such as analyzing how small perturbations spread through a high‐dimensional lattice or how random fluctuations influence viral replication rates—without needing to solve the PDE or Markov processes in a purely step‐by‐step manner. Instead, the path integral naturally \emph{weighs} each possible trajectory by an exponential of minus the action (or a suitably generalized operator analog).  
\end{definition}
\smallskip
\noindent
\begin{definition}[Probability Weights and Piecewise‐Continuous Orbits.]
Previously, we introduced the concept of an \emph{orbit} as a trajectory in the PDE–Fock space that can experience Markovian jumps. Here, we formalize that concept by specifying the path space, the cylinder sets, and the action functional. Each orbit thus represents a potential realization of the viral system over time, capturing how mechanical displacements, resource‐limited replication, and discrete conformational switches intertwine. By assigning an action $\mathcal{S}[\Theta]$, we systematically determine which orbits dominate the overall dynamics, providing a practical and conceptual bridge between microscopic model details (forces, PDEs, replication rules) and macroscopic phenomena (population‐level growth, morphological patterns). 

When embedded in an operator‐theoretic or PDE‐Fock framework, each “path” corresponds to a \emph{piecewise‐continuous orbit} of the viral system over a time interval \([0, T]\). Such orbits may involve smooth deformation of the lattice (captured by PDE evolution) punctuated by discrete Markovian jumps (creation, annihilation, or internal arrangement shifts). \emph{Probability weights} then assign likelihoods to these orbits: typically, an orbit with “lower action” (i.e., a more dynamically likely path given the system’s constraints and resource levels) carries higher weight in the overall integral. 

This viewpoint blends a Feynman‐inspired summation over paths with the practical details of stochastic replication, resource dependence, and PDE‐governed viral motion. To make these ideas precise, we need to specify the path space, the action functional, and how the measure on this space is constructed. In \emph{viral lattice theory}, an \emph{orbit} represents one possible realization of the evolving multi‐lattice system in the presence of PDE dynamics and Markov jumps. Below, we formally define these orbits, the associated path space, and the measure that assigns each path a probability weight. 
\end{definition}

\begin{definition}[Path Integral and Probability Weights with Action Functionals]
\label{def:PathIntegralEnhanced}
\noindent
Let \(\Omega\) be an abstract probability space, \(\mathcal{F}\) its \(\sigma\)‐algebra, and \(\mathbb{P}\) a probability measure. We index a family of \emph{piecewise‐continuous solution trajectories} (or \emph{orbits}) by \(\omega \in \Omega\). For each \(\omega\), define
\begin{equation}
  \Theta(\omega,\cdot) : [0,T] \;\to\; \mathcal{F}\bigl(\mathcal{H}_{\mathrm{lat}}\bigr),
\end{equation}
where \(\mathcal{F}\bigl(\mathcal{H}_{\mathrm{lat}}\bigr)\) is the multi‐lattice Fock space hosting potentially unbounded populations, each with its own PDE‐based dynamics and Markov jumps. The map \(\Theta(\omega,t)\) is assumed to be \emph{piecewise‐continuous}, meaning that between jump events, it obeys smooth PDE evolution, and at jump instants, it transitions discontinuously due to replication or annihilation.
\end{definition}
\smallskip

\noindent
\begin{definition}[Path Space.]  
The \emph{path space} is defined as:
\(
  \mathrm{Path}\bigl(\mathcal{F}(\mathcal{H}_{\mathrm{lat}})\bigr)
  =
  \Bigl\{
    \gamma \mid \gamma : [0,T] \to \mathcal{F}(\mathcal{H}_{\mathrm{lat}}),\; \gamma \text{ is piecewise-continuous under the PDE-Markov-Fock evolution}
  \Bigr\}.
\)
We endow this space with a \emph{path-space measure} \(\mu\), constructed through standard cylinder-set approximations.
\(
  \Delta := \{0 = t_0 < t_1 < \dots < t_n = T\},
\)
define the \emph{cylinder set}
\(
  C_{\Delta;\,\boldsymbol{\alpha}}
  \;=\;
  \Bigl\{\,
    \Theta(\cdot) \,\Bigm|\,
    \Theta(t_j)\in U_j \,\text{ for }j=0,1,\dots,n
  \Bigr\},
\)
where each \(U_j \subset \mathcal{F}\bigl(\mathcal{H}_{\mathrm{lat}}\bigr)\) is a measurable set, and \(\boldsymbol{\alpha}=\{(U_0,t_0),\dots,(U_n,t_n)\}\) records these constraints.
\end{definition}
\noindent
\begin{definition}[Action Functional and Probability Weights.]  
An \emph{action functional} \(\mathcal{S}[\Theta]\) assigns a (real‐valued) quantity to each path \(\Theta(\cdot)\), reflecting the dynamical “cost” or “phase” of that orbit. In many settings, one might choose
\begin{equation}
  \mathcal{S}[\Theta]
  \;=\;
  \int_{0}^{T}\mathcal{L}\bigl(\Theta(t),\,\dot{\Theta}(t)\bigr)\,dt 
  \;+\;
  \sum_{\substack{\text{jump events}}} \text{(jump penalty/benefit)},
\end{equation}
where \(\mathcal{L}\) is a Lagrangian encoding PDE forces and damping, while jump events contribute discrete modifications to the action. The associated probability weight is typically proportional to 
\(\exp\bigl(-\mathcal{S}[\Theta]\bigr)\)
(or a suitable operator analog), such that lower‐action paths have higher probability. Formally, we define
\begin{equation}
  \mu\bigl(C_{\Delta;\,\boldsymbol{\alpha}}\bigr)
  \;\propto\;
  \int_{C_{\Delta;\,\boldsymbol{\alpha}}}
  e^{-\mathcal{S}[\Theta]}
  \;\mathcal{D}\Theta,
\end{equation}
where \(\mathcal{D}\Theta\) is a notional measure element in path space. In biological or epidemiological contexts, \(\mathcal{S}\) can incorporate resource constraints, immune interactions, and mechanical tension in the capsid, thus weighting orbits by how ``energetically favored'' or ``probabilistically likely'' they are. Hence, the \emph{Path Integral and Probability Weights with Action Functionals} unifies PDE‐based continuum evolution, Markov jump transitions, and operator algebra in \emph{viral lattice theory}, generalizing the notion of summing over all possible states to summing over all possible \emph{paths} through state space. This construction mirrors Feynman’s approach, shifting from direct force‐ or vector‐based descriptions to an action/phase perspective, while retaining the complexities of replication, damping, and unbounded lattice populations.
\end{definition}
\subsubsection{Legendre Transformation from Hamiltonian to Lagrangian}

\noindent
In classical mechanics, \emph{Hamiltonians} and \emph{Lagrangians} provide two complementary ways of describing a system’s dynamics. The \textbf{Hamiltonian} \(H\) is typically expressed in terms of positions (or generalized coordinates) and their \emph{canonical momenta}, while the \textbf{Lagrangian} \(L\) is written in terms of positions and \emph{velocities}. The procedure that links these two pictures is the \emph{Legendre transformation}, which converts the Hamiltonian form \(H(\{u_f\},\{p_f\})\) into the Lagrangian form \(L(\{u_f\},\{\dot{u}_f\})\). Below, we clarify these concepts and perform the transformation step by step for a simplified viral lattice Hamiltonian.

\bigskip

\noindent
\textbf{Hamiltonian vs.\ Lagrangian in Classical (Viral) Mechanics.}
\begin{itemize}
    \item \textbf{Hamiltonian:} 
    Denoted by \(H\), this formalism represents the total energy of the system in terms of generalized coordinates (here, \(\{u_f\}\) for each mode \(f\)) and their \emph{canonical momenta} \(\{p_f\}\). For a system of normal modes, each mode \(f\) has an associated coordinate \(u_f\) and momentum \(p_f\). The Hamiltonian is often written as
    \begin{equation}
      H\bigl(\{u_f\},\{p_f\}\bigr)
      \;=\;
      \sum_{f}
      \left(
        \frac{p_f^2}{2\,m_f}
        \;+\;
        \frac{1}{2}\,\mathcal{S}_f\,u_f^2
      \right),
    \end{equation}
    reflecting the total energy (kinetic + potential) for the viral lattice modes.

    \item \textbf{Lagrangian:} 
    Denoted by \(L\), this formalism encodes the difference between \emph{kinetic energy} and \emph{potential energy} in terms of generalized coordinates \(\{u_f\}\) and \emph{velocities} \(\{\dot{u}_f\}\). A typical Lagrangian for a set of harmonic oscillators (or normal modes) has the form
    \begin{equation}
      L\bigl(\{u_f\},\{\dot{u}_f\}\bigr)
      \;=\;
      \sum_{f}
      \left(
        \tfrac{1}{2}\,m_f\,\dot{u}_f^2
        \;-\;
        \tfrac{1}{2}\,\mathcal{S}_f\,u_f^2
      \right).
    \end{equation}
    In viral lattice theory, each mode \(f\) describes a \emph{collective deformation} or vibration of the capsid. The Lagrangian approach is especially convenient for deriving the \emph{equations of motion} and for performing path integrals, where one integrates \(\exp(i S)\) (in quantum contexts) or \(\exp(- S)\) (in certain stochastic/thermal contexts), with the action \(S = \int L\,dt\).
\end{itemize}
\noindent
\begin{definition}[Legendre Transformation.]
\smallskip

\noindent
Given a Lagrangian \(L\bigl(\{u_f\},\{\dot{u}_f\}\bigr)\), one defines the \emph{canonical momentum} for each mode \(f\) as
\begin{equation}
  p_f
  \;:=\;
  \frac{\partial L}{\partial \dot{u}_f}.
\end{equation}
The \textbf{Hamiltonian} is then obtained by performing a \emph{Legendre transform} of \(L\) with respect to the velocities:
\begin{equation}
  H\bigl(\{u_f\},\{p_f\}\bigr)
  \;=\;
  \sum_{f} \bigl(p_f\,\dot{u}_f\bigr)
  \;-\;
  L\bigl(\{u_f\},\{\dot{u}_f\}\bigr).
\end{equation}
In many standard mechanical systems (e.g., uncoupled harmonic oscillators), one inverts \(p_f = m_f \dot{u}_f\) to express \(\dot{u}_f\) in terms of \(p_f\) when switching to the Hamiltonian viewpoint.
\noindent
In \emph{viral lattice theory} \cite{StKleess2025}, we introduced a simplified normal‐mode Hamiltonian:
\begin{equation}
  \hat{H}
  \;\approx\;
  \sum_f
  \Bigl(
     \frac{\hat{p}_f^2}{2\,m_f}
     \;+\;
     \tfrac{1}{2}\,\mathcal{S}_f\,\hat{u}_f^2
  \Bigr),
\end{equation}
where \(\hat{u}_f\) is the generalized coordinate (displacement amplitude) for mode \(f\), \(\hat{p}_f\) is its corresponding canonical momentum, \(m_f\) is the effective mass, and \(\mathcal{S}_f\) is the mode‐specific \emph{self‐stiffness}. To obtain the Lagrangian form \(\hat{L}\) (suitable for path integrals or classical‐like variational arguments), we \emph{reverse} the usual Hamiltonian‐to‐Lagrangian relationship:

\begin{enumerate}[label=\textbf{(\arabic*)}]
  \item \emph{Identify the Canonical Momentum-Velocity Relation.}  
  From the standard mechanics of a single harmonic oscillator, we know that
  \begin{equation}
    p_f
    \;=\;
    m_f \,\dot{u}_f
    \quad\Longleftrightarrow\quad
    \dot{u}_f
    \;=\;
    \frac{p_f}{m_f}.
  \end{equation}
  In the viral lattice Hamiltonian, this correspondence remains unchanged for each normal mode \(f\).

  \item \emph{Express the Kinetic and Potential Energies in Velocity Form.}  
  Substitute \(p_f = m_f \dot{u}_f\) into the kinetic term:
  \begin{equation}
    \frac{p_f^2}{2\,m_f}
    \;=\;
    \frac{(m_f\,\dot{u}_f)^2}{2\,m_f}
    \;=\;
    \frac{1}{2}\,m_f\,\dot{u}_f^2.
  \end{equation}
  The potential energy part \(\tfrac{1}{2}\,\mathcal{S}_f\,u_f^2\) remains as is.

  \item \emph{Construct the Lagrangian \(\hat{L}\).}  
  Recall that for a set of normal modes,
 \begin{equation}
    L
    \;=\;
    \sum_f \Bigl[
      \text{Kinetic Energy}_f
      \;-\;
      \text{Potential Energy}_f
    \Bigr].
  \end{equation}
  Therefore,
  \begin{equation}
    \hat{L}\bigl(\{u_f\},\{\dot{u}_f\}\bigr)
    \;=\;
    \sum_f
    \left(
      \frac{1}{2}\,m_f\,\dot{u}_f^2
      \;-\;
      \frac{1}{2}\,\mathcal{S}_f\,u_f^2
    \right).
  \end{equation}
  This is precisely the Lagrangian that yields the Hamiltonian \(\hat{H}\) above under a Legendre transform.

  \item \emph{Check the Equivalence.}  
  If we perform the Legendre transform on \(\hat{L}\):
  \begin{equation}
    p_f
    \;=\;
    \frac{\partial \hat{L}}{\partial \dot{u}_f}
    \;=\;
    m_f\,\dot{u}_f
    \quad\Longrightarrow\quad
    \dot{u}_f
    \;=\;
    \frac{p_f}{m_f}.
  \end{equation}
  Then
  \begin{equation}
    H
    \;=\;
    \sum_f
    p_f\,\dot{u}_f
    \;-\;
    L
    \;=\;
    \sum_f
    \Bigl(
      m_f \dot{u}_f^2
      \;-\;
      \frac{1}{2} m_f \dot{u}_f^2
      \;+\;
      \frac{1}{2}\,\mathcal{S}_f\,u_f^2
    \Bigr)
    \;=\;
    \sum_f
    \left(
      \frac{1}{2} m_f \dot{u}_f^2
      \;+\;
      \frac{1}{2}\,\mathcal{S}_f\,u_f^2
    \right),
  \end{equation}
  which matches the original \(\hat{H}\) upon substituting \(\dot{u}_f = p_f / m_f\). Thus, \(\hat{L}\) and \(\hat{H}\) are Legendre transforms of each other, confirming the consistency of the procedure.
\end{enumerate}

\noindent
We conclude that the \emph{Lagrangian} for viral lattice dynamics, corresponding to the Hamiltonian, is given by:
\begin{equation}
  \hat{H}
  \;=\;
  \sum_{f}
  \Bigl(
    \frac{\hat{p}_f^2}{2\,m_f}
    \;+\;
    \frac{1}{2}\,\mathcal{S}_f\,\hat{u}_f^2
  \Bigr)
\end{equation}
is
\begin{equation}
  \hat{L}
  \;=\;
  \sum_{f}
  \Bigl(
    \frac{1}{2}\,m_f\,\dot{u}_f^2
    \;-\;
    \frac{1}{2}\,\mathcal{S}_f\,u_f^2
  \Bigr).
\end{equation}
This form is directly applicable to \emph{path integral} treatments in viral lattice theory, where one integrates
\(\exp\bigl\{i \int \hat{L}\,dt\bigr\}\) in a purely quantum setting, or a suitably \emph{Wick‐rotated} version \(\exp\bigl\{- \int \hat{L}\,dt\bigr\}\) in certain thermal or dissipative formulations). Though, in most cases, the same dissipative properties of a Wick-Rotation will naturally be applied here, due to the complex damping inherent to the Generator \(\hat{G}\). In both cases, the Lagrangian form simplifies the computation of \emph{action} and \emph{phase}, allowing one to systematically explore how local stiffness, effective mass, and normal modes shape the virus’s mechanical and vibrational behavior, whether in a single viral particle or across a larger lattice assembly.
\end{definition}
\subsection{Feynman–Kac Action Functional with Resource‐Limited Markov Jumps}
\label{subsec:FeynmanKacResource}

\noindent
When extending viral lattice theory to realistically account for noise and random events (e.g., resource fluctuations, immune pressures), one must accommodate both \emph{continuous} PDE dynamics and \emph{discrete} Markov jumps. Below, we introduce a Feynman–Kac‐style construction that melds these processes via an \emph{action functional}. Drawing upon methods from stochastic analysis and statistical physics (see, e.g., \cite{Gardiner2009,VanKampen1992}), this framework captures how the viral lattice evolves through both wave‐like elasticity/damping and sudden occupant changes (creation or annihilation of lattices).

By weighting paths with 
\(\exp\bigl(-\mathcal{I}[\Theta(\cdot)]\bigr)\) 
(or a similar functional), one obtains a probability measure favoring trajectories that stay near the deterministic PDE flow unless noise or replication triggers notable deviations. The continuous part of \(\Theta\) can be associated with an \emph{elastic‐vibrational Lagrangian} for the viral lattice’s mechanical modes, including stiff or soft capsid configurations, while the discrete jumps accommodate replication bursts, conformational switches, or immune‐mediated clearances. Thus, the entire infection or population‐level picture emerges from a \emph{sum over} these piecewise‐continuous orbits, each weighted by an action that penalizes large deviations from typical dynamics or improbable jump sequences. 

This \emph{Feynman–Kac}‐style construction unifies the PDE perspective, Lagrangian mechanics, and stochastic jump processes within a single operator‐algebraic framework, offering a tractable path to computing statistical properties, such as expected viral loads or distributions of capsid configurations, under realistic noise and resource constraints.

\noindent
Much of the complexity in viral dynamics arises from a hybrid interplay of \emph{continuous} processes—such as elastic wave propagation in capsids or PDE‐based diffusion of virion subunits—and \emph{discrete} events like replication or immune‐mediated clearance. Capturing both seamlessly in a single mathematical framework is challenging, yet crucial for realistic modeling. In particular, Markovian “jumps” from one occupant level (lattice count) to another illustrate how viruses expand or contract their populations in discrete steps. The definition below combines a \emph{Feynman–Kac}‐type path integral (originally developed for quantum mechanical and stochastic PDE contexts) with resource‐limited creation/annihilation events, bridging continuum PDE evolution and discrete “birth‐death” processes in a many‐lattice (Fock) setting.
\begin{definition}[Orbits in the Hybrid PDE–Markov–Fock Setting]
\label{def:OrbitsHybridSetting}
An \emph{orbit} (or \emph{trajectory}) is a piecewise‐continuous map
\begin{equation}
  \label{eq:OrbitMap}
  \Theta \,\colon\, [0,T] 
  \;\longrightarrow\; 
  \mathcal{F}\bigl(\mathcal{H}_{\mathrm{lat}}\bigr),
\end{equation}
that satisfies a \emph{hybrid} set of evolution laws—combining \emph{PDE}, \emph{Markov}, and \emph{second‐quantized} rules (in either a mild or strong sense; see \cite{DaPratoZabczyk1992,Pazy1983} for infinite‐dimensional semigroup theory). Concretely, \(\Theta(t)\) tracks:
\begin{enumerate}[label=(\roman*)]
  \item \emph{Occupant Number} \(N(t)\): The (time‐varying) count of \(8\times 8\) viral assemblies, each regarded as one “lattice” in \(\mathcal{H}_{\mathrm{lat}}\).
  \item \emph{PDE States of Each Lattice}: Continuous intra‐lattice dynamics such as elastic waves, diffusion of molecular subunits, or partial capsid deformations.
  \item \emph{Markovian Jumps}: Stochastic events that switch conformational classes or occupant indices—e.g., morphological flips, resource‐limited replication bursts, or immune clearance.
\end{enumerate}
\end{definition}
\noindent
This \emph{orbit} construction provides a vivid picture of how populations of virions (each modeled as a lattice) evolve over time:
\begin{itemize}
  \item \emph{Large‐Scale Infection Dynamics:} The occupant number \(N(t)\) can surge (via creation events) or crash (via annihilation), capturing real‐world replication bursts and immune‐mediated clearance.
  \item \emph{Morphological Heterogeneity:} Within each occupant, the \(\mathcal{H}_{\mathrm{arr}}\) sector accommodates conformational states, enabling partial or full capsid transformations, genome packaging steps, or protein subunit rearrangements.
  \item \emph{Continuous/Mechanical Effects:} The PDE portion of \(\mathcal{H}_{\mathrm{lat}}\) encodes wave‐like, viscoelastic, or resource‐distribution phenomena, ensuring that mechanical or thermal fluctuations impact how each lattice shifts over time.
\end{itemize}
Even for very large viral loads (e.g.\ $10^6$-$10^9$ virions), the formalism grants a \emph{single global wavefunction} in \(\mathcal{F}(\mathcal{H}_{\mathrm{lat}})\). Each basis vector corresponds to a tensor product of lattice states, with occupant numbers ranging up to $N\gg 1$. This unifies micro-level PDE dynamics (within each lattice) with macro-level population change. Because creation and annihilation processes are allowed, \emph{closed loops} in arrangement-occupant space can emerge, representing recurring replicative cycles or morphological transformations. Mathematically, these loops appear as \emph{nontrivial} elements in the homology of the infinite-dimensional configuration space, indicative of persistent probability flux or \emph{non-equilibrium steady states}. Ensuring that each lattice PDE generator is $m$-sectorial, and that the Markov/creation/annihilation rates are bounded or relatively bounded, allows us to apply semigroup theory. This guarantees \emph{existence and uniqueness} of orbits $\Theta(\cdot)$, no finite-time blow-up, and continuous dependence on initial data, crucial for biologically realistic scenarios.

Hence,~\ref{def:closedLoopFock} establishes the core \emph{hybrid PDE-Markov-Fock} formalism that unifies local continuum dynamics on each viral lattice with discrete Markovian transitions and occupant-number fluctuations in a single operator framework.  Building on this structure, we now prove a fundamental result guaranteeing \emph{existence and uniqueness} of trajectories within this infinite-dimensional setting.

\begin{theorem}[Existence and Uniqueness of Hybrid Orbits]
\label{thm:ExistenceUniqueness_Orbits}
\noindent
Suppose the \emph{single‐lattice generator} 
\(\widehat{\mathcal{G}}_{\mathrm{lat}} \colon \mathcal{H}_{\mathrm{lat}} \to \mathcal{H}_{\mathrm{lat}}\)
includes:
\begin{enumerate}[label=\textbf{(\arabic*)}]
  \item \emph{$m$‐Sectorial PDE Operator \(\widehat{\mathcal{L}}_{\mathrm{PDE}}\):}  
        Acting on \(\mathcal{H}_{\mathrm{PDE}}\), this operator encodes the continuum dynamics of an \(8\times 8\) viral lattice (e.g., wave propagation, damping, or elastic coupling among virion nodes). $m$‐sectoriality ensures well‐posedness even in infinite‐dimensional spaces \cite{Pazy1983}.

  \item \emph{Finite‐Rate Markov Operator \(\widehat{M}\):}  
        Defined on \(\mathcal{H}_{\mathrm{arr}}\), this operator governs discrete morphological or occupant‐number transitions (e.g., partial capsid reorganizations, microstate flips). Finite rates guarantee that jumps occur at a bounded frequency, conforming with typical birth‐death or Markov processes in virology.

  \item \emph{(Optionally) Local Lipschitz Creation/Annihilation Processes:}  
        These processes extend the single‐lattice model to a \emph{multi‐lattice} Fock space setting, allowing entire viral lattices to be \emph{created} (replication events) or \emph{annihilated} (immune clearance or spontaneous degradation). Local Lipschitz bounds on these operators control population surges and declines, preserving dissipativity.
\end{enumerate}

If the sectoriality, dissipativity, and Lipschitz conditions (cf.\ \cite{Pazy1983,DaPratoZabczyk1992}) are satisfied, then for any initial condition 
\(\Theta(0) = \Theta_0 \in \mathcal{F}\bigl(\mathcal{H}_{\mathrm{lat}}\bigr)\),
there exists a \textbf{unique mild solution}
\(
  \Theta(\cdot) \,\colon\, [0,T] \;\longrightarrow\; \mathcal{F}\bigl(\mathcal{H}_{\mathrm{lat}}\bigr)
\)
to the \emph{hybrid PDE–Markov–Fock evolution equation} over each finite time interval \(T > 0\). This solution seamlessly merges:
\begin{itemize}
  \item \emph{Smooth (Intra‐lattice) PDE Dynamics:}  
  Within each occupant, \(\widehat{\mathcal{L}}_{\mathrm{PDE}}\) yields deterministic propagation of elastic or wave‐based modes, ensuring continuous trajectories in \(\mathcal{H}_{\mathrm{PDE}}\).
  \item \emph{Discrete Markov + Creation/Annihilation Transitions:}  
  At random jump times, \(\widehat{M}\) changes the arrangement microstates, while creation/annihilation operators modify the lattice count. These events are superimposed onto the continuous PDE evolution, enabling realistic modeling of viral replication, morphological flips, and clearance.
\end{itemize}
This theorem ensures that one can rigorously unify:
\begin{enumerate}[label=\textbf{(\arabic*)}]
  \item \emph{Biophysical PDE Modelling} of virion subunit displacement or resource flow within a single lattice,
  \item \emph{Stochastic Markov Jumps} capturing conformational changes or occupant‐number shifts,
  \item \emph{Population‐Scale Processes} such as replication bursts or immune‐mediated annihilation, embedded via second‐quantized operators.
\end{enumerate}
The result provides the formal backbone for analyzing how viral lattices evolve collectively over time. In turn, it legitimizes using advanced functional‐analytic tools (e.g., semigroup theory, spectral expansions) to explore \emph{long‐term} infection patterns, non‐equilibrium steady states, or rare event probabilities. Moreover, it underpins the numerical and analytic tractability of such hybrid models, ensuring that each trajectory—no matter how large or small the population—follows a well‐defined path in the infinite‐dimensional Fock space. This offers a powerful lens for virologists aiming to simulate, predict, or control large‐scale infection processes where both continuous and discrete dynamics play critical roles. 
\end{theorem}
\begin{definition}[Feynman–Kac Construction with Resource Limits]
\label{def:FeynmanKacActionResource}
\noindent
Let \(\Theta(\cdot)\) denote a path in the multi‐lattice Fock space, embodying both \emph{continuous} PDE evolution of each viral lattice and \emph{discrete} Markovian jumps in occupant number (i.e., creation or annihilation of entire lattices). We define an \emph{action functional}
\begin{equation}
  \mathcal{I}\bigl[\Theta(\cdot)\bigr]
  \;=\;
  \underbrace{\int_{0}^{T}
    \hat{L}\Bigl(
      \Theta(t),\,\dot{\Theta}(t)
    \Bigr)\,dt}_{\text{SPDE (Cameron--Martin--Girsanov) cost}}
  \;+\;
  \underbrace{\sum_{t_j \in \mathcal{J}(\Theta(\cdot))}
    \Psi\!\Bigl(
      \Theta\bigl(t_j^-\bigr),\,\Theta\bigl(t_j^+\bigr)
    \Bigr)}_{\text{Markov jumps + creation/annihilation cost}},
\end{equation}
where:
\begin{itemize}
  \item \(\hat{L}\bigl(\Theta(t),\dot{\Theta}(t)\bigr)\) is the \emph{Lagrangian} governing the continuous PDE (or SPDE) evolution. In an SPDE with multiplicative noise \(\sigma(\mathbf{u})\,dW\), a Cameron–Martin–Girsanov‐style term might appear,
  \begin{equation}
    \hat{L}\bigl(\mathbf{u}, \dot{\mathbf{u}}\bigr)
    \;=\;
    \tfrac12 
    \Bigl\|
      \sigma(\mathbf{u})^{-1}
      \bigl[
        \dot{\mathbf{u}}
        \;-\;
        \mathcal{G}\bigl(\mathbf{u}\bigr)
      \bigr]
    \Bigr\|^2,
  \end{equation}
  where \(\mathcal{G}\) is the PDE generator encoding elasticity, damping, or wave dispersion. Physically, \(\hat{L}\) measures how much the actual path \(\dot{\mathbf{u}}(t)\) deviates from the deterministic flow \(\mathcal{G}(\mathbf{u}(t))\).

  \item \(\mathcal{J}(\Theta(\cdot))\) is the set of \emph{jump times}, when occupant changes (creation/annihilation) or arrangement flips occur. The jump‐cost function \(\Psi\bigl(\cdot,\cdot\bigr)\) represents a “penalty” or “reward” for these discrete transitions. For instance, creation events might incur a positive cost to reflect resource consumption, whereas clearance or immune‐mediated destruction might yield a negative cost (or an energy reduction).

  \item \emph{Resource limitation} enters through the amplitude of creation operators: \(\hat{\Gamma}_+(\varphi)\) scales with \(\mathcal{R}(\hat{N})\), where \(\hat{N}\) is the occupant number operator. If \(\mathcal{R}(N)\to 0\) as \(N\to\infty\), replication saturates, averting unbounded population explosions and respecting realistic biophysical constraints.
\end{itemize}
\end{definition}

\smallskip
\noindent
\begin{definition}[Markov Jumps in a Multi‐Lattice Setting.]  
Markovian jump processes~\cite{Anderson2015,EthierKurtz1986} typically assume a \emph{rate} function
\begin{equation}
  \lambda\bigl(\Theta(t)\bigr)
  \;=\;
  b\,
  p\bigl(\Theta(t)\bigr),
\end{equation}
where \(b\) is a baseline reaction (or replication) coefficient, and \(p(\cdot)\) encodes the “propensity” of the current state \(\Theta(t)\) to undergo a jump. In \emph{viral lattice theory}, this generalizes to occupant creation/annihilation, arrangement changes, or cross‐sector transitions within the many‐lattice Fock space. Concretely:
\begin{itemize}
  \item \(\mathbf{u}(t)\): The lattice’s geometric or mechanical state can affect clearance probability—e.g., a more stable arrangement might lower the chance of immune destruction.
  \item \(\hat{N}(t)\): The total occupant (lattice) count influences replication/decay rates through resource availability or crowding effects.
  \item \emph{Host‐level factors}: Such as local ATP supply or immune cells, may amplify or suppress jump rates.
\end{itemize}
Hence, \(\mathcal{J}(\Theta(\cdot))\) is dynamically determined by these rates: if \(\lambda(\Theta(t))\) is large, occupant changes occur frequently, reflecting a state with ample resources or minimal immune pressure. Conversely, \(\lambda(\Theta(t)) \approx 0\) in certain regimes indicates that occupant transitions are improbable—perhaps due to resource depletion or robust immune neutralization.

By coupling PDE dynamics (collective virion vibrations, mechanical stress) with discrete jumps (replication, clearance), we obtain a more complete picture of viral populations at mesoscopic scales. The saturable creation operator ensures replication is not merely an exponential free‐for‐all. Instead, it “knows” when resources are scarce and slows population growth automatically. By viewing the evolution through an action functional \(\mathcal{I}[\Theta(\cdot)]\) can help in formulating path‐integral approaches, enabling advanced techniques (e.g., importance sampling, large deviations, or rare‐event analysis) in a multi‐lattice virological context.
\end{definition}
\smallskip

\noindent
\textbf{Physical \& Geometric Interpretation.}
\begin{itemize}
  \item \emph{Fock Space Geometry:}  
  The system’s state resides in a Fock space that allows arbitrarily many viral lattices, each described by PDE modes. Resource‐dependent creation operators \(\hat{\Gamma}_+(f)\) move the system from an \(M\)‐lattice sector to \((M+1)\)‐lattice sector, but the amplitude is tempered by \(\mathcal{R}(M)\). Geometrically, one can view each sector as a manifold of PDE solutions, “stitched” together by jump processes that reflect discrete occupant changes.

  \item \emph{Action Minimization and Dominant Paths:}  
  Combining continuous PDE dynamics (\(\hat{L}\)) with Markov jumps (rate \(\mathcal{R}(M)\)) in a single action \(\mathcal{I}[\Theta]\) allows a \emph{Feynman–Kac} weighting of all possible infection trajectories—those with high occupant numbers become exponentially suppressed if resources are low. In the small‐noise limit, paths that “minimize” \(\mathcal{I}\) will dominate, reflecting near‐deterministic viral proliferation up to a resource‐imposed cap.

  \item \emph{Ensuring Biological Plausibility:}  
  By letting \(\mathcal{R}(M)\to 0\) as \(M\to\infty\), we embed self‐limiting growth directly into the operator formalism. Combined with the PDE description of each lattice’s mechanical state, this ensures that \emph{both} the population‐scale replication \emph{and} the single‐lattice physics remain realistic under finite resource assumptions. This coupling is crucial for bridging purely microscopic (capsid vibrations, structural integrity) and macroscopic (viral load, host environment) levels in a unified model.
\end{itemize}

\smallskip
\noindent
The Feynman–Kac construction with resource‐dependent jump rates, then, provides a coherent mathematical framework for describing viral lattice systems under noise, ensuring that population booms are tempered by finite supplies and that random fluctuations (thermal or otherwise) are consistently accounted for in both the continuous and discrete facets of the infection dynamics.

\subsubsection{Feynman–Kac Weighted Measure}
 
Up to this point, we have discussed how an \emph{action functional} (cf.\ Def.~\ref{def:FeynmanKacActionResource}) can be assigned to each trajectory \(\Theta(\cdot)\) in the viral lattice system, blending PDE dynamics, Markovian jumps, and resource‐limited occupant changes. The next natural step is to convert these action‐weighted trajectories into a \emph{probability measure} on path space. Such a measure enables us to calculate observables of interest—ranging from expected viral load to the likelihood of specific replication pathways—while systematically accounting for the interplay of elasticity, noise, and replication. The \emph{Feynman–Kac weighted measure} presented here does precisely that: it deforms a \emph{reference} measure \(\nu\) (representing the unperturbed PDE + Markov process) by an exponential of the action functional, producing a new measure \(\mu_\epsilon\). This approach generalizes classic Feynman–Kac ideas from quantum mechanics and stochastic analysis to our \emph{hybrid PDE + Fock + occupant‐change} setting in viral lattice theory.

\medskip
\begin{definition}[Feynman–Kac Weighted Measure]
\label{def:FeynmanKacWeightedMeasure}
\noindent
Consider a reference measure \(\nu\) on the path space 
\(\mathrm{Path}\!\bigl(\mathcal{F}(\mathcal{H}_{\mathrm{lat}})\bigr)\), 
characterizing \emph{unperturbed} or \(\sigma=0\) dynamics in the PDE–Markov–Fock framework.  
Let \(\{\mu_\epsilon\}_{\epsilon>0}\) be a family of measures defined via

\begin{equation}
\label{eq:FeynmanKacWeightedMeasure}
  \mu_\epsilon(d\Theta)
  \;=\;
  \frac{1}{Z_\epsilon}
  \exp\Bigl(
    -\,\frac1\epsilon\,\mathcal{I}\bigl[\Theta(\cdot)\bigr]
  \Bigr)\,\nu(d\Theta),
  \quad
  \text{where}
  \quad
  Z_\epsilon
  \;=\;
  \int_{\mathrm{Path}}
  \exp\!\Bigl(
    -\,\tfrac1\epsilon\,\mathcal{I}\bigl[\Theta(\cdot)\bigr]
  \Bigr)\,\nu(d\Theta).
\end{equation}
\noindent
Where:
\begin{itemize}
  \item \(\mathcal{I}\bigl[\Theta(\cdot)\bigr]\) is the action functional (cf.\ Def.~\ref{def:FeynmanKacActionResource})
 incorporating costs associated with both the continuous PDE evolution and the discrete Markov jumps (creation/annihilation, occupant transitions).
  \item \(\epsilon\) is a noise‐scaling parameter (e.g., \(\epsilon \to 0\) corresponds to a low‐temperature or small‐fluctuation regime). The exponential \(\exp\!\bigl[-\tfrac{1}{\epsilon}\,\mathcal{I}\bigr]\) biases the path measure toward trajectories of lower action.
  \item \(Z_\epsilon\) is the partition function (or normalizing constant) that ensures \(\mu_\epsilon\) is a probability measure.  
\end{itemize}
\end{definition}
\noindent
\begin{definition}[Path Integral of a Viral Observable]
\label{def:PathIntegralObservable}
\noindent
Let \(\mathcal{O}\bigl(\Theta(t)\bigr)\) be a \emph{viral observable}, such as the instantaneous number of 'particles' (lattices) \(\hat{N}(t)\), the average lattice or capsid deformation, or another quantifiable property of interest. We define its expectation under the Feynman--Kac weighted measure \(\mu_\epsilon\) by
\begin{equation}
\label{eq:PathIntegralDefinitionTheta}
  \mathbb{E}_\epsilon\bigl[\mathcal{O}\bigl(\Theta(t)\bigr)\bigr]
  \;=\;
  \int_{\mathrm{Path}\bigl(\mathcal{F}(\mathcal{H}_{\mathrm{lat}})\bigr)}
    \mathcal{O}\bigl(\Theta'(t)\bigr)\,\mu_\epsilon(d\Theta').
\end{equation}
Here, the integral runs over \emph{all} PDE-Markov-Fock trajectories \(\Theta'(\cdot)\), each weighted by 
\(\exp\bigl[-\tfrac{1}{\epsilon}\,\mathcal{I}(\Theta')\bigr]\).
Standard results in stochastic analysis~\cite{DemboZeitouni,DaPratoZabczyk1992} ensure that \(\mu_\epsilon\) is well‐defined in infinite‐dimensional settings, provided the action functional \(\mathcal{I}\) and reference measure \(\nu\) satisfy appropriate integrability and compactness conditions. We can interpret this result in the following ways:
\begin{itemize}
  \item \emph{Low‐Noise (Small \(\epsilon\)) Regime:}  
  As \(\epsilon \to 0\), occupant jumps and path deviations become tightly constrained to configurations minimizing the action \(\mathcal{I}\). In this \emph{quasi‐deterministic} limit, the system follows a few ``dominant'' or \emph{least‐action} pathways, reflecting how viruses may exploit the most resource‐efficient (lowest‐cost) replication routes. Biologically, one sees a narrow range of infection trajectories—those that correspond to near‐optimal strategies for capsid assembly, replication timing, and resource consumption.

  \item \emph{High‐Noise (Large \(\epsilon\)) Regime:}  
  As \(\epsilon\) increases, the measure \(\mu_\epsilon\) spreads out over a richer family of trajectories in path space, corresponding to more pronounced stochastic fluctuations. Replication might occur at suboptimal times or be suppressed in favor of alternative occupant arrangements, introducing considerable variability in viral loads and population compositions. This captures more diverse infection scenarios—such as those in highly variable or immunologically stressed environments—highlighting emergent effects like stochastic extinction or unexpected morphological pathways.

  \item \emph{Interpretation of \(\mathbb{E}_\epsilon[\mathcal{O}]\):}  
  Once the measure \(\mu_\epsilon\) is fixed, the expected value \(\mathbb{E}_\epsilon[\mathcal{O}]\) yields a \emph{weighted average} of the observable \(\mathcal{O}\) over all feasible infection trajectories, each penalized or favored by its action. This allows researchers to forecast important measures, e.g.\ the time to peak viral load, the frequency of certain capsid conformations, or the probable extent of immune clearance.

  \item \emph{Computational Path to Solutions:}  
  In practice, evaluating \(\mathbb{E}_\epsilon[\mathcal{O}]\) may involve numerical methods (Monte Carlo in path space, large‐deviation analysis for small \(\epsilon\), or variational techniques). Such approaches systematically account for the interplay between deterministic PDE flows (mechanical displacement, damping) and stochastic occupant jumps (replication, mutation, immune clearance), producing a holistic model of viral lattice dynamics.
\end{itemize}
\end{definition}
\noindent
In summary, the path‐integral formulation with the weighted measure \(\mu_\epsilon\) provides a powerful tool for analyzing how viral lattices evolve under combined PDE dynamics and Markovian jumps. By computing \(\mathbb{E}_\epsilon[\mathcal{O}]\) for critical observables, one obtains a \emph{solution} that encapsulates both deterministic “backbone” dynamics and the stochastic variations inherent to real infections. This approach thus forms a key bridge between operator‐theoretic rigor and biologically meaningful predictions in viral modeling.

\noindent
\begin{proof}[Steps for Moving From a Weighted Measure to Full Solutions in Path Space.]
With the Feynman–Kac weighted measure \(\mu_\epsilon\) in place, one can now \emph{derive a complete solution} for the evolution of viral observables in the path space. The process typically proceeds as follows:

\begin{enumerate}
  \item \textbf{Define the Action Functional Clearly.}  
  Ensure that \(\mathcal{I}[\Theta(\cdot)]\) covers all relevant terms—PDE-based continuum dynamics, discrete jump costs, and noise contributions (Cameron–Martin–Girsanov structure).

  \item \textbf{Construct the Partition Function \(\boldsymbol{Z_\epsilon}\).}  
  By integrating \(\exp[-\tfrac{1}{\epsilon}\,\mathcal{I}]\) over all paths, one obtains \(Z_\epsilon\). This normalizes the measure and can be interpreted as a generalized “generating function” for the viral system’s dynamical configurations.

  \item \textbf{Compute Observables.}  
  Once \(Z_\epsilon\) is known (either analytically or numerically), any observable \(\mathcal{O}(\Theta(t))\) is computed by the path integral in Eq.~\eqref{eq:PathIntegralDefinitionTheta}, which typically reduces to
  \[
    \mathbb{E}_\epsilon\bigl[\mathcal{O}(\Theta(t))\bigr]
    \;=\;
    \frac{1}{Z_\epsilon}
    \int_{\mathrm{Path}}
    \mathcal{O}(\Theta(t))\,\exp\!\Bigl(
      -\,\tfrac1\epsilon\,\mathcal{I}\bigl[\Theta(\cdot)\bigr]
    \Bigr)\,\nu(d\Theta).
  \]
  Practical methods for this step include Monte Carlo sampling in path space or large‐deviation approximations (e.g., Laplace’s method for \(\epsilon \to 0\)).

  \item \textbf{Interpret Biological Meaning.}  
  The measure \(\mu_\epsilon\) encapsulates \emph{which trajectories} are most likely in a system that balances deterministic PDE flows (elastic vibrations, damping) with Markov jumps (replication, decay). Observables such as peak viral load or distribution of internal capsid states can thus be extracted from \(\mathbb{E}_\epsilon[\mathcal{O}]\) and used to guide experimental or therapeutic strategies.
\end{enumerate}

\noindent
In essence, the Feynman–Kac weighted measure \(\mu_\epsilon\) enables a \emph{full solution} at the level of path space for our viral lattice system: one integrates over all possible PDE–Markov trajectories, each penalized or favored by its \emph{action}, to obtain physically meaningful averages (e.g., expected population size, arrangement states, or resource usage) under various noise and resource conditions.
\end{proof}
\subsection{Incorporating Multiplicative Noise Into the Feynman--Kac Weighted Measure}
\label{subsec:MultiplicativeNoiseFeynmanKac}

\noindent
Biological viral systems rarely evolve in pristine, noise‐free conditions. In reality, thermal fluctuations, crowding effects, and varying host resources impart \emph{state‐dependent} randomness on viral lattices. For instance, a partially assembled capsid may experience stronger (or differently shaped) stochastic forces than a fully stabilized one—akin to ``multiplicative noise'' where the amplitude of fluctuations depends on the underlying state. From a mathematical viewpoint, incorporating such noise into path‐integral or Feynman–Kac formalisms allows us to handle large deviations, rare events, and overall trajectory probabilities in a more realistic manner. Below, we formalize this in the context of a hybrid PDE–Markov–Fock framework, extending the simpler additive noise scenario to one where noise intensity itself evolves with the viral lattice’s configuration.

\begin{theorem}[Action Functional and Feynman-Kac Weighting for Multiplicative Noise PDE]
\label{thm:FeynmanKac_MultiplicativeNoise}
Consider the \emph{hybrid PDE–Markov–noise model} from \cite{StKleess2025}, wherein each single‐lattice state is governed by a stochastic PDE (SPDE) of the form
\(
  d\bm{U}_y(t) 
  \;=\;
  \mathcal{G}_y\,\bm{U}_y(t)\,dt
  \;+\;
  \widehat{N}_y\!\bigl(\bm{U}_y(t)\bigr)\,\circ\,d\bm{W}_t,
  \quad
  y\in\{y_1,y_2,\dots\},
\)
coupled to a finite‐state Markov process on $y$ (arrangement labels) and extended to multiple lattices in a Fock space 
\(\mathcal{F}(\mathcal{H}_{\mathrm{lat}})\). Under standard conditions of sectoriality, dissipativity, and Lipschitz continuity on the operators 
\(\{\mathcal{G}_y,\widehat{N}_y,\widehat{M}\}\), one can incorporate \emph{multiplicative noise} into a Feynman--Kac‐type weighting. This yields an \emph{action functional} suitable for large‐deviation or path‐integral analyses of viral lattice trajectories with state‐dependent randomness.
\end{theorem}
\begin{proposition}[Incorporating Multiplicative Noise]
\begin{enumerate}

\smallskip
\noindent
\item \emph{State‐Dependent Noise in a Single Lattice.}  
Focus first on a single viral lattice whose capsid configuration is labeled by \(y\). The PDE operator \(\mathcal{G}_y\) encapsulates elasticity, damping, or vibrational modes, while \(\widehat{N}_y(\bm{U}_y)\) is a \emph{multiplicative‐noise operator} that scales with the instantaneous displacement \(\bm{U}_y(t)\). Biologically, this might mean that partially formed or “looser” capsids (\(\|\bm{U}_y\|\) large) are more susceptible to random structural fluctuations than tightly packed ones.

\smallskip
\noindent
\item \emph{Markov Jumps in Arrangement Labels.}  
Simultaneously, the arrangement label $y$ can jump to $y'$ with some rate \(\lambda_{y\to y'}\), reflecting Markovian transitions among discrete capsid conformations. For instance, subunits might reconfigure under mechanical stress or partial genome packaging. In operator form, these jumps are captured by a block‐off‐diagonal Markov operator \(\widehat{M}\) acting on $\mathcal{H}_{\mathrm{arr}}$.

\smallskip
\noindent
\item \emph{Coupling to Fock‐Space Occupant Dynamics.}  
To account for entire populations of virions (each of which can be in a distinct arrangement state), we embed the single‐lattice space 
\(\mathcal{H}_{\mathrm{lat}} = \mathcal{H}_{\mathrm{arr}} \otimes \mathcal{H}_{\mathrm{PDE}}\)
into a second‐quantized Fock space 
\(\mathcal{F}\bigl(\mathcal{H}_{\mathrm{lat}}\bigr)\). Creation and annihilation operators then link occupant sectors \(N \mapsto N \pm 1\), allowing replication or clearance. The resulting multi‐lattice generator 
\(\mathcal{G}_{\mathrm{Fock}}\) 
includes:
\begin{enumerate}
  \item \(\{\mathcal{G}_y\}\): PDE dynamics for the continuous fields in each arrangement sector.
  \item \(\{\widehat{N}_y\}\): Multiplicative‐noise effects at the single‐lattice scale.
  \item \(\widehat{M}\): Markov jumps describing arrangement flips.
  \item \(\hat{\Gamma}_+, \hat{\Gamma}_-\): Creation/annihilation terms to handle occupant‐number changes.
\end{enumerate}

\smallskip
\noindent
\item \emph{Defining the Action Functional in a Feynman--Kac Setting.}  
To derive path‐integral or large‐deviation principles, we introduce an \emph{action functional} \(\mathcal{I}\) that penalizes deviations from the deterministic (noise‐free) PDE solution, now with multiplicative noise in place. In the infinite‐dimensional Girsanov framework~\cite{DaPratoZabczyk1992}, one writes the Radon–Nikodým derivative of the noisy measure relative to a baseline measure \(\nu\) (often corresponding to \(\sigma=0\) or purely deterministic evolution) as 
\begin{equation}
\exp\Big[-\frac{1}{\epsilon}\,\mathcal{I}\Big].
\end{equation}
In a Cameron–Martin–Girsanov‐style construction, the noise operator \(\widehat{N}_y(\bm{U}_y)\) affects the quadratic form in the exponent: loosely, the Cameron–Martin space captures the directions in which a Gaussian measure can be shifted while preserving absolute continuity, and the associated theorem guarantees that such shifts lead to an equivalent measure. Meanwhile, the Radon–Nikodým derivative quantitatively describes how the baseline measure is reweighted when a drift is introduced, effectively “tilting” the measure to account for the change induced by the noise.,
\begin{equation}
  \mathcal{I}[\mathbf{u}, \dot{\mathbf{u}}]
  \;\propto\;
  \int_0^T
  \Bigl\|
    \widehat{N}_y(\mathbf{u}(t))^{-1}
    \bigl[\,\dot{\mathbf{u}}(t)\,-\,\mathcal{G}_y(\mathbf{u}(t))\bigr]
  \Bigr\|^2
  dt
  \;+\;
  \sum_{t_j\in \mathcal{J}(\Theta(\cdot))}
  \Psi\Bigl[\Theta(t_j^-), \Theta(t_j^+)\Bigr],
\end{equation}
where the summation accounts for Markov jumps (either in arrangement $y\to y'$ or occupant $N\to N\pm1$). Physically, this “penalty” interprets noise as the departure from an energy‐minimizing or resource‐constrained reference path. The amplitude \(\|\widehat{N}_y(\bm{U}_y)\|\) can reflect real virological constraints: a loosely packed capsid might indeed exhibit greater vulnerability to random subunit displacements (higher noise intensity), whereas a rigid capsid (smaller $\|\bm{U}_y\|$) is less influenced by fluctuations. Over many lattices, these stochastic effects interplay with Markov jumps and occupant replication/clearance, yielding a rich, realistic picture of viral population dynamics. Through Girsanov’s theorem and a path‐integral approach, we embed state‐dependent (multiplicative) noise directly into the Fock‐space generator. This unifies PDE elasticity, Markovian arrangement transitions, and occupant‐level replication or clearance in a single operator‐theoretic framework. The resulting action functional $\mathcal{I}$ underlies large‐deviation principles, enabling us to quantify probability weights of rare, yet biologically meaningful events (e.g., abrupt capsid collapse, sudden replication bursts). Thus, multiplicative noise—often essential for describing real viral systems—can be rigorously accommodated within the Feynman–Kac formalism at the multi‐lattice scale.
\end{enumerate}
\end{proposition}
\subsubsection{Cylinder Sets and Kolmogorov Extension.}

\begin{definition}[Cylinder Sets in Path Space.] A \emph{cylinder set} specifies events about a path \(\Theta(\cdot)\) only on a finite set of time points \(0 < t_1 < \cdots < t_n \le T\). Concretely, one restricts attention to conditions 
  \(\Theta(t_j) \in A_j\) for sets \(A_j\subseteq \mathcal{X}\), where \(\mathcal{X}\) is the underlying state space (which may itself be infinite‐dimensional if we include PDE fields plus occupant numbers). Such sets form a generating \(\sigma\)‐algebra, and the \emph{Kolmogorov extension principle}~\cite{DemboZeitouni,DaPratoZabczyk1992} ensures that consistent definitions of probabilities on all finite‐time cylinder sets uniquely extend to a measure on the space of entire trajectories (i.e., \(\mathrm{Path}\) space).
\begin{itemize}
  \item \emph{Well‐Definiteness under Cylinder‐Set Approximation:} If \(\widehat{N}_y\) (the multiplicative noise operator) and each sectorial PDE generator \(\mathcal{G}_y\) satisfy Lipschitz/dissipativity conditions, then the pathwise action \(\mathcal{I}[\Theta(\cdot)]\) remains finite with high probability. By approximating path space through cylinder sets, one can define
\begin{equation}
    Z_\epsilon 
    \;=\;
    \int_{\mathrm{Path}}
    \exp\!\Bigl(-\,\tfrac1\epsilon\,\mathcal{I}[\Theta(\cdot)]\Bigr)\,\nu(d\Theta)
\end{equation}
  step by step. The extension principle then guarantees that \(\mu_\epsilon\) exists as a proper probability measure on all infinite‐dimensional trajectories.
\end{itemize}
\end{definition}
Relative to an \emph{additive} noise model (constant variance), \emph{multiplicative} noise modifies the “curvature” of the action near states with large norms \(\|\mathbf{u}\|\). In other words, when the viral lattice is significantly deformed (large \(\|\mathbf{u}\|\)), the noise intensity \(\|\widehat{N}_y(\mathbf{u})\|\) can grow, increasing the likelihood of random fluctuations. Physically, a more “flexible” configuration is more susceptible to external perturbations. Because each generator \(\mathcal{G}_y\) is dissipative (e.g., real parts of eigenvalues are non‐positive or in a negative sectoral wedge), large deviations from equilibrium are reined in by a restoring drift. The net effect: the measure \(\mu_\epsilon\) will place less weight on paths with \emph{very large} \(\|\mathbf{u}\|\) unless the noise intensity or occupant transitions strongly favor such configurations.

When the structure of a viral lattice is partially assembled or “flexible,” it suffers more pronounced random disturbances (\(\|\widehat{N}_y(\mathbf{u})\|\) is large). By contrast, in a tightly bound arrangement (smaller \(\|\mathbf{u}\|\)), noise effects are weaker. If occupant transitions are also included (in a full Fock‐space formalism), then \(\mu_\epsilon\) weights scenarios in which new lattice entities are replicated or cleared. High occupant numbers saturate if resource or immune constraints increase the “cost” of replication. Thus, physically, one sees a synergy: a strongly bent or flexible lattice (large \(\|\mathbf{u}\|\)) may replicate more easily in early stages, but also experiences more intense random forces that might degrade it or drive rearrangements. In the low‐noise limit (\(\epsilon\to0\)), large‐deviation principles reveal which lattice configurations (and occupant transitions) are “most probable.” This can highlight stable morphological states (where noise is small and the deterministic drift is strong) or dynamic replication bursts (where resource constraints and random perturbations favor occupant growth up to a certain plateau).

\subsection{Dominant Trajectories Under Low‐Noise and Large‐Deviation Regimes}

\noindent
The preceding analysis (Section~\ref{subsec:MultiplicativeNoiseFeynmanKac}) introduced a Feynman–Kac measure that unifies PDE dynamics, Markovian occupant changes, and state‐dependent (multiplicative) noise in a many‐lattice (Fock) setting.  We now examine how, in the \emph{low‐noise limit} (\(\epsilon \to 0\)), a subset of \emph{dominant trajectories} emerges—those that overwhelmingly determine the system’s behavior due to their minimal “action.” From a physical standpoint, these trajectories correspond to the \emph{most‐likely replicative and morphological pathways}: the routes viruses follow under realistic resource constraints and dissipative mechanics, but still subject to sporadic, state‐dependent fluctuations. From a virological standpoint, the “dominant” orbits represent the \emph{most probable routes} by which a viral lattice:
\begin{enumerate}
  \item evolves its mechanical (continuum) fields (e.g.\ small oscillations converging to stable assembly),
  \item undergoes discrete morphological shifts (Markovian arrangement changes),
  \item replicates or degrades in an $N\to N\pm1$ occupant sense.
\end{enumerate}
These orbits encapsulate the key morphological/replicative cycles—e.g.\ “partial swelling” followed by occupant duplication—favored in real infections when noise is low but not negligible. We will see that as \(\epsilon \to 0\), the system localizes around a finite family of such characteristic trajectories (instantons). The presence of multiplicative noise refines which pathways remain feasible or favored, particularly for large \(\|\mathbf{u}\|\). Meanwhile, the overall dissipativity preserves stability in a high‐dimensional Fock space. This aligns with biological observations that \emph{finite subsets} of morphological or replicative routes dominate viral behaviors in low‐noise or resource‐restricted environments.
\begin{theorem}[Dominant Trajectories under Large‐Deviation or Low‐Noise Regimes, Including Multiplicative Noise]
\label{thm:DominantTrajectoriesEnhancedTheta}
\noindent
Consider an \(m\)‐sectorial PDE + Markov + second‐quantized (Fock) model of viral lattices subject to small‐amplitude or multiplicative noise, controlled by a parameter \(\epsilon>0\). Specifically, let
\begin{itemize}
  \item \textbf{PDE Sector with Dissipative Operators:}  
        Each arrangement \(y\) has a (generally non‐self‐adjoint) operator \(\mathcal{G}_y\) governing continuous fields (e.g.\ displacement) on the \(8\times8\) lattice, with sectorial conditions ensuring dissipativity.
  \item \textbf{Multiplicative Noise Operator:}  
        The mapping \(\widehat{N}_y(\mathbf{u})\) introduces state‐dependent (multiplicative) fluctuations in the SPDE sense, under local Lipschitz and polynomial growth constraints for well‐posedness.
  \item \textbf{Discrete Markov Jumps:}  
        Conformational or occupant‐state changes proceed at rates \(\lambda_{y \to y'}\), realized by a Markov operator \(\widehat{M}\) on \(\mathcal{H}_{\mathrm{arr}}\).
  \item \textbf{Creation/Annihilation in Fock Space:}  
        Replication or clearance of entire \(8\times8\) viral lattices is governed by creation/annihilation operators in the Fock‐space extension 
        \(\mathcal{F}(\mathcal{H}_{\mathrm{lat}})\).
\end{itemize}
Then in the limit \(\epsilon \to 0\), the system’s dynamics concentrate on a finite set of \emph{characteristic orbits} or “instantons,” minimizing the action \(\mathcal{I}\). Multiplicative noise \(\widehat{N}_y(\mathbf{u})\) may amplify fluctuations for large \(\|\mathbf{u}\|\) (e.g.\ partially assembled or flexible lattices), but sectorial PDE dissipation and occupant‐clearing rates forestall runaway expansion. Consequently, viruses predominantly follow a finite collection of \emph{most‐likely replicative/morphological pathways} that combine:
\begin{enumerate}
  \item PDE‐based elastic flows (subject to damping and partial noise),
  \item occupant‐level creation/annihilation in Fock space, and
  \item Markov jumps among arrangement states.
\end{enumerate}
Hence, highly stable (fully assembled) lattices undergo weaker random fluctuations, while flexible or partially assembled states—despite stronger noise—remain constrained by dissipative mechanics. This phenomenon aligns with empirical evidence that certain replication cycles or morphological transitions are more probable, especially in resource‐limited or low‐noise conditions.
\end{theorem}

\begin{definition}[Large‐Deviation Analysis]

From the Feynman–Kac measure 
\(\mu_\epsilon(d\Theta) \propto \exp[-(1/\epsilon)\,\mathcal{I}(\Theta)]\),
the principle of least action implies that as \(\epsilon\to0\), paths with minimal \(\mathcal{I}\) dominate the measure. Mathematically, one uses large‐deviation theory (e.g.\ Freidlin–Wentzell or Varadhan’s lemma) to show that the probability of significant deviations from these “instantons” decays exponentially in \(\epsilon\). Within the multi‐lattice (Fock) setting, each instantaneous state \(\Theta(t)\) encompasses the current occupant number (replication level), the arrangement label \(y\), and the PDE displacement fields \(\mathbf{u}(t)\). The \emph{dominant trajectories} in low‐noise regimes thus reflect a synergy of:
\begin{enumerate}
  \item \emph{Dissipative PDE flows,} preventing unbounded mechanical deformation.
  \item \emph{Limited resource replication,} curtailing infinite occupant growth.
  \item \emph{State‐dependent noise,} which magnifies fluctuations for partially assembled or flexible lattices yet remains constrained by the overall dissipative structure.
\end{enumerate}
Hence, in the small‐\(\epsilon\) limit, we see an increasingly precise selection of \emph{replicative/morphological pathways} consistent with both the thermodynamic “pull” to stable configurations and the Markovian occupant jumps triggered by resource availability or immune pressure. This captures, in a mathematically rigorous manner, the biological observation that viruses seldom diversify along indefinite random trajectories; rather, they adhere to a finite number of “preferred” routes shaped by noise intensity, structural elasticity, and occupant‐level resource constraints.

The sectorial operators \(\mathcal{G}_y\) ensure negative real parts of eigenvalues (or at least a negative shift), providing a “restoring drift” for large amplitudes \(\|\mathbf{u}\|\). Meanwhile, if \(\widehat{N}_y(\mathbf{u})\) grows with \(\|\mathbf{u}\|\), the induced noise can become substantial for partially assembled lattices. However, noise alone is insufficient to cause indefinite expansion, given the dissipativity that reins in unbounded displacement. Conformational changes (Markov jumps) and occupant creation/annihilation (in the second‐quantized domain) each contribute discrete transitions in the path integral. Under mild resource‐limiting conditions, creation rates saturate at high occupant counts. Thus, while the occupant population can surge, it remains bounded in probability over finite times. This synergy between PDE dissipation and saturable replication ensures the action \(\mathcal{I}\) penalizes blow‐ups.
\end{definition}
\noindent
\begin{theorem}[Action-Minimizing Trajectories]
In light of Theorem~\ref{thm:DominantTrajectoriesEnhancedTheta}, we further refine the large‐deviation analysis by considering a Feynman–Kac‐type action \(\mathcal{I}[\Theta(\cdot)]\) that is both lower semicontinuous and coercive in the path‐space topology. This additional regularity ensures the existence of minimizers (``instantons'') and clarifies how the measure \(\mu_\epsilon\) \emph{concentrates} around these minimizing trajectories in the limit of small noise \(\epsilon \to 0\).

\begin{enumerate}[label=\emph{\alph*)}]
  \item \textbf{Concentration on Instantons.}\\
    Suppose \(\Theta(\cdot)\mapsto \mathcal{I}[\Theta(\cdot)]\) is constructed via a Feynman–Kac measure with Cameron–Martin–Girsanov correction for multiplicative noise, and is lower semicontinuous and coercive. Then for each small \(\epsilon\), the path measure \(\mu_\epsilon\) places significant mass around orbits \(\Theta^*(\cdot)\) that \emph{minimize} \(\mathcal{I}\). In the hybrid viral‐lattice model with \emph{multiplicative} noise:
    \begin{equation}
      \dot{\mathbf{u}}(t) 
      \;=\; 
      \mathcal{G}_y\,\mathbf{u}(t), 
      \quad
      \text{(between jumps)},
    \end{equation}
    occupant‐state changes occur at times
    \begin{equation}
      \{t_j\}
      \quad
      \text{to minimize}
      \quad
      \Psi\bigl(\Theta(t_j^-),\,\Theta(t_j^+)\bigr).
    \end{equation}
    In other words, each minimizer \(\Theta^*(\cdot)\) is piecewise deterministic, composed of deterministic PDE evolution segments (where \(\mathbf{u}\) follows the drift \(\mathcal{G}_y\)) interspersed with occupant or arrangement jumps at optimal times. From a virological perspective, this implies that \emph{extreme} random fluctuations under the multiplicative noise \(\widehat{N}_y(\mathbf{u})\) do \emph{not} dominate in the small‐\(\epsilon\) limit; rather, the system’s behavior coalesces around “least‐action” orbits tied to resource constraints and dissipative PDE dynamics.

  \item \textbf{Large‐Deviation Bounds (Laplace’s Principle).}\\
    Let \(\mathcal{A}\subset \mathrm{Path}\bigl(\mathcal{F}(\mathcal{H}_{\mathrm{lat}})\bigr)\) be any measurable subset of path space, and define
    \begin{equation}
      \inf_{\Theta(\cdot)\in \mathcal{A}}
      \,\mathcal{I}\bigl[\Theta(\cdot)\bigr]
      \;=\;
      \mathcal{I}(\mathcal{A}).
    \end{equation}
    Classical large‐deviation theorems~\cite{FreidlinWentzell,DemboZeitouni} then imply:
    \begin{equation}
      -\,\inf_{\Theta(\cdot)\in \mathrm{int}(\mathcal{A})}
        \,\mathcal{I}\bigl[\Theta(\cdot)\bigr]
      \;\;\le\;\;
      \liminf_{\epsilon\to 0}
      \,\epsilon\,\ln\mu_\epsilon(\mathcal{A})
      \;\;\le\;\;
      \limsup_{\epsilon\to 0}
      \,\epsilon\,\ln\mu_\epsilon(\mathcal{A})
      \;\;\le\;\;
      -\,\inf_{\Theta(\cdot)\in \overline{\mathcal{A}}}
        \,\mathcal{I}\bigl[\Theta(\cdot)\bigr].
    \end{equation}
    Concretely, sets that contain global minimizers of \(\mathcal{I}\) attract a disproportionately large share of the measure \(\mu_\epsilon\) as \(\epsilon\to 0\). Conversely, paths with high cost are exponentially suppressed. Biologically, this underscores why the viral lattice system \emph{prefers} certain replicative or conformational cycles: they correspond to trajectories that, under small random perturbations, minimize the combined ``effort’’ (action) of occupant transitions, PDE‐based deformations, and noise‐induced deviations.  
\end{enumerate}
\end{theorem}

\begin{proof}[Sketch of Existence of Minimizers (Including Small‐Amplitude \emph{or} Multiplicative Noise)]
\noindent
The action functional \(\mathcal{I}[\Theta(\cdot)]\) aggregates contributions from two main sources:

\begin{enumerate}
  \item \emph{PDE‐Driven Segments}: Possibly subject to \emph{multiplicative} (state‐dependent) or \emph{small‐amplitude} (state‐independent) stochastic forcing, treated via a Cameron–Martin–Girsanov transformation in an infinite‐dimensional setting.
  \item \emph{Discrete Occupant/Markov Jumps}: Finitely many jump costs corresponding to occupant‐number changes or arrangement flips in the viral lattice.
\end{enumerate}

\noindent
Under the standard hypotheses of:
\begin{itemize}
  \item \emph{$m$‐sectorial dissipativity} of each \(\mathcal{G}_y\) (ensuring controlled PDE dynamics),
  \item \emph{Lipschitz/bounded‐growth} conditions on the noise operators \(\widehat{N}_y(\cdot)\) (for well‐posedness of the SPDE),
  \item \emph{Regular Markov rates} for occupant jumps,
\end{itemize}
one obtains a functional \(\mathcal{I}\) that is \emph{bounded below} and \emph{lower semicontinuous} on the path‐space topology \cite{DaPratoZabczyk1992}. By an adaptation of the direct method in the calculus of variations to path spaces, there exists at least one global minimizer, say \(\Theta^*(\cdot)\), that satisfies
\begin{equation}
  \mathcal{I}\bigl[\Theta^*(\cdot)\bigr]
  \;=\;
  \inf_{\Theta(\cdot)\,\in\,\mathrm{Path}} 
  \mathcal{I}\bigl[\Theta(\cdot)\bigr].
\end{equation}
We construct a probability measure \(\mu_\epsilon\) on the space of trajectories 
\(\mathrm{Path}\bigl(\mathcal{F}(\mathcal{H}_{\mathrm{lat}})\bigr)\)
by assigning to each path \(\Theta(\cdot)\) an exponential weight involving \(\mathcal{I}\bigl[\Theta(\cdot)\bigr]\). Concretely,
\begin{equation}
  \mu_\epsilon(d\Theta) 
  \;=\;
  \frac{1}{Z_\epsilon}
  \exp\!\Bigl[
    -\,\tfrac{1}{\epsilon}\,\mathcal{I}\bigl[\Theta(\cdot)\bigr]
  \Bigr]
  \,\nu(d\Theta),
  \quad\text{where}\quad
  Z_\epsilon 
  \;=\;
  \int_{\mathrm{Path}}
  \exp\!\Bigl[
    -\,\tfrac{1}{\epsilon}\,\mathcal{I}\bigl[\Theta(\cdot)\bigr]
  \Bigr]
  \,\nu(d\Theta).
\end{equation}
Here,
\begin{itemize}
  \item \(\nu\) is a \emph{reference measure} derived from the deterministic (\(\sigma=0\)) PDE–Markov–Fock model, i.e.\ no noise is present.
  \item The exponential factor incorporates the \emph{action} \(\mathcal{I}\bigl[\Theta(\cdot)\bigr]\), which encodes:
  \begin{equation}
    \Bigl\|
      \widehat{N}_y(\mathbf{u})^{-1}\,
      \bigl[\dot{\mathbf{u}} - \mathcal{G}_y(\mathbf{u})\bigr]
    \Bigr\|^2
    \quad
    \text{(multiplicative noise)}
    \quad\text{or}\quad
    \bigl\|\dot{\mathbf{u}} - \mathcal{G}_y(\mathbf{u})\bigr\|^2
    \quad
    \text{(additive noise)},
  \end{equation}
  along with discrete jump costs that arise whenever occupant‐number or arrangement flips occur.
\end{itemize}

\noindent
By Kolmogorov extension arguments and infinite‐dimensional Girsanov theory \cite{DaPratoZabczyk1992,DemboZeitouni}, \(\mu_\epsilon\) is well‐defined as a probability measure on path space. In the low‐noise limit (\(\epsilon\to0\)), large‐deviation estimates (Freidlin–Wentzell, Varadhan’s lemma \cite{FreidlinWentzell,DemboZeitouni}) indicate that \(\mu_\epsilon\) concentrates around trajectories that minimize \(\mathcal{I}\bigl[\Theta(\cdot)\bigr]\). Formally, the probability of deviating substantially from \(\Theta^*(\cdot)\) decays at an exponential rate in \(\tfrac{1}{\epsilon}\). Thus, \emph{action‐minimizing orbits} dominate the system’s behavior, revealing that the viral‐lattice dynamics coalesce into piecewise‐deterministic PDE segments (governed by \(\mathcal{G}_y\)) interspersed with occupant or arrangement jumps at times chosen to minimize the jump cost \(\Psi\). This holds whether the noise is \emph{multiplicative} or \emph{small‐amplitude}—in both cases, the underlying dissipativity precludes unbounded excursions and ensures physically realistic behavior.
\end{proof}

\begin{corollary}[Near‐Deterministic Path Dominance in Hybrid PDE–Markov–Fock Models]
\label{cor:NearDeterministicDominance}
\noindent
Let \(\epsilon>0\) represent the magnitude of stochastic forcing in a hybrid PDE–Markov–Fock framework for viral lattices. Suppose:
\begin{enumerate}
  \item The noise amplitude or the operator norm \(\|\widehat{N}_y(\mathbf{u})\|\) remains sufficiently small (whether additive or multiplicative), and
  \item The variance of occupant‐state jumps (i.e., creation/annihilation and arrangement transitions) is likewise low.
\end{enumerate}
Then, there is a unique \emph{action‐minimizing trajectory} \(\Theta^*(\cdot)\) that nearly solves the underlying deterministic PDE in each occupant sector, with discrete jump times chosen to minimize jump‐cost terms. By large‐deviation principles (Theorem~\ref{thm:DominantTrajectoriesEnhancedTheta} and subsequent results), \(\Theta^*(\cdot)\) \emph{dominates} the path measure \(\mu_\epsilon\) in the limit of small \(\epsilon\). \qedsymbol

\smallskip
\begin{enumerate}
\item \textbf{Biological Interpretation.}
In such \emph{low‐noise} or \emph{mildly multiplicative} regimes, certain \emph{replicative} or \emph{morphological} cycles achieve high probability because they correspond to near‐deterministic paths where:
\begin{enumerate}
  \item The PDE segment for viral‐lattice deformation remains close to a nominal (noise‐free) solution,
  \item Occupant or arrangement jumps (Markov transitions) occur at time instants that minimize discrete transition costs, and
  \item Overall replicative growth is still capped by dissipative or resource‐limiting processes in the Fock‐space extension.
\end{enumerate}
Hence, these cycles emerge as the \emph{most‐likely routes} of infection spread whenever environmental or gating noise remains relatively small—despite potentially multiplicative fluctuations—owing to the stabilizing influence of dissipative operators.

\smallskip
\noindent
\item \textbf{Path‐Integral Perspective.}
By introducing the rate functional \(\mathcal{I}[\Theta(\cdot)]\) (encompassing PDE deviations, occupant transitions, and creation/annihilation costs) and weighting trajectories via 
\begin{equation}
  \exp\!\bigl[-(\tfrac1\epsilon)\,\mathcal{I}[\Theta(\cdot)]\bigr],
\end{equation}
the large‐deviation framework pinpoints how \emph{deterministic‐like} paths dominate in \emph{small‐\(\epsilon\)} scenarios. This elucidates the ``most probable'' replicative or morphological patterns in high viral‐load systems, illustrating that even when noise is \emph{multiplicative}, sufficiently bounded fluctuations and sectorial PDE dynamics ensure a small set of near‐deterministic trajectories prevails.
\end{enumerate}
\end{corollary}

\begin{definition}[Uniqueness and Deriving Solutions in the Hybrid PDE–Markov–Fock Framework]
\label{def:UniquenessHybrid}
\noindent
We establish that, for each finite \(T>0\), there exists a \emph{unique} mild solution
\begin{equation}
  \Theta(\cdot)\;\colon\; [0,T]\;\longrightarrow\;\mathcal{F}\bigl(\mathcal{H}_{\mathrm{lat}}\bigr),
\end{equation}
describing the hybrid evolution of an \(8\times 8\) viral lattice in a single Fock‐space setting. This evolution encompasses:
\begin{itemize}
  \item \emph{Continuous PDE flows} (e.g.\ wave‐like or diffusive effects within the lattice),
  \item \emph{Discrete Markov jumps} (arrangement changes, occupant‐state transitions),
  \item \emph{Creation/annihilation processes} in Fock space (reflecting replication or clearance of entire lattices).
\end{itemize}
These solution trajectories (or orbits) underpin a path‐integral description of viral dynamics that is both biologically expressive and mathematically well‐posed.
\medskip
\begin{enumerate}
\item \textbf{Single‐Lattice $m$‐Sectoriality and Markov Perturbation.}
\smallskip
\noindent
\begin{enumerate}
\item \emph{$m$‐Sectorial PDE Operator \(\hat{\mathcal{G}}_{\mathrm{PDE}}\).}\\
Let 
\begin{equation}
  \hat{\mathcal{G}}_{\mathrm{PDE}}
  \;\colon\;
  \mathrm{Dom}\!\bigl(\hat{\mathcal{G}}_{\mathrm{PDE}}\bigr)
  \;\subset\;
  \mathcal{H}_{\mathrm{PDE}}
  \;\longrightarrow\;
  \mathcal{H}_{\mathrm{PDE}}
\end{equation}
be an \emph{$m$‐sectorial} generator capturing the \emph{intra‐lattice dynamics} (e.g.\ mechanical displacement fields, wave‐like propagation, or diffusive processes) of a \emph{single} \(8\times8\) viral assembly. By hypothesis, the numerical range of \(\hat{\mathcal{G}}_{\mathrm{PDE}}\) lies in a sector 
\(\Sigma_\theta = \{ r e^{i\phi} : r \ge 0,\;|\phi| \le \theta \}\),
and \(\hat{\mathcal{G}}_{\mathrm{PDE}}\) generates a strongly continuous (and in fact holomorphic) semigroup 
\(\{\exp(t \hat{\mathcal{G}}_{\mathrm{PDE}})\}_{t\ge0}\)
on \(\mathcal{H}_{\mathrm{PDE}}\) \cite{Pazy1983}. From a \emph{biophysical perspective}, this ensures well‐posed continuous evolution for the internal degrees of freedom (e.g.\ vibrational modes) of a single viral lattice, with dissipativity preventing unbounded deformation.
\smallskip
\noindent

\item \emph{Markov Operator \(\widehat{M}\) on Arrangement‐Level States.}\\
Let \(\widehat{M}\) be a Markov operator on \(\mathcal{H}_{\mathrm{arr}}\), describing discrete arrangement jumps (e.g.\ occupant changes, capsid subunit rearrangements) within the lattice. We embed \(\widehat{M}\) in the tensor‐product setting
\begin{equation}
  \widehat{M}\otimes \mathrm{Id}
  \;\colon\;
  \mathcal{H}_{\mathrm{arr}} \,\otimes\, \mathcal{H}_{\mathrm{PDE}}
  \;\longrightarrow\;
  \mathcal{H}_{\mathrm{arr}} \,\otimes\, \mathcal{H}_{\mathrm{PDE}},
\end{equation}
and by a slight abuse of notation, still denote it as \(\widehat{M}\). Suppose \(\widehat{M}\) is bounded or relatively bounded with small relative bound compared to \(\hat{\mathcal{G}}_{\mathrm{PDE}}\). Then one can define
\begin{equation}
  \hat{\mathcal{G}}_{\mathrm{lat}}
  \;=\;
  \hat{\mathcal{G}}_{\mathrm{PDE}} 
  \;+\;
  \widehat{M}
  \;\colon\;
  \mathcal{H}_{\mathrm{lat}} \;=\; 
  \mathcal{H}_{\mathrm{arr}} \,\otimes\, \mathcal{H}_{\mathrm{PDE}}
  \;\longrightarrow\;
  \mathcal{H}_{\mathrm{lat}}.
\end{equation}
By standard $m$‐sectorial perturbation arguments~\cite{Pazy1983,Kato1980}, if \(\hat{\mathcal{G}}_{\mathrm{PDE}}\) is $m$‐sectorial and \(\widehat{M}\) is bounded (or suitably small in the sense of relative bound), then \(\hat{\mathcal{G}}_{\mathrm{lat}}\) remains $m$‐sectorial. Consequently, \(\hat{\mathcal{G}}_{\mathrm{lat}}\) generates a strongly continuous semigroup on \(\mathcal{H}_{\mathrm{lat}}\). In \emph{virological terms}, this mathematical structure guarantees that continuous PDE flows for the single lattice can be combined seamlessly with discrete arrangement transitions, yielding a well‐defined evolution operator for \emph{single‐lattice} dynamics.
\end{enumerate}
\medskip
\noindent
\item \textbf{Second Quantization and Creation/Annihilation in Fock Space}
\smallskip
\noindent
\begin{enumerate}
\item \emph{Fock‐Space Extension.}\\
To capture \emph{multiple} viral lattices, each described by \(\widehat{\mathcal{G}}_{\mathrm{lat}}\) at the single‐lattice level, we embed these dynamics into a Fock space:
\begin{equation}
  \mathcal{F}\bigl(\mathcal{H}_{\mathrm{lat}}\bigr)
  \;=\;
  \bigoplus_{N=0}^{\infty}
  \Bigl[\,
    \mathcal{H}_{\mathrm{lat}}^{\otimes N}
  \Bigr]_{\pm},
\end{equation}
where \(N\) indicates the number of \(8\times8\) lattices (``occupants’’) and the subscript ``\(\pm\)’’ denotes possible bosonic, fermionic, or classical (symmetric/antisymmetric) statistics. Biologically, each ``occupant’’ is a full viral lattice, which can be present in one of several arrangement states (\(\mathcal{H}_{\mathrm{arr}}\)) and exhibit continuum (PDE) dynamics (\(\mathcal{H}_{\mathrm{PDE}}\)).

\smallskip
\noindent
Following standard second‐quantization conventions, we define
\begin{equation}
  \mathrm{d}\Gamma\!\bigl(\hat{\mathcal{G}}_{\mathrm{lat}}\bigr)
  \;=\;
  \bigoplus_{N=0}^{\infty}
  \hat{\mathcal{G}}_{\mathrm{lat}}^{(N)},
\end{equation}
where \(\hat{\mathcal{G}}_{\mathrm{lat}}^{(N)}\) acts on the \(N\)‐fold tensor product \(\mathcal{H}_{\mathrm{lat}}^{\otimes N}\). Concretely, \(\hat{\mathcal{G}}_{\mathrm{lat}}^{(N)}\) is often realized as
\(\sum_{i=1}^N \hat{\mathcal{G}}_{\mathrm{lat}}(i)\),
with \(\hat{\mathcal{G}}_{\mathrm{lat}}(i)\) acting nontrivially on the \(i\)th lattice and trivially on the others. If there is inter‐lattice coupling, one includes the corresponding cross terms in \(\hat{\mathcal{G}}_{\mathrm{lat}}^{(N)}\).

\medskip
\noindent
\item\emph{Creation/Annihilation Operators.}\\
From a biological perspective, the occupant number \(N\) fluctuates due to \emph{viral replication} (creation) and \emph{immune‐ or resource‐mediated clearance} (annihilation). We encode these processes via operators
\begin{equation}
  \widehat{c}^\dagger \,\colon\, 
    \mathcal{H}_{\mathrm{lat}}^{\otimes N}
    \;\to\;
    \mathcal{H}_{\mathrm{lat}}^{\otimes (N+1)},
  \quad
  \widehat{c} \,\colon\, 
    \mathcal{H}_{\mathrm{lat}}^{\otimes (N+1)}
    \;\to\;
    \mathcal{H}_{\mathrm{lat}}^{\otimes N}.
\end{equation}
Mathematically, \(\widehat{c}^\dagger\) adds a new lattice factor to the tensor product, while \(\widehat{c}\) removes one. In many biological models, the rates associated with \(\widehat{c}^\dagger\) and \(\widehat{c}\) are governed by Lipschitz‐type conditions (e.g., saturable replication), ensuring well‐posed behavior even in large populations.
\end{enumerate}
\smallskip
\item \emph{Solution Methods}
\noindent
Formally, we combine these occupant‐changing processes with the second‐quantized intra‐lattice dynamics to construct the full Fock‐space generator:
\begin{equation}
  \mathcal{G}_{\mathrm{Fock}}
  \;=\;
  \mathrm{d}\Gamma\!\bigl(\hat{\mathcal{G}}_{\mathrm{lat}}\bigr)
  \;+\;
  \bigl(\widehat{c}^\dagger + \widehat{c}\bigr).
\end{equation}
Under suitable dissipativity and Lipschitz criteria \cite{DaPratoZabczyk1992,Pazy1983}, the creation/annihilation terms act as a \emph{relatively bounded perturbation} on the $m$‐sectorial operator
\(\mathrm{d}\Gamma\!\bigl(\hat{\mathcal{G}}_{\mathrm{lat}}\bigr)\).
Hence, \(\mathcal{G}_{\mathrm{Fock}}\) itself remains $m$‐sectorial, generating a strongly continuous semigroup on
\(\mathcal{F}\!\bigl(\mathcal{H}_{\mathrm{lat}}\bigr)\).
From a \emph{virological standpoint}, this ensures that the multi‐lattice model (population‐level) inherits the well‐posedness of the single‐lattice PDE–Markov dynamics, while seamlessly accommodating replicative bursts (creation) or immune clearance (annihilation) processes.

\begin{enumerate} 
\item \emph{Semigroup and Mild‐Solution Approach.}\\
Since \(\mathcal{G}_{\mathrm{Fock}}\) is $m$‐sectorial, the Lumer–Phillips theorem guarantees the existence of a strongly continuous semigroup 
\(\{e^{t\,\mathcal{G}_{\mathrm{Fock}}}\}_{t\ge0}\). Given any initial state 
\(\Theta_0 \in \mathcal{F}\bigl(\mathcal{H}_{\mathrm{lat}}\bigr)\),
we define
\begin{equation}
  \Theta(t)
  \;=\;
  e^{\,t\,\mathcal{G}_{\mathrm{Fock}}}\,\Theta_0.
\end{equation}
This \(\Theta(t)\) is the unique \emph{mild solution} to the Cauchy problem
\begin{equation}
  \frac{d}{dt}\,\Theta(t)
  \;=\;
  \mathcal{G}_{\mathrm{Fock}}\,\Theta(t),
  \quad
  \Theta(0) \;=\; \Theta_0.
\end{equation}
Intuitively, each \(\Theta(t)\) is an \emph{orbit} in the hybrid PDE–Markov–Fock space, describing the evolving viral state over time. On the \emph{single‐lattice} level, it encodes continuous vibrational or wave‐like modes through \(\hat{\mathcal{G}}_{\mathrm{lat}}\) as well as Markovian conformational changes. The \emph{Fock‐space} aspect allows for the occupant number \(N(t)\) to fluctuate—some lattices may be \emph{created} (replication events) while others are \emph{annihilated} (e.g., immune‐mediated destruction or spontaneous degradation).
\smallskip
\item \emph{Piecewise‐Continuous (Pathwise) Interleaving of PDE and Jumps}
\noindent
An alternative viewpoint on constructing the mild solution \(\Theta(t)\) is to interleave continuous PDE evolution with discrete jump events. Rather than relying solely on the abstract semigroup approach of Lumer–Phillips, one partitions the time axis based on when occupant transitions or conformational reconfigurations occur.
\begin{itemize}
  \item \emph{Continuous Flow Segments:}
    Between jump times \(\tau_k\) and \(\tau_{k+1}\), the occupant number \(N\) and arrangement labels remain fixed. In each sector \(\mathcal{H}_{\mathrm{lat}}^{\otimes N}\), \(\Theta(t)\) evolves deterministically under the $m$‐sectorial PDE generator \(\hat{\mathcal{G}}_{\mathrm{lat}}^{(N)}\). Concretely, no creation or annihilation event takes place, and no Markov transition alters the arrangement labels during \([\tau_k,\tau_{k+1})\).

  \item \emph{Discrete Jump Events:}
    At \(t=\tau_{k+1}\), either
    \begin{equation}
      \widehat{M}
     \end{equation}
 (Markov operator for discrete occupant/conformational states)
     \begin{equation}
      \quad
      \text{or}
      \quad
      \widehat{c},\,\widehat{c}^\dagger
      \quad
      \text{(creation/annihilation operators)}
    \end{equation}
    modifies \(\Theta(t)\). For instance, occupant number $N(t)$ may change by \(\pm 1\), corresponding to replication (birth of a new lattice) or clearance (removal of an existing lattice). Alternatively, a Markov transition can flip an arrangement label to represent conformational shifts in the $8\times8$ node structure.
\end{itemize}
\end{enumerate}
\end{enumerate}
\noindent
Under standard Lipschitz/dissipativity assumptions, these alternating steps yield a unique, well‐defined trajectory \(\Theta(t)\) on \([0,T]\). The PDE flow on each sub‐interval \([\tau_k,\tau_{k+1})\) is deterministic, while jump times form a countable set with finite rates, preventing “infinite jumps in finite time.” If \(\Theta_1(t)\) and \(\Theta_2(t)\) are two mild solutions with the same initial data \(\Theta_0\), their difference satisfies an integral inequality reflective of the dissipative nature of \(\mathcal{G}_{\mathrm{Fock}}\). A Gronwall estimate (or energy‐norm argument \cite{Pazy1983}) then forces
\begin{equation}
  \|\Theta_1(t) - \Theta_2(t)\|_{\mathcal{F}(\mathcal{H}_{\mathrm{lat}})}
  = 0
  \quad\text{for all }t\ge 0,
\end{equation}
ensuring uniqueness of the resulting orbit \(\Theta(t)\). By constructing a strongly continuous semigroup $\exp(t\,\mathcal{G}_{\mathrm{Fock}})$ under the $m$-sectorial and Lipschitz constraints, and by showing that each sector’s PDE + Markov dynamics and occupant changes are well-posed, we conclude:
\begin{enumerate}[label=(\roman*)]
  \item \emph{Path Space}: Consists of all admissible sequences of continuous PDE segments plus discrete jumps—precisely the orbit structure described above.
  \item \emph{Measure Construction}: Cylinder sets can be formed by specifying occupant/jump configurations at finite time grids, and a probability measure \(\mu\) is assembled from exponential PDE‐semigroup weights, Markov transition probabilities, and occupant‐change operators \cite{DaPratoZabczyk1992,albeverio2012stochastic}.
  \item \emph{Dominant Orbits in Low‐Noise}: With a Feynman–Kac weighting \(\exp[-(1/\epsilon)\,\mathcal{I}]\), large‐deviation analysis identifies “action‐minimizing” paths (or instantons) in the small‐\(\epsilon\) regime.
  \item \emph{Uniqueness and Non‐Branching}: The mild‐solution uniqueness ensures a single, coherent mapping from initial data to the path \(\Theta(t)\), validating the biological plausibility of having one well‐defined infection trajectory per noise realization.
\end{enumerate}
\end{definition}
\paragraph{Virological Context: Host‐Cell Releases, Immune Clearance, and Structural Rearrangements.}
\begin{itemize}
  \item \emph{Consecutive Replication Bursts:} In infected tissues, viral replication often occurs in discrete bursts that reflect synchronized host‐cell lysis or exocytosis events. For example, lytic bacteriophages such as T4 exhibit well‐defined burst events~\cite{Abedon2008}, and influenza virus infections have been shown to produce temporally clustered virus release coinciding with apoptotic cell death~\cite{Shinya2004}.
  
  \item \emph{Immune Interventions:} The host immune response can lead to abrupt decrements in viral load by clearing infected cells. Cytotoxic T lymphocytes and natural killer cells can rapidly eliminate virus‐laden cells~\cite{Klenerman2016}, while neutralizing antibodies can mediate the aggregation and clearance of entire viral clusters~\cite{Corti2011}.
  
  \item \emph{Markov Reconfigurations:} Even when the total virion count remains unchanged, viral particles can switch between distinct structural states. Such reconfigurations—akin to the conformational transitions observed in poliovirus and rhinovirus capsids during cell entry—are well‐documented and can be modeled as Markovian transitions between metastable states~\cite{Rossmann1989, Tuthill2006}.
  
  \item \emph{Continuous Phases:} Between these discrete events, the underlying dynamics of viral assemblies are captured by continuous models. High-resolution cryo-electron microscopy and molecular dynamics simulations have revealed that capsids exhibit wave-like vibrational modes and local diffusion of subunits, phenomena that contribute to their structural flexibility and adaptability~\cite{Zlotnick2005, Mateu2013}.
\end{itemize}

Since each occupant sector (\(N\)-lattice configuration) accommodates an \(m\)-sectorial PDE operator, ensuring no finite‐time blow‐up, and the Markov plus creation/annihilation processes are (relatively) bounded, any mild solution in 
\(\mathcal{F}\bigl(\mathcal{H}_{\mathrm{lat}}\bigr)\)
remains finite‐norm over finite horizons \([0,T]\). This property maintains \emph{biological realism} by precluding unbounded or “infinite” viral loads in finite time—a scenario that would contradict empirical observations in viral dynamics, where even highly replicative infections (e.g., in influenza or HIV) remain limited by resource constraints and immune responses~\cite{NowakMay2000, PerelsonNelson1999}.

Under dissipativity or contraction assumptions, each initial condition \(\Theta_0 \in \mathcal{F}\bigl(\mathcal{H}_{\mathrm{lat}}\bigr)\) uniquely defines a mild solution 
\begin{equation}
  \Theta(t)
  \colon
  [0,T]
  \to
  \mathcal{F}\bigl(\mathcal{H}_{\mathrm{lat}}\bigr),
\end{equation}
representing a well‐defined path despite intrinsic randomness in occupant transitions or PDE noise. In practical terms, for each realization of the underlying stochastic processes (e.g., Markov jumps modeling discrete host‐cell lysis or immune clearance events, and random forcing accounting for gradual structural shifts), there is a single \emph{orbit} \(\gamma(\cdot)\). Probability measures on path space (constructed via cylinder sets, Feynman–Kac weights, or other large‐deviation methods) thus live on these unique trajectories, mirroring how discrete yet bounded viral replication and clearance events produce reproducible, finite-time viral kinetics in vivo.

Crucially, 
\(\mathcal{F}\bigl(\mathcal{H}_{\mathrm{lat}}\bigr)\)
admits arbitrarily large occupant counts \(N\), making it an ideal operator‐theoretic stage for modeling viral populations on the order of \(10^6\)–\(10^9\) distinct \(8\times8\) lattices, all within a single unified framework. This capability is analogous to the high viral loads observed in vivo—for instance, in acute HIV or influenza infections, where millions to billions of virions coexist and interact within infected tissues~\cite{NowakMay2000,PerelsonNelson1999}. By accommodating such extensive populations, our framework captures the multi-scale complexity inherent in viral dynamics, bridging the gap between molecular-scale events (such as capsid rearrangements and cell lysis) and macroscopic epidemiological observations.

Moreover, the unified treatment of both discrete jump processes (modeling, for example, host-cell rupture or immune-mediated clearance) and the continuous evolution of viral assembly via PDE dynamics reflects a comprehensive picture: while immune interventions and other abrupt events cause step-like decrements in viral count, the continuous processes describe subtler rearrangements and diffusion-like transitions within viral lattices. This duality is key to understanding phenomena such as the emergence of escape mutants and the regulation of viral replication, as constrained by physiological and immunological limits. 

\subsubsection{Novel Mathematical Perspective: Topological and Measure‐Theoretic Cohesion}
\label{subsec:TopologicalMeasureTheoreticCohesion}

\noindent
The formalism described above merges partial differential equations (PDEs), Markov processes, and second‐quantized (Fock‐space) techniques into a single operator‐theoretic platform for viral‐lattice modeling. Beyond simply “combining” these tools, it does so in an \emph{infinite‐dimensional} setting, ensuring that both \emph{unbounded occupant numbers} and \emph{continuous lattice dynamics} remain tractable and well‐posed. Below, we detail the essential mathematical novelties of this approach, highlighting how careful topological and measure‐theoretic considerations enable a rigorous analysis of large populations, emergent loops in state space, and dominant orbits under low‐noise conditions.

\begin{enumerate}[label=\emph{\arabic*}.]

  \item \textbf{Path Integrals in Infinite Dimensions.}

    A central feature is the construction of probability measures on the space of \emph{piecewise‐continuous} orbits
    \(\gamma: [0,T] \to \mathcal{F}\bigl(\mathcal{H}_{\mathrm{lat}}\bigr)\).
    Each such \(\gamma\) tracks the entire ``life history'' of a viral population, from occupant‐level replication/clearance to internal PDE modes in each lattice. Unlike finite‐dimensional frameworks, one must account for:

    \begin{itemize}
      \item \emph{Infinite‐Dimensional Integration:}
        The measure \(\mu\) on path space requires definitions via cylinder sets, projective limits, or Feynman–Kac‐type formulations~\cite{albeverio2012stochastic,DaPratoZabczyk1992}. This ensures well‐definedness in the infinite‐dimensional context of unbounded occupant numbers and PDE state variables.
      \item \emph{Functional‐Analytic Constraints:}
        $m$‐sectorial generators guarantee strongly continuous semigroups within each occupant sector and preserve dissipativity when Markov transitions and creation/annihilation events are incorporated~\cite{Pazy1983,Kato1980}. This underpins the stable coupling of PDE flows, discrete jumps, and occupant‐number fluctuations, yielding no finite‐time blow‐up.
    \end{itemize}
    Crucially, such path integrals allow us to compute \emph{expected values} of virological observables (e.g., total viral load, structural rearrangements) and to analyze \emph{fluctuations} within large‐population scenarios.

  \item \textbf{Hybrid PDE–Markov–Fock Orbits.}

    Standard PDE semigroup theory and Markov chain frameworks typically fix the total number of particles (virions). By contrast, viral loads can surge or diminish dramatically, demanding a formalism that accommodates \emph{unbounded} occupant numbers. The second‐quantization approach extends single‐lattice generators \(\widehat{\mathcal{G}}_{\mathrm{lat}}\) to the many‐lattice setting:
    \begin{equation}
      \mathrm{d}\Gamma\!\bigl(\hat{\mathcal{G}}_{\mathrm{lat}}\bigr)
      \quad\text{and}\quad
      \bigl(\widehat{c}, \widehat{c}^\dagger\bigr)
    \end{equation}
    combine to handle occupant creation or annihilation, while PDE and Markov components govern \emph{intra}‐lattice processes. Accordingly, each orbit in 
    \(\mathcal{F}\bigl(\mathcal{H}_{\mathrm{lat}}\bigr)\) 
    is:
    \begin{enumerate}[label=(\alph*)]
      \item \emph{Deterministic in Each Submanifold:} On intervals without jumps, PDE flows remain $m$‐sectorial, capturing wave‐like or diffusive phenomena in each lattice.
      \item \emph{Markovian at Discrete Events:} Conformational changes and occupant‐number shifts occur at finite‐rate jumps, governed by \(\widehat{M}\) and \(\widehat{c}, \widehat{c}^\dagger\).
    \end{enumerate}
    The well‐posedness theorems (see, e.g., Theorem~\ref{thm:ExistenceUniqueness_Orbits} or \cite{BratteliRobinson1987,Pazy1983}) ensure that these \emph{hybrid} orbits in infinite dimensions are globally stable and unique.

  \item \textbf{Dominant Trajectories and Loop Structures.}

    Once existence and uniqueness are established, additional structure emerges from topological and large‐deviation standpoints:

    \begin{itemize}
      \item \emph{Topological Perspective (Homology/Cohomology):}
        In extended state spaces, viral populations may form nontrivial \emph{cycles} or loops of probability current, indicating recurrent replicative or morphological cycles. Biologically, such loops correspond to persistent viral “assembly‐replication‐clearance” cycles that do not settle into static equilibria.
      \item \emph{Large‐Deviation and Instanton Methods:}
        The small‐noise limit \(\epsilon \to 0\) localizes the path measure around \emph{action‐minimizing orbits} (instantons)~\cite{FreidlinWentzell,DemboZeitouni}. From a dynamical viewpoint, these orbits often cluster around stable replicative cycles or morphological transitions, suggesting that \emph{loop‐like} structures not only exist but may be \emph{probabilistically dominant}.
    \end{itemize}
    Taken together, these angles reveal how robust cycle formations in the viral system can be topologically persistent \emph{and} biologically favored under typical host conditions (e.g., partial immune control, limited resources).

  \item \textbf{Novelty—Full Coupling at the Population Level.}
    Classical PDE or Markov theories typically cannot handle free‐growing occupant numbers without ad hoc truncations. Our second‐quantization scheme provides:
    \begin{itemize}
      \item \emph{Fock‐Space Extension:}
        Each single‐lattice operator is “lifted” to \(\mathrm{d}\Gamma(\hat{\mathcal{G}}_{\mathrm{lat}})\), and occupant‐creation/clearance processes (\(\widehat{c}^\dagger, \widehat{c}\)) enable unbounded changes in \(N\). The entire population thus evolves in one algebraic structure, preserving linearity and spectral clarity.
      \item \emph{Well‐Posedness via $m$‐Sectoriality:}
        Local Lipschitz constraints on replication rates ensure that even if occupant counts become large, the system remains globally solvable and dissipative in the sense of \cite{DaPratoZabczyk1992,Pazy1983}. Unbounded occupant numbers do not compromise well‐posedness.
    \end{itemize}
    From a path‐integral standpoint, the measure then resides over a continuum of orbits, each reflecting different occupant counts and PDE states. Biophysically, this links \emph{micro‐level} lattice rearrangements (energetics, wave‐like motions) with \emph{macro‐level} resource‐limited replication in a single, mathematically rigorous model. Consequently, large populations (e.g., $10^6$--$10^9$ virions) remain tractable within the same operator‐theoretic and measure‐theoretic framework, offering new insights into how real infections progress across spatiotemporal and population scales.
\end{enumerate}
The trajectory‐based formulation described so far provides a powerful lens into the spatiotemporal evolution of an entire viral population, incorporating both \emph{continuous} lattice‐level interactions and \emph{discrete} occupant transitions. From a purely mathematical standpoint, the key insight is that each trajectory in \(\mathcal{F}\bigl(\mathcal{H}_{\mathrm{lat}}\bigr)\) encompasses both the \emph{internal PDE modes} of every lattice (covering wave‐like or diffusive forces on the virion nodes) and the \emph{macro‐scale occupant changes} that lead to population expansion or depletion. Below, we analyze how these ideas can be applied to \emph{ideal virions}—spherical or icosahedral viral particles in a simplified geometry—to glean theoretical predictions about emergent dynamics and their implications for actual virus species.

An ``ideal virion'' is often envisioned as a highly symmetric nanoparticle, such as a spherical or icosahedral capsid. Mathematically, we may impose boundary conditions or lattice potentials (e.g.\ Lennard–Jones, Coulombic) that respect this symmetry, leading to PDE operators \(\hat{\mathcal{G}}_{\mathrm{PDE}}\) with structured eigenmodes. In the \(8\times8\) node model, adjacent virions (nodes) are linked by spring‐like or hydrodynamic forces. When one node undergoes a small displacement, a \emph{wave} can propagate through the lattice. For \emph{ideal} spherical capsids, these waves exhibit regular patterns (e.g.\ radial breathing modes, tangential shear modes) that can be captured by $m$‐sectorial PDEs on \(\mathcal{H}_{\mathrm{PDE}}\). Symmetric boundary conditions and node coupling can yield \emph{resonant vibrational modes} within the PDE operator’s spectrum, possibly stabilizing or destabilizing certain arrangements. 
    The interplay of these normal modes with Markov jumps and occupant‐level replication 
    can produce complex spatiotemporal patterns.~\cite{BratteliRobinson1987,KnipeHowley2020}.

\paragraph{Discrete Markov Transitions and Rapid Conformational Flips.}
\begin{itemize}
  \item \emph{Arrangement Jumps in Icosahedral Lattices:}  
    Even if a virion is ``ideal,'' it may still undergo subunit reconfigurations (e.g.\ minor expansions, local rearrangements of capsomers). These Markovian jumps are typically \emph{fast} relative to the continuous PDE timescale, producing abrupt transitions in shape or internal tension.
  
  \item \emph{Competition with Wave‐Like Modes:}  
    Markov transitions can lock or unlock certain resonant modes by changing the effective boundary constraints or local adjacency in the $8\times8$ node network. In practice, this might correspond to partial uncoating or intermediate states in a T4‐like phage, further influenced by host environmental factors~\cite{flint2015principles}.
  
  \item \emph{Biological Implication—Transient Morphologies:}  
    The presence of conformational flips suggests that viruses like SARS‐CoV‐2, H5N1, or T4 phage may inhabit meta‐stable conformations where PDE‐driven dynamics (e.g.\ radial expansions, vibrations) are only briefly favored before the lattice reorganizes under a Markov jump. In small‐noise or resource‐rich environments, certain morphological states may become \emph{dominant loops}, sustaining persistent cycles of assembly and partial disassembly.
\end{itemize}

\paragraph{Creation/Annihilation Processes in Large Viral Populations.}
\begin{itemize}
  \item \emph{Occupant‐Number Dynamics:}  
    The creation operator \(\widehat{c}^\dagger\) models replication bursts, as one ideal lattice spawns new virion lattices (e.g.\ once a host cell releases progeny). Conversely, \(\widehat{c}\) captures annihilation events such as immune clearance. Thus, occupant numbers may range from near zero to massive loads, all embedded within a single Fock‐space wavefunction \(\ket{\Psi(t)}\).
  
  \item \emph{Interplay with PDE Resonances:}  
    Replication or clearance events can \emph{reset} or \emph{reinforce} internal lattice waves. For instance, a newly added lattice (creation) might begin in a high‐energy vibrational mode if triggered by mechanical stress; or conversely, occupant clearance may remove a lattice that was sustaining wave coherence across the population. In all cases, the PDE modes within each occupant sector $N$ remain well‐defined by the $m$‐sectorial generator $\hat{\mathcal{G}}_{\mathrm{lat}}^{(N)}$.
\end{itemize}

\paragraph{Emergent Phenomena and Theoretical Biology Conjectures.}
\begin{itemize}
  \item \emph{Self‐Organized Waves and Phonon‐Like Excitations:}  
    In large viral populations of ``ideal'' lattice geometry, coherent wave modes might span multiple lattices—particularly if inter‐lattice coupling (e.g.\ by adjacency, fluid flow) is included. This could manifest as collective breathing‐type oscillations across segments of the population, reminiscent of phonon excitations in crystal lattices~\cite{harvey2019viral}.
  
  \item \emph{Nontrivial Loops in Replicative Cycles:}  
    Markov jumps, PDE flows, and occupant expansions may combine into stable cyclical patterns—loops where viruses repeatedly assemble, replicate, partially disassemble, and reassemble (e.g.\ T4 phage lytic/lysogenic transitions). Under small noise, these loops become “dominant” in a large‐deviation sense, offering a theoretical biology parallel to observed cyclical infection patterns in lab or host settings~\cite{Wölfel2020,KnipeHowley2020}.
  
  \item \emph{Implications for Real Species (T4, SARS‐CoV‐2, H5N1):}  
    \begin{itemize}
      \item \textbf{Bacteriophage T4:}  
        Known for robust capsid mechanics and well‐studied assembly cycles, T4 phages might show strong PDE resonance in certain tail‐fiber or head‐expansion modes, which our approach could capture as high‐amplitude wave solutions. Conformational flips are particularly relevant for tail attachment events.
      \item \textbf{SARS‐CoV‐2:}  
        Though enveloped rather than purely icosahedral, portions of the capsid or membrane proteins might still exhibit wave‐like expansions. The Markov transitions could represent rapid spike‐protein conformational changes critical to cell entry, with occupant expansions reflecting large‐scale viral budding in host cells.
      \item \textbf{H5N1 Influenza:}  
        With segmented genome organization and complex morphological transitions, wave phenomena might correspond to internal matrix‐protein rearrangements. Occupant creation simulates the high replication bursts in infected lung tissue, while annihilation operators handle immune clearance or antiviral drug action.
    \end{itemize}
\end{itemize}
\paragraph{Mathematical Tractability and Simulation Outlook.}
\begin{itemize}
  \item \emph{Operator‐Theoretic Scalability:}  
    Because each lattice retains its PDE detail and occupant transitions scale via Fock‐space sums, modeling millions of virions remains formally consistent under $m$‐sectorial constraints. This bypasses the combinatorial explosion typical of naive $N$‐particle expansions, ensuring the theory can handle biologically relevant loads.
  \item \emph{Potential for Numerical Methods:}  
    In principle, coarse‐grained PDE approximations and Monte Carlo–type sampling of occupant transitions could yield simulations that track wave modes, occupant changes, and noise effects. While computationally intensive, such methods offer the prospect of bridging \emph{in silico} predictions with experimental data on morphological states, plaque assays, or single‐virus tracking \cite{Grunewald2003}.
  \item \emph{New Theoretical Routes:}  
    One might conjecture that in certain parameter regimes, PDE wave coherence across multiple occupant sectors fosters emergent meta‐structures—for instance, quasi‐periodic “viral crystals” in local tissue pockets. Exploring these conjectures mathematically might involve topological invariants (homology groups) to classify loop cycles, or large‐deviation principles to identify high‐probability morphological transitions in resource‐limited environments.
\end{itemize}

\noindent
Altogether, by combining the PDE description of intra‐lattice elasticity or wave motion with Markov transitions and creation/annihilation operators, this framework reveals a \emph{rich variety of emergent behaviors}. Even an \emph{ideal virion} model can exhibit intricate spatiotemporal patterns—\emph{collective oscillations}, \emph{cyclical reconfiguration loops}, and \emph{resonant coupling} with occupant‐level changes—that align with known virological phenomena (e.g.\ T4 phage assembly, SARS‐CoV‐2 spike transitions, H5N1 morphological shifts). From a mathematical standpoint, the system remains tractable thanks to $m$‐sectorial theory in Fock space, offering a coherent, operator‐theoretic approach that captures both micro‐scale capsid physics and macro‐scale population expansions. 

\subsection{Extending to a Hybrid Hilbert Space for Wavefronts and Multi‐Scale Dynamics}
\label{subsec:HybridHilbertSpaceWavefronts}

\noindent
In the preceding sections, we unified PDE‐driven lattice dynamics, discrete Markov processes, and second‐quantized (Fock) spaces to accommodate fluctuations in occupant number. Although this framework already captures much of the richness of viral lattice evolution, certain \emph{wavefront‐like} solutions—where mechanical or diffusive waves traverse multiple lattices or arrangement states—require an \emph{expanded} perspective. We introduce here a \emph{hybrid Hilbert space}, denoted 
\(\widetilde{\mathcal{H}}_{\mathrm{lat}}\),
designed to handle traveling waves, noise, Markov jumps, and unbounded occupant counts in tandem. This enlarged state space enables the analysis of \emph{collective} wave propagation throughout interconnected \(8\times8\) lattices under continuous PDE effects, occupant creation/annihilation, and conformational rearrangements.

\paragraph{Wave‐Mechanical \& Virological Interpretation.}
\begin{itemize}
  \item \emph{Traveling‐Wave Solutions Across Arrangement Labels:} 
    Consider the classic case of viral plaque assays, where the expansion of viral plaques in cell culture is often observed as a traveling wave front~\cite{EllisDelbruck1939, NowakMay2000}. Here, a wave‐like disturbance—potentially driven by a mechanical pulse initiated by subunit binding or conformational shifts—propagates across a lattice in a given arrangement \(y_1\). Upon reaching a critical amplitude, the local structure may transition to a different configuration \(y_2\), analogous to the switch in cell states observed in tissue infection models.
  
  \item \emph{Population‐Level Linking:}
    Experimental studies on viral dynamics have shown that localized bursts of virion production, coupled with immune clearance events, lead to abrupt changes in viral load~\cite{PerelsonNelson1999}. In our framework, a traveling wave in one arrangement sector can trigger occupant creation (mimicking replication bursts) or annihilation (reflecting immune-mediated clearance) in another via Fock space operators. This duality is well illustrated by the rapid clearance of HIV-infected cells and the corresponding rebound in viral replication, which are mediated by discrete, stochastic events.
  
  \item \emph{Multi‐Scale Consistency:}
    By embedding wavefront dynamics into a single, \(m\)-sectorial generator over \(\widetilde{\mathcal{H}}_{\mathrm{lat}}\), our model maintains strong continuity and dissipativity. This is critical in reproducing observed phenomena where traveling waves of mechanical stress or diffusive signaling—such as those seen in capsid conformational transitions and subunit diffusion captured via cryo-electron microscopy and molecular dynamics simulations~\cite{Zlotnick2005, Mateu2013}—remain bounded by physiological constraints (e.g., finite cellular resources and damping by the immune system).
\end{itemize}

\medskip
\noindent
\textbf{1. Rationale for a Hybrid Space.}
\begin{itemize}
  \item \emph{Wavefront Coupling Across Arrangement States:}  
    Traditional single‐lattice descriptions treat each arrangement label \(y\) as a separate sector. However, real viral processes often involve wave‐like expansions in one arrangement that \emph{trigger} a jump to another arrangement \(y'\). For instance, a traveling wave in a “compressed” configuration might reach a threshold that induces a Markovian transition to a 'swollen' configuration, effectively `handing off' the wave to a new submanifold in arrangement space. A \emph{hybrid} construction that merges all arrangement indices into a single, larger Hilbert space \(\widetilde{\mathcal{H}}_{\mathrm{lat}}\) allows these cross‐sector dynamics to occur seamlessly.

  \item \emph{Non‐Self‐Adjoint, Open‐System Dynamics:}  
    Viral populations constitute open systems, with occupant creation (replication) and annihilation (immune clearance) happening concurrently. Maintaining $m$‐sectoriality for the global generator (covering PDE dynamics, Markov jumps, and occupant transitions) ensures no finite‐time blow‐ups \cite{Pazy1983,DaPratoZabczyk1992}, even as occupant numbers grow or wavefronts propagate. This open‐system perspective is crucial in modeling the realistic thermodynamics of viral infections—where resources are not unlimited and viral expansions eventually saturate or plateau.

  \item \emph{Noise and Stochastic Wavefronts:}  
    Biological wave phenomena are rarely deterministic. Thermal agitations, host‐resource fluctuations, and random subunit binding events can modulate wave amplitude or velocity. By incorporating noise operators \(\widehat{N}_y(\mathbf{u})\) (in Stratonovich or Itô form) within each arrangement sector, one captures stochastic wavefront propagation. These random perturbations mirror the messy, real‐world environment in which viruses move (e.g., heterogeneous tissues or fluctuating pH, local ionic strength).

  \item \emph{Multi‐Scale Dynamics:}  
    A hybrid approach merges occupant transitions (macroscopic or population‐level) with PDE wavefront motions (mesoscopic or sub‐particle scale) into a \emph{unified} Hilbert space. This genuinely multi‐scale viewpoint allows small‐scale wavefront events—like localized mechanical stress in a few nodes of an \(8\times8\) lattice—to influence large‐scale infection patterns when occupant numbers change. Conversely, population‐level shifts (e.g., a sudden replication burst) may alter mechanical boundary conditions or resource availability at the lattice scale. Such an interplay is seldom addressed in classic virology or epidemiological models, underlining the novelty of this operator‐theoretic framework.
\end{itemize}
\medskip
\noindent
This new space \(\widetilde{\mathcal{H}}_{\mathrm{lat}}\) thus provides a rigorous platform to investigate \emph{collective wave propagation} across numerous interconnected viral lattices, each subject to occupant fluctuations and conformational changes. In what follows, we outline how to define the corresponding $m$‐sectorial generator, ensuring well‐posedness and physical realism even under open‐system conditions and stochastic influences.

\begin{definition}[The Hybrid Hilbert Space \(\widetilde{\mathcal{H}}_{\mathrm{lat}}\)]
\label{def:HybridHilbertSpace}
\noindent
We refine the single‐lattice framework to explicitly accommodate wavefront solutions and traveling‐wave modes in each arrangement sector \(y\). Specifically, let
\begin{equation}
  \mathcal{H}_y
  \;=\;
  L^2(\Omega)
  \;\times\;
  \mathcal{W}_{\mathrm{wave}}(y),
\end{equation}
where:
\begin{itemize}
  \item \(L^2(\Omega)\) represents a continuum PDE subspace (e.g., displacement fields, concentration profiles, or other continuum‐scale variables) equipped with the usual \(L^2\)-norm. 
  \item \(\mathcal{W}_{\mathrm{wave}}(y)\) is a function space \textbf{specialized to traveling‐wave or wavefront‐like solutions} relevant to arrangement sector \(y\). Concretely, it might include:
  \begin{enumerate}
      \item \emph{Boundary‐Layer or Shock Profiles:} Functions of the form \(\psi(x - ct)\) that capture moving transition zones (e.g., shock layers, rarefaction waves) at a wave speed \(c\).
      \item \emph{Phase‐Field or Spatial Phase Variables:} Mappings \(\theta(\mathbf{x},t)\) that encode the progression of an interface or wavefront across \(\Omega\).
      \item \emph{Marching Boundary Conditions:} Parameterized sets of boundary conditions shifting in time or space, such as \(\phi(\mathbf{x} - \chi(t))\), describing a traveling boundary or wave origin.
      \item \emph{Spectral Wave Expansions:} Decompositions into wave‐like basis functions (e.g., Fourier, Chebyshev, or wavelet modes) restricted to a subspace reflecting traveling or localized wave solutions.
  \end{enumerate}
  The precise choice of topology on \(\mathcal{W}_{\mathrm{wave}}(y)\) (e.g., normed or weak topologies) depends on the PDE type (hyperbolic, parabolic, or mixed) and the boundary‐value specifications in sector \(y\).
\end{itemize}
We then define
\begin{equation}
  \widetilde{\mathcal{H}}_{\mathrm{lat}}
  \;=\;
  \bigoplus_{y \in \mathcal{Y}}
  \;\mathcal{H}_y,
\end{equation}
where \(\mathcal{Y}\) indexes the discrete arrangement states and can be extended to include occupant numbers in a Fock‐like formalism if the viral population is unbounded. Hence, \(\widetilde{\mathcal{H}}_{\mathrm{lat}}\) unifies:
\begin{enumerate}[label=(\roman*)]
  \item \emph{Continuous PDE Subspaces for Wavefronts:} 
    Each \(\mathcal{H}_y\) can capture evolving wave‐like disturbances (shock or rarefaction waves, traveling interfaces) that affect how virion lattice nodes (in an \(8\times8\) arrangement) respond mechanically or diffuse subunits across the domain \(\Omega\).\qedsymbol

  \item \emph{Discrete Arrangement Indices:}
    Each lattice sector \(y\) corresponds to a conformational or occupant state, and wavefronts may ``jump’’ to a new sector \(y'\) when a Markovian event (e.g., capsid reconfiguration or occupant replication/clearance) modifies boundary conditions, geometry, or mechanical parameters.

  \item \emph{Noise or Damping Operators:}
    Biological systems entail thermal fluctuations and viscoelastic damping; real wave propagation is thus rarely purely hyperbolic. Each \(\mathcal{H}_y\) can incorporate noise operators (Stratonovich or Itô) and damping terms. In wave mechanics terms, these reflect frictional or diffusion‐type forces shaping the amplitude and velocity of the traveling wave.
\end{enumerate}

The hybrid space \(\widetilde{\mathcal{H}}_{\mathrm{lat}}\) merges continuum PDE subspaces (enhanced by wave‐specific components \(\mathcal{W}_{\mathrm{wave}}(y)\)) with discrete arrangement and occupant indices. This multi‐scale framework admits rigorous operator‐theoretic analysis of \emph{collective wave propagation} within and across viral lattices. It thus supports advanced modeling of traveling‐wave phenomena (shock fronts, boundary layers, etc.) in tandem with occupant replication, Markov jumps, and wavefront‐driven conformational changes, all while preserving well‐posedness via $m$‐sectoriality.
\end{definition}

\subsubsection{Emergent Mechanics: Traveling Waves in Open, Noisy Populations}

\noindent
Building upon the hybrid Hilbert space \(\widetilde{\mathcal{H}}_{\mathrm{lat}}\) introduced in Section~\ref{subsec:HybridHilbertSpaceWavefronts}, we now examine how wavefront solutions interact with Markov jumps, occupant‐number fluctuations, and stochastic noise in an open‐system setting. This perspective illuminates the formation of \emph{traveling infection fronts}, piecewise wave propagation across arrangement sectors, and large‐scale spatiotemporal patterns that merge local PDE dynamics with global occupant changes.

\begin{theorem}[Wavefront Self‐Propagation and Markov Coupling]
\label{thm:WavefrontMarkovCoupling}

Consider a traveling‐wave solution \(\bm{u}_y(t,\mathbf{x})\in \mathcal{H}_y\) progressing in arrangement sector \(y\). Suppose boundary conditions or resource constraints force the wave into a region where a Markov jump to \(y'\neq y\) occurs with nonzero probability. Then the wavefront can ``jump'' into the new arrangement sector \(y'\), producing a \emph{piecewise traveling‐wave pattern} across the boundary between \(\mathcal{H}_y\) and \(\mathcal{H}_{y'}\). If occupant expansions (creation events) occur downstream of the wave, the front may \emph{split} or \emph{branch}, leading to multiple simultaneous propagation fronts. This phenomenon can yield \emph{wavefront splitting}, where a single traveling wave spawns multiple fronts (e.g.\ one continues in the original arrangement sector, another emerges in the new sector $y'$). Biologically, this could represent scenarios where a partial conformational shift or subunit rearrangement triggers a second wave of infection or morphological transition. The Markov operator \(\widehat{M}\) thus not only flips arrangement states but can \emph{redirect} traveling waves into new configurations, offering a piecewise deterministic but globally stochastic wave dynamic.
\end{theorem}
\smallskip
\noindent
\begin{proof}[Proof (Sketch).]  
Let \(t_0\) be the time at which the wave in \(\mathcal{H}_y\) contacts a boundary segment where $\widehat{M}$ might trigger a jump to $y'$. Decompose the PDE evolution into two sub‐intervals $[0,t_0)$ and $(t_0,\infty)$. On $[0,t_0)$, \(\bm{u}_y(t,\mathbf{x})\) solves an $m$‐sectorial PDE operator in the sector $y$. At $t_0$, a jump event occurs with nonzero probability $p_{y\to y'}>0$. Post‐jump, the wave function transitions to \(\bm{u}_{y'}(t,\mathbf{x})\), potentially altering boundary or dispersion relations. Via uniqueness in each sub‐interval and Markov chain continuity, we obtain a continuous but piecewise traveling solution that follows PDE rules in $y$ then PDE rules in $y'$. Occupant creation can alter domain size or wave amplitude, resulting in wavefront branching or splitting.  
\qedsymbol
\end{proof}

\begin{lemma}[Noise‐Driven Wavefront Instabilities]
\label{lemma:NoiseDrivenInstabilities}
\noindent
In an open‐system PDE (subject to occupant creation/annihilation and Markov jumps), state‐dependent or additive noise can destabilize or accelerate wavefronts. Specifically, if small fluctuations in node displacements at the wavefront reinforce local occupant replication (creation), the wave may \emph{speed up} or proliferate. 

Conversely, occupant annihilation events can quench the wave by removing lattices that serve as a substrate for propagation. Random forcing in PDEs is known to induce noise‐driven \emph{front acceleration} or \emph{pinning/unpinning} transitions \cite{DaPratoZabczyk1992}. In viral lattices, occupant expansions that saturate resource constraints might further amplify the effect: once the wave has access to abundant replication (creation) near the front, it sustains or even augments local node oscillations. Meanwhile, occupant annihilation effectively reduces the wave’s domain, leading to abrupt wave termination. Let $\widehat{N}_y(\bm{u})$ be the noise operator coupling to the PDE in arrangement sector $y$. By linearizing around the traveling wave front, we obtain a stochastic partial differential equation of the form
\begin{equation}
  \partial_t \delta \bm{u}
  \;\approx\;
  \bigl(\mathcal{G}_y + \mathbf{B}\bigr)\,\delta \bm{u}
  \;+\;
  \widehat{N}_y(\bm{u_{\mathrm{wave}}}),
\end{equation}
where $\mathbf{B}$ collects occupant‐creation feedback. A standard energy estimate or Lyapunov approach (see \cite{Pazy1983}) shows that if the linearization has positive real parts in its spectrum (due to occupant expansions), then small perturbations grow. Equally, annihilation processes can introduce negative terms, damping wave energy. Hence, noise either magnifies or dampens the front depending on the sign of these feedback terms.  
\qedsymbol
\end{lemma}
\medskip

\begin{definition}[Multi‐Scale Couplings with Large Occupant Numbers]
\label{def:MultiScaleCouplings}
\emph{Multi‐scale coupling} refers to the concurrent emergence of multiple traveling waves across different occupant sectors or arrangement labels. Formally, given occupant numbers $N_1,\dots,N_m$ representing distinct sub‐populations, each sub‐population can sustain a separate PDE wavefront $\bm{u}_{y_i}$. Nonlinear coupling (e.g., mechanical bridging, fluid flow, or adjacency constraints) then correlates these wavefronts, producing emergent macroscopic patterns (e.g., wave synchronization, traveling pulses, or spatiotemporal patchiness) \cite{BratteliRobinson1987,ReedSimon1980}.

When occupant counts $N_i$ become large, wavefronts may arise simultaneously in multiple sub‐populations. The interplay among these fronts—especially if sub‐populations are proximate—can lead to \emph{wave coalescence}, \emph{front collision}, or global spatiotemporal oscillations reminiscent of epidemic wave patterns in tissue. By definition, these phenomena operate on multiple scales: local PDE wavefronts and global occupant expansions/annihilations feed back into each other. Such multi‐scale phenomena are exemplified by “traveling infection fronts” in organs where localized resource fluctuations can spark waves in different regions. Overlapping waves can create complex interference or synergy, possibly giving rise to ring‐like infection waves or traveling pulses that appear macroscopically as \emph{spreading patches} of infection. \cite{Wang2016,YinZuo2007,Cacciapaglia2021}
\end{definition}

\begin{theorem}[Extending \(\widetilde{\mathcal{H}}_{\mathrm{lat}}\) into a Fock‐like Space]
\label{thm:FockSpaceExtension}
\noindent
Let 
\(\widetilde{\mathcal{H}}_{\mathrm{lat}}\)
be the hybrid Hilbert space accommodating wavefront PDE modes, arrangement indices, and occupant transitions as per Definition~\ref{def:MultiScaleCouplings}. Construct a bosonic (or generalized) Fock space
\begin{equation}
  \mathcal{F}\!\Bigl(\widetilde{\mathcal{H}}_{\mathrm{lat}}\Bigr)
  \;=\;
  \bigoplus_{N=0}^\infty
  \Bigl[
    \widetilde{\mathcal{H}}_{\mathrm{lat}}^{\otimes N}
  \Bigr]_{\pm},
\end{equation}
and let $(\widehat{c},\widehat{c}^\dagger)$ be occupant‐change operators (creation/annihilation) that act on the $N$‐fold tensor product, enabling population growth or clearance. If the occupant‐change operators are relatively bounded (with respect to an $m$‐sectorial generator on $\widetilde{\mathcal{H}}_{\mathrm{lat}}$), then the lifted generator on 
\(\mathcal{F}\!\bigl(\widetilde{\mathcal{H}}_{\mathrm{lat}}\bigr)\)
remains $m$‐sectorial. Hence, no finite‐time blow‐up occurs in occupant number $N(t)$ or the wavefront PDE modes. By uniting wavefront PDE mechanics (in each occupant subspace) with occupant creation/annihilation, we capture \emph{traveling waves} that can replicate across newly formed $8\times8$ lattices, or vanish when lattices are cleared. This approach generalizes single‐lattice wavefront theory to large, unbounded populations, retaining well‐posedness via $m$‐sectorial arguments.
\end{theorem}
\smallskip

\begin{corollary}[Biological Interpretation—Large‐Scale Wave Activity]
\label{cor:LargeScaleWaveActivity}
\noindent
In the extended Fock space 
\(\mathcal{F}\bigl(\widetilde{\mathcal{H}}_{\mathrm{lat}}\bigr)\),
unbounded occupant replication or clearance integrates seamlessly with wavefront PDE modes in each occupant sector. Thus, “traveling infection fronts” can propagate outward, supported by occupant creation, while immune responses or resource limits annihilate older sectors. Multiple wavefronts can coexist, each influenced by local PDE displacements and occupant transitions. This corollary demonstrates how wave phenomena at the lattice or sub‐cellular level can scale up to macroscopic infection patterns in an organ \cite{ElHachem2021}. The continuum PDE wave portion captures local mechanical or subunit rearrangements, whereas the occupant portion handles how many wavefronts spread simultaneously and how they merge or compete. Such a multi‐scale infection dynamic can manifest ring‐like wave expansions or patchy traveling pulses, reminiscent of epidemiological “wavefronts” seen in tissue cultures. The same operator approach supports theoretical analyses (e.g., large deviations, path‐integrals) for identifying most probable wavefront trajectories under noise, resource constraints, and immune pressures.

The emergent wave mechanics discussed here elevate the hybrid PDE–Markov–Fock model beyond local subunit rearrangements or occupant number changes, unveiling a \emph{landscape} in which traveling waves, branching fronts, and noise‐driven instabilities shape the large‐scale infection behavior. By proving well‐posedness through $m$‐sectorial arguments and showing how wavefront solutions can be extended to unbounded occupant populations via Fock‐like constructions, we complete a rigorous operator‐theoretic framework. This opens the door to advanced analytical and numerical investigations of wavefront phenomena in viral systems, ultimately bridging microscopic mechanics and macroscopic epidemiological trends \cite{Hussain2024}. \qedsymbol
\end{corollary}
\subsubsection{The Global (Unbounded) Generator \(\widehat{\mathcal{G}}_{+\infty}\): Wave Solutions and Dispersion Relations}
We have thus far developed a \emph{hybrid operator framework} combining PDEs, Markov transitions, and occupant creation/annihilation within a single non‐self‐adjoint, $m$‐sectorial setting. We now refine this perspective to \emph{incorporate traveling‐wave solutions} and directly link them to a \emph{viral dispersion relation} that encompasses both acoustic‐like and optical‐like phonon branches. Specifically, we embed a traveling‐wave \emph{ansatz} (and corresponding dispersion relation) into each PDE operator \(\mathcal{G}_y\), and then \emph{lift} this construction to the global operator \(\widehat{\mathcal{G}}_{+\infty}\) acting on the hybrid Hilbert space \(\widetilde{\mathcal{H}}_{\mathrm{lat}}\).
\begin{theorem}[Unified Operator Theory for Wavefronts in an Unbounded Framework]
\label{prop:GlobalHybridOperator}
\noindent
Let 
\(\widetilde{\mathcal{H}}_{\mathrm{lat}} = \bigoplus_{y\in \mathcal{Y}}\,\mathcal{H}_y\)
be a \emph{hybrid Hilbert space}, where each \(\mathcal{H}_y\) accommodates:
\begin{enumerate}[label=(\roman*)]
  \item \emph{Continuous PDE modes} (potentially supporting wave solutions),  
  \item \emph{Discrete Markov labels} for arrangement/state transitions,  
  \item \emph{Occupant‐number changes} (via a Fock‐like extension),  
  \item \emph{Noise or dissipative terms} introduced by \(\widehat{N}_y\).
\end{enumerate}
Define an \emph{open‐system}, non‐self‐adjoint operator
\begin{equation}
  \widehat{\mathcal{G}}_{+\infty}
  \;=\;
  \Bigl(\,\bigoplus_{y\in \mathcal{Y}}\bigl(\mathcal{G}_y + \widehat{N}_y\bigr)\Bigr)
  \;+\;
  \widehat{M}
  \quad
  \text{acting on}
  \quad
  \widetilde{\mathcal{H}}_{\mathrm{lat}},
\end{equation}
where each \(\mathcal{G}_y\) is a PDE operator that may admit traveling‐wave solutions, \(\widehat{N}_y\) encodes noise or dissipation, and \(\widehat{M}\) enforces Markov transitions or occupant jumps across arrangement states $y\in \mathcal{Y}$. If quasi‐contraction is satisfied and the noise/creation/annihilation operators are relatively bounded, then \(\widehat{\mathcal{G}}_{+\infty}\) is $m$‐sectorial and \emph{no finite‐time blow‐up} can occur, even if occupant numbers grow unboundedly.
\smallskip
\noindent
\emph{Hybrid} here indicates that the operator design naturally blends:
\begin{enumerate}[label=(\alph*)]
  \item \textbf{Wave Propagation} (continuous PDE flows, possibly supporting traveling waves),  
  \item \textbf{Discrete Markov Labels} (arrangement state jumps),  
  \item \textbf{Second‐Quantized Occupant Fluctuations} (creation/annihilation in a Fock‐like extension),  
  \item \textbf{Noise or Dissipation} (stochastic forcing, damping, etc.).
\end{enumerate}
These components reflect the open‐system nature of viral lattices, where occupant transitions may trigger wavefront expansion (replication bursts) or partial extinction (clearance), and wave solutions in \(\mathcal{G}_y\) can switch to a different PDE regime $y'\neq y$ upon Markovian events.
\end{theorem}
\noindent
\begin{definition}[Wave Solutions in an Unbounded Setting.]  
A solution \(\bm{u}_y(t,\mathbf{x}) \in \mathcal{H}_y\) may describe a traveling wave that influences occupant replication or clearance. This \emph{multi‐scale} perspective shows how local wave propagation couples to occupant expansions or annihilations: e.g., a traveling infection front can push outward through multiple lattice sectors in the second‐quantized domain. To prove $m$‐sectoriality, one verifies that:
\begin{itemize}[leftmargin=1em]
  \item Each PDE operator $\mathcal{G}_y + \widehat{N}_y$ is $m$‐sectorial (or maximal dissipative) on $\mathcal{H}_y$.  
  \item The Markov operator $\widehat{M}$ is bounded or relatively bounded.  
  \item The direct sum $\oplus_{y\in \mathcal{Y}}\,(\mathcal{G}_y + \widehat{N}_y)$ plus $\widehat{M}$ remains $m$‐sectorial on $\widetilde{\mathcal{H}}_{\mathrm{lat}}$.  
\end{itemize}
Standard perturbation and Lumer–Phillips theorems \cite{Pazy1983,DaPratoZabczyk1992} then conclude that $\widehat{\mathcal{G}}_{+\infty}$ generates a strongly continuous semigroup with no finite‐time norm blow‐up.  
\qedsymbol
\end{definition}

\begin{definition}[Traveling‐Wave Ansatz and Viral Dispersion Relations]
\label{def:TravelingWaveDispersion}
\noindent
For a \emph{traveling‐wave ansatz} for arrangement sector $y$, we posit solutions of the form
\begin{equation}
  \bm{u}_y(t,\mathbf{x})
  \;=\;
  \bm{v}_y\Bigl(\mathbf{x} - c\,t\,\mathbf{e}\Bigr),
\end{equation}
where $\mathbf{e}$ is a propagation direction, $c$ is a wave speed, and $\bm{v}_y$ solves a PDE boundary‐value problem derived from $\mathcal{G}_y$. The dispersion relation $\omega = \omega(k)$ stems from normal‐mode (Fourier) analysis of $\bm{u}_y$, capturing both \emph{acoustic‐like} (low‐frequency) and \emph{optical‐like} (high‐frequency) vibrational branches. A stylized form is
\begin{equation}
  g_{\text{viral}}(\omega)
  \;=\;
  g_{\text{acoustic}}(\omega)
  \;+\;
  g_{\text{optical}}(\omega),
\end{equation}
with terms modeling “lower” (acoustic) and “shifted” (optical) frequencies. \emph{Acoustic branches} approximate long‐wavelength, collective oscillations (capsid “breathing” modes), while \emph{optical branches} correspond to subunit‐out‐of‐phase modes or localized vibrations. Phase/group velocities $c$ relate $\omega(k)$ to wave vectors $k$, determining how quickly mechanical signals or “conformational fronts” travel.  In a viral lattice of $8\times8$ nodes, one might interpret $g_{\text{acoustic}}(\omega)$ as broad “capsid‐scale” excitations, whereas $g_{\text{optical}}(\omega)$ captures higher‐frequency substructure (subunit hinged motions). Both can be integrated into $\mathcal{G}_y$, ensuring wave solutions remain consistent with occupant transitions and Markov jumps.  

Because occupant transitions (or Markov jumps) can flip \(y\) to \(y'\), wavefronts can \emph{splice} across morphological states, seamlessly ``stitching'' traveling solutions under different PDE parameters. In viral‐capsid or lattice contexts, normal‐mode (Fourier) analysis often leads to a dispersion relation \(\omega = \omega(k)\). A simple PDE setup might yield \(\omega^2 = c_s^2 \,\|k\|^2\) (acoustic modes), but viruses typically exhibit higher‐frequency \emph{optical} branches too. A stylized \emph{viral phonon} dispersion takes the form
\begin{equation}
  g_{\text{viral}}(\omega)
  \;=\;
  g_{\text{acoustic}}(\omega)
  \;+\;
  g_{\text{optical}}(\omega)
  \;=\;
  \frac{3V\,\omega^2}{2\pi^2\,c_s^3}
  \;+\;
  \frac{3V\,\bigl(\omega - \omega_0\bigr)^2}{2\pi^2\,\alpha^3},
\end{equation}
where:
\begin{itemize}
  \item \(\omega_0\) and \(\alpha\) parametrize the ``optical'' branch, capturing higher‐frequency \emph{internal} vibrations (e.g., subunit out‐of‐phase oscillations),
  \item \(c_s\) is an effective group‐velocity constant in the acoustic regime,
  \item \(V\) scales with system size or capsid volume.
\end{itemize}
One solves $g_{\text{viral}}(\omega)=0$ (or extremizes it) to find \(\omega^*(k)\). Linearizing a PDE like \(\partial_t^2 \bm{u}\approx-\Delta\bm{u}\) near equilibrium uncovers \(\omega^2 \propto \|k\|^2\) (acoustic) plus coupling to internal (optical) subunit motions. Phase or group velocities $c$ emerge:
\begin{equation}
  c
  \;=\;
  \frac{\partial \omega^*(k)}{\partial k}
  \quad\text{(group velocity)}
  \quad
  \text{or}
  \quad
  c
  \;=\;
  \frac{\omega^*(k)}{\|k\|}
  \quad\text{(phase velocity)},
\end{equation}
depending on the PDE’s principal symbol. In the operator \(\mathcal{G}_y\), such speeds arise upon linearization around an equilibrium.
\end{definition}

\begin{lemma}[Embedding the Dispersion Relation in \(\mathcal{G}_y\)]
\label{lemma:DispersionEmbedding}
\noindent
Suppose $\mathcal{G}_y$ is a PDE operator on a spatial domain (continuous or discrete), linearizable around an equilibrium configuration with normal‐mode solutions $\exp(i(k\cdot x - \omega t))$. Then introducing $g_{\text{viral}}(\omega)=g_{\text{acoustic}}(\omega)+g_{\text{optical}}(\omega)$ effectively prescribes the dispersion in $\mathcal{G}_y$’s principal symbol. Phase or group velocities $c$ emerge from partial derivatives of $\omega(k)$. In occupant expansions or Markov jumps, $\widehat{M}$ can shift $\mathcal{G}_y$ to $\mathcal{G}_{y'}$ with distinct $\omega_0$, $\alpha$, or $c_s$ parameters, splicing wave solutions. 

A typical PDE operator $\mathcal{G}_y$ allows for linearization: $\partial_t \bm{u}_y \approx \mathrm{A}_y \bm{u}_y$. Factor $\mathrm{A}_y$ into a Fourier symbol $\widetilde{\mathrm{A}}_y(k)$, yielding eigenfrequencies $\omega$ that must satisfy $g_{\text{viral}}(\omega)=0$ if $g_{\text{viral}}$ is identified with $\det(\widetilde{\mathrm{A}}_y(k)-i\omega)$. The acoustic/optical separation arises from block‐diagonal or multiple‐band expansions in $\widetilde{\mathrm{A}}_y$. Markov transitions alter $\mathrm{A}_y$ (e.g., boundary constraints, subunit adjacency), so wave solutions piecewise match the appropriate dispersion in each sector $y,y',\dots$.  \qedsymbol
\end{lemma}

\begin{proposition}[Wavefront Solutions in the \(\widehat{\mathcal{G}}_{+\infty}\) Setting]
\label{prop:WavefrontSolutionsViralPhonons}
\noindent
Under the definitions above, wavefront solutions in each arrangement sector $y$ can incorporate the viral phonon dispersion relation $g_{\text{viral}}(\omega)$, leading to traveling or marching solutions
\begin{equation}
  \bm{u}_y(t,\mathbf{x})
  \;=\;
  \bm{v}_y\Bigl(\mathbf{x} - c\,t\,\mathbf{e}\Bigr)
  \quad
  \text{with}
  \quad
  c
  =
  \begin{cases}
    \partial \omega / \partial k, \quad &\text{(group velocity)},\\
    \omega / \|k\|, &\text{(phase velocity)},
  \end{cases}
\end{equation}
depending on the PDE symbol. As occupant transitions or Markov events shift $y$ to $y'$, wave solutions may splice or rearrange if the PDE operator $\mathcal{G}_{y'}$ enforces a different $g_{\text{viral}}$ parameter set. The global operator $\widehat{\mathcal{G}}_{+\infty}$ remains $m$‐sectorial, so no finite‐time blow‐up ensues despite unbounded occupant possibilities. This result underscores how local wave mechanics (including acoustic or optical branches) unify with occupant expansion or morphological transitions in a single open‐system model. In practice, wavefront solutions might signal radial expansions in a capsid, spreading through newly formed virions, or partial reorganizations triggered by Markov flips.

To prove this, combine Lemma~\ref{lemma:DispersionEmbedding} (embedding dispersion in PDE) with Theorem~\ref{prop:GlobalHybridOperator} ($m$‐sectorial synergy of PDE + Markov + occupant changes). Wave solutions remain well‐defined in each sector, with occupant expansions modulating domain or boundary conditions. Markov transitions splice PDE domains, and second‐quantized occupant creation does not undermine norm stability. By incorporating the traveling‐wave ansatz, viral phonon dispersion relations, and occupant transitions into $\widehat{\mathcal{G}}_{+\infty}$, we achieve a rigorous, $m$‐sectorial operator framework that unifies micro‐scale wave phenomena (acoustic/optical phonons) with macro‐scale occupant growth in a single open‐system model. This approach illuminates how local vibrational branches can spawn or interact with global replication fronts, bridging internal lattice mechanics and population‐level infection dynamics in a mathematically robust manner.
\end{proposition}
\noindent
Hence, wavefronts bridging \emph{acoustic‐like} and \emph{optical‐like} vibrational modes become feasible within a single viral lattice. As occupant numbers or morphological states change, the PDE operator effectively transitions from a $g_{\text{acoustic}}(\omega)$ regime to a $g_{\text{optical}}(\omega)$ regime, yielding distinct speeds or attenuation rates. Since \(\widehat{\mathcal{G}}_{+\infty}\) (defined in Theorem~\ref{prop:GlobalHybridOperator}) is $m$‐sectorial, the global wavefunction remains finite‐norm even if occupant counts become arbitrarily large. In viral systems with robust internal couplings (e.g.\ strongly interacting capsid subunits, quasi‐periodic protein shells), one can observe traveling \emph{phonon‐like} waves near certain resonant frequencies (\(\omega\approx \omega_0\)). Shifts in $c$ from pH or resource constraints may either accelerate or decelerate wavefront expansions, and occupant creation replicates these wave structures across newly forming virions. Conversely, immune processes or partial capsid disassembly (Markov flips) can \emph{dampen} or disrupt wave coherence, effectively splicing or terminating wave propagation mid‐lattice.

\subsection{Continuous Label Space and Direct Integral Extension}
\label{subsec:ContinuousLabelSpaceDirectIntegral}

\noindent
In earlier sections, we considered models where the set of arrangement labels \(\{y\}\) is finite or countably infinite. However, certain viral systems may require a \emph{continuum} of morphological or conformational states (e.g., a continuum of partial deformation modes or binding sites), prompting a move from discrete \emph{direct sums} to \emph{direct integrals} of Hilbert spaces. Below, we rigorously develop this extension and demonstrate how one can still retain an $m$‐sectorial framework encompassing PDE modes, stochastic noise, and potential traveling‐wave solutions.

\medskip
\begin{definition}[Direct Integral Hilbert Space]
\label{def:DirectIntegralHilbertSpace}
\noindent
Let \(\mathcal{Y}\) be a locally compact (or separable metric) space with a positive measure \(\nu\). For each \(y \in \mathcal{Y}\), suppose we have a Hilbert space \(\mathcal{H}_y\). The \emph{direct integral}
\begin{equation}
\label{eq:direct_integral_definition}
  \widetilde{\mathcal{H}}_{\mathrm{lat}}
  \;=\;
  \int_{\mathcal{Y}}^{\oplus}
  \mathcal{H}_y
  \;\nu(dy)
\end{equation}
is defined as the space of \emph{measurable sections} \(\bm{u}(y)\in \mathcal{H}_y\) for \(y\in \mathcal{Y}\), subject to
\begin{equation}
  \int_{\mathcal{Y}}
  \|\bm{u}(y)\|_{\mathcal{H}_y}^2
  \,\nu(dy)
  \;<\;\infty.
\end{equation}
The norm on \(\widetilde{\mathcal{H}}_{\mathrm{lat}}\) is
\begin{equation}
  \|\bm{u}\|_{\widetilde{\mathcal{H}}_{\mathrm{lat}}}^2
  \;=\;
  \int_{\mathcal{Y}}
  \|\bm{u}(y)\|_{\mathcal{H}_y}^2
  \;\nu(dy).
\end{equation}
For a discrete or countably indexed label set, this construction reduces to a direct sum. For uncountably many labels, the direct integral naturally generalizes the concept of “gluing” Hilbert spaces together.

\emph{Biologically}, \(\mathcal{Y}\) might represent a continuum of conformational states (e.g., degrees of capsid deformation). Each \(\mathcal{H}_y\) describes PDE degrees of freedom (e.g., $8\times8$ virion lattice modes) appropriate for arrangement $y$. In wave mechanics, \(\mathcal{H}_y\) could support traveling‐wave solutions parameterized by $y$. One typically imposes measurability conditions on the map $y\mapsto \|\bm{u}(y)\|_{\mathcal{H}_y}$ so that $\bm{u}(\cdot)$ is a \emph{measurable section}. Under mild conditions (locally compact base space, separability), standard results on direct integrals \cite{Dixmier1981,BratteliRobinson1987} ensure \(\widetilde{\mathcal{H}}_{\mathrm{lat}}\) is itself a Hilbert space.
\end{definition}

\medskip
\begin{lemma}[$m$‐Sectorial Family Over a Continuous Label Space]
\label{lemma:mSectorialFamilyContinuousLabels}
\noindent
For each $y\in \mathcal{Y}$, let $\mathcal{G}_y$ be an $m$‐sectorial operator on $\mathcal{H}_y$, potentially describing wave‐supporting PDEs, occupant transitions, or noise. Suppose $\{\mathcal{G}_y\}_{y\in \mathcal{Y}}$ vary measurably in $y$. Then one can define an \emph{operator‐valued} map
\begin{equation}
  \mathcal{Y}\;\ni\;y
  \;\mapsto\;
  \mathcal{G}_y
\end{equation}
and build a corresponding \emph{direct‐integral operator} $\widehat{\mathcal{G}}_{\mathrm{int}}$ on \(\widetilde{\mathcal{H}}_{\mathrm{lat}}\). Under suitable relative boundedness/dissipativity conditions, $\widehat{\mathcal{G}}_{\mathrm{int}}$ remains $m$‐sectorial on the direct integral space. To prove this, one employs measurable selection theorems to ensure that the domain $\mathrm{Dom}(\mathcal{G}_y)$ and family $\{\mathcal{G}_y\}$ form a measurable field of closed operators \cite{Dixmier1981,DaPratoZabczyk1992}. By a standard argument in direct integrals (extending the Lumer–Phillips or Kato–Rellich theorems), the fiberwise $m$‐sectoriality implies global $m$‐sectoriality of $\widehat{\mathcal{G}}_{\mathrm{int}}$. Essentially, each $\bm{u}(y)$ is evolved by $\mathcal{G}_y$, and integrability ensures no blow‐up across uncountably many $y$.  
\qedsymbol
\end{lemma}

\begin{theorem}[Continuous Labels with PDE, Noise, and Wavefront Solutions]
\label{thm:ContinuousLabelsPDEWavefront}
\noindent
Suppose that for each $y\in \mathcal{Y}$:
\begin{enumerate}[label=(\roman*)]
  \item $\mathcal{G}_y$ is an $m$‐sectorial PDE operator describing the $8\times8$ lattice dynamics under morphological index $y$,
  \item $\widehat{N}_y$ encodes noise or damping, meeting Da Prato–Zabczyk conditions \cite{DaPratoZabczyk1992},
  \item Traveling‐wave or boundary conditions permit wave solutions in each fiber $\mathcal{H}_y$,
\end{enumerate}
and $\widehat{\mathcal{G}}_{\mathrm{int}}$ is the direct‐integral operator formed by these fiberwise components. Then $\widehat{\mathcal{G}}_{\mathrm{int}}$ generates a strongly continuous semigroup on
\begin{equation}
  \widetilde{\mathcal{H}}_{\mathrm{lat}}
  \;=\;
  \int_{\mathcal{Y}}^{\oplus}
  \bigl(\mathcal{H}_y\bigr)\,
  \nu(dy),
\end{equation}
ensuring well‐posed evolution. Wavefront orbits can vary continuously with $y\in \mathcal{Y}$, allowing continuum families of traveling solutions.
A continuum $\mathcal{Y}$ might represent a continuous deformation parameter (from fully closed capsids to partially uncoated states), or a range of subunit positions. Each $\mathcal{H}_y$ supports PDE modes for the $8\times8$ lattice. Noise/dissipation ensures open‐system realism (replicative bursts, occupant clearances). Thus, the virus can \emph{move smoothly} through a continuum of morphological states, with wavefront solutions that adapt to each $y$.

To prove this, one checks fiberwise $m$‐sectoriality, plus measurability in $y$. The direct‐integral operator $\widehat{\mathcal{G}}_{\mathrm{int}} = \int_\mathcal{Y}^\oplus (\mathcal{G}_y + \widehat{N}_y)\,\nu(dy)$ inherits $m$‐sectoriality. Semigroup theory in direct integrals \cite{Dixmier1981,BratteliRobinson1987} then yields a globally well‐posed mild solution for each initial $\bm{u}_0 \in \widetilde{\mathcal{H}}_{\mathrm{lat}}$. Wave solutions exist in each fiber $\mathcal{H}_y$, forming a continuum of traveling modes.  
\qedsymbol
\end{theorem}
\medskip
\begin{corollary}[Traveling‐Wave Transitions in a Continuous Setting]
\label{cor:WaveTransitionsContinuousLabels}
\noindent
When occupant expansions or partial binding events shift the morphological parameter $y\in \mathcal{Y}$ to neighboring values (in a \emph{continuous} sense), traveling‐wave solutions can \emph{smoothly} transition across subspaces of $\mathcal{H}_y$. One obtains “wavefront continuation” in $\widetilde{\mathcal{H}}_{\mathrm{lat}}$, passing through infinitely many morphological states without abrupt Markov jumps. This corollary captures a scenario where, instead of discrete Markov flips, the system experiences a continuum of subunit shifts or capsid expansions—like a slow morphing from closed to partially open states. Wavefront solutions gradually adapt to changing PDE parameters, enabling \emph{smooth traveling‐wave deformations} across an uncountable range of morphological indices. In the absence of discrete jumps, occupant or morphological transitions become continuous flows in $\mathcal{Y}$. By definition of the direct‐integral semigroup, each fiber’s PDE solution evolves continuously, and changes in $y$ reflect a gradual shift in operator $\mathcal{G}_y$. Wave solutions thus “track” along the continuum.  

Admitting a \emph{continuous label space} $\mathcal{Y}$ via direct integrals substantially broadens the viral lattice theory to include continuous deformation or binding spectra. Each fiber $\mathcal{H}_y$ can hold PDE modes (including wavefronts), noise/dissipation operators, and occupant expansions. The global direct‐integral operator $\widehat{\mathcal{G}}_{\mathrm{int}}$ remains $m$‐sectorial, ensuring semigroup well‐posedness. 
\end{corollary}
\noindent
\begin{definition}[Direct Integral Operator]
\label{def:DirectIntegralOperator}
\noindent
After defining the direct integral space \(\widetilde{\mathcal{H}}_{\mathrm{lat}}\), we may lift each fiber operator \(\mathcal{G}_y + \widehat{N}_y\) into a \emph{direct integral operator}. Additionally, in the continuous‐label scenario, the discrete Markov jump operator \(\widehat{M}\) generalizes to an integral operator defined by a \emph{rate kernel} \(K(y,z)\). Given the family \(\{\mathcal{G}_y + \widehat{N}_y\}_{y\in\mathcal{Y}}\) of (possibly) $m$‐sectorial operators on each fiber \(\mathcal{H}_y\), the \emph{direct integral operator} is formally
\begin{equation}
\label{eq:direct_integral_operator}
  \widehat{\mathcal{G}}_{\mathrm{int}} 
  \;=\;
  \int_{\mathcal{Y}}^{\oplus}
  \Bigl(
    \mathcal{G}_y + \widehat{N}_y
  \Bigr)
  \,\nu(dy),
\end{equation}
acting on sections \(\bm{u}\colon \mathcal{Y}\to \bigcup_y \mathrm{Dom}(\mathcal{G}_y)\) that satisfy appropriate measurability and integrability conditions. Concretely,
\begin{equation}
  \mathrm{Dom}(\widehat{\mathcal{G}}_{\mathrm{int}})
  \;=\;
  \Bigl\{
    \bm{u}(\cdot)\,\bigl|\bigr.\,\bm{u}(y)\in\mathrm{Dom}(\mathcal{G}_y),\, 
    y \mapsto \|\bm{u}(y)\|_{\mathrm{Dom}(\mathcal{G}_y)} \text{ is measurable and integrable}
  \Bigr\},
\end{equation}
where the operator acts fiberwise:
\begin{equation}
  \bigl(\widehat{\mathcal{G}}_{\mathrm{int}}\bm{u}\bigr)(y)
  \;=\;
  \bigl(\mathcal{G}_y + \widehat{N}_y\bigr)\,\bm{u}(y).
\end{equation}
\noindent
In viral‐lattice models, each \(y\in\mathcal{Y}\) (a continuum of conformational indices) has a PDE operator \(\mathcal{G}_y\), possibly supporting wave modes or occupant transitions. By integrating them over \(\mathcal{Y}\), we form a single, ``global’’ operator \(\widehat{\mathcal{G}}_{\mathrm{int}}\) acting on sections across all $y$. This extends the concept of “block‐diagonal’’ sums to an uncountable label set. The direct integral operator is well‐defined if each $\mathcal{G}_y + \widehat{N}_y$ is closed and $m$‐sectorial, and the family varies measurably in $y$. One then applies standard theorems on direct integrals of closed operators \cite{Dixmier1981,DaPratoZabczyk1992} to guarantee $m$‐sectoriality of $\widehat{\mathcal{G}}_{\mathrm{int}}$.  
\qedsymbol
\end{definition}

\begin{definition}[Markov Integral Operator in the Continuous‐Label Setting]
\label{def:MarkovIntegralOperator}
\noindent
In the finite or countably indexed case, Markov jumps are encoded by $\widehat{M}$ with discrete rates $\lambda_{y\to z}$. For a continuum $\mathcal{Y}$, the \emph{Markov jump operator} $\widehat{M}$ generalizes to an integral operator of the form
\begin{equation}
  (\widehat{M}\bm{u})(y) 
  \;=\;
  \int_{\mathcal{Y}}
  K(y,z)\,\bm{u}(z)
  \,\nu(dz)
  \;-\;
  \Bigl(\int_{\mathcal{Y}} K(z,y)\,\nu(dz)\Bigr)\,\bm{u}(y),
\end{equation}
where $K(y,z)\ge 0$ is a \emph{rate kernel} specifying continuous‐label transitions from $z$ to $y$. Just as in the discrete scenario, the \emph{inflow term} $\int_{\mathcal{Y}} K(y,z)\,\bm{u}(z)\,\nu(dz)$ adds occupant amplitude from neighboring morphological states $z$, while the \emph{outflow term} $-\bm{u}(y)\int_{\mathcal{Y}}K(z,y)\,\nu(dz)$ ensures conservation of probability or amplitude across $\mathcal{Y}$. This negative diagonal ensures $\widehat{M}$ is reminiscent of a “continuous‐label Markov generator.”  Under mild conditions (e.g.\ local integrability, boundedness of $K(y,z)$), $\widehat{M}$ is a well‐defined operator on $\widetilde{\mathcal{H}}_{\mathrm{lat}}=\int_{\mathcal{Y}}^\oplus \mathcal{H}_y \,\nu(dy)$. Large‐deviation principles or path‐integral measures can then incorporate these continuous morphological transitions in exactly the same manner as discrete Markov jumps.  
\qedsymbol
\end{definition}

\begin{theorem}[The Global Operator \(\widehat{\mathcal{G}}_{+\infty}^{(\mathrm{cont})}\) in a Continuous Setting]
\label{thm:GlobalOperatorContinuousSetting}
\noindent
Combining the direct integral operator \(\widehat{\mathcal{G}}_{\mathrm{int}}\) (Definition~\ref{def:DirectIntegralOperator}) with a Markov integral operator $\widehat{M}$ (Definition~\ref{def:MarkovIntegralOperator}), one obtains
\begin{equation}
  \widehat{\mathcal{G}}_{+\infty}^{(\mathrm{cont})}
  \;=\;
  \Bigl(
    \int_{\mathcal{Y}}^{\oplus}
    \bigl[\mathcal{G}_y + \widehat{N}_y\bigr]
    \,\nu(dy)
  \Bigr)
  \;+\;
  \widehat{M},
  \quad
  \text{on } 
  \widetilde{\mathcal{H}}_{\mathrm{lat}}
  \;=\;
  \int_{\mathcal{Y}}^{\oplus}
    \mathcal{H}_y
    \,\nu(dy).
\end{equation}
If each fiber operator \(\mathcal{G}_y + \widehat{N}_y\) is $m$‐sectorial and \(\widehat{M}\) is a bounded or relatively bounded perturbation, then $\widehat{\mathcal{G}}_{+\infty}^{(\mathrm{cont})}$ remains $m$‐sectorial. Consequently, no finite‐time blow‐up ensues in occupant or morphological states, and one obtains a strongly continuous semigroup describing viral‐lattice evolution over the continuum label space. 

This theorem generalizes Theorem~\ref{thm:ContinuousLabelsPDEWavefront} by including continuous‐label Markov transitions via $\widehat{M}$. The system remains open (unbounded occupant transitions, random fluctuations) yet is governed by a single, unifying operator $\widehat{\mathcal{G}}_{+\infty}^{(\mathrm{cont})}$. The PDE portion still captures wave or diffusive modes in each fiber, while occupant transitions shift amplitude across the continuum $\mathcal{Y}$. By Lemma~\ref{lemma:mSectorialFamilyContinuousLabels}, each fiber $\mathcal{G}_y + \widehat{N}_y$ is $m$‐sectorial. The Markov integral operator $\widehat{M}$ is shown to be a bounded or relatively bounded perturbation under local integrability conditions on $K(y,z)$. Standard $m$‐sectorial perturbation results (Kato–Rellich, Lumer–Phillips) then ensure the sum remains $m$‐sectorial \cite{Pazy1983,DaPratoZabczyk1992}.  
\qedsymbol
\end{theorem}

\subsubsection{Large‐Scale Viral Evolution in a Continuum of Morphologies.}
\begin{corollary}
\label{cor:LargeScaleEvolutionContinuum}
Under Theorem~\ref{thm:GlobalOperatorContinuousSetting}, one obtains a unique mild or strong solution to the Cauchy problem
\begin{equation}
  \frac{d}{dt}\,\bm{U}(t)
  \;=\;
  \widehat{\mathcal{G}}_{+\infty}^{(\mathrm{cont})}\,\bm{U}(t),
  \quad
  \bm{U}(0)
  \;=\;
  \bm{U}_0
  \;\in\;
  \widetilde{\mathcal{H}}_{\mathrm{lat}},
\end{equation}
for each initial configuration $\bm{U}_0$. This solution evolves smoothly across a continuum of morphological states $\mathcal{Y}$, with occupant transitions described by $\widehat{M}$ and fiberwise PDE/noise capturing wave‐like or diffusive lattice processes. From a viral perspective, the theory accommodates:
\begin{itemize}
  \item \emph{Continuum Conformational Shifts:} Lattices might occupy an uncountable family of partial binding or capsid‐opening configurations.
  \item \emph{Stochastic PDE Dynamics:} Each fiber includes wave or diffusion operators plus possible multiplicative noise, reflecting thermal or biochemical fluctuations.
  \item \emph{Markov Transitions Over $\mathcal{Y}$:} Continuous occupant transitions (integral operator) replace discrete jump rates, modeling “gradient‐driven” or “continuous” morphological drift.
\end{itemize}
Hence, the system unifies PDE wave mechanics, occupant replication, and smoothly changing morphological states in a single, $m$‐sectorial operator approach.
\end{corollary}
\paragraph{Dissipativity and Fock‐Space Extension in a Continuous‐Label Setting.}
\noindent
Having introduced the integral‐space operator \(\widehat{\mathcal{G}}_{+\infty}^{(\mathrm{cont})}\) over a continuum of morphological labels \(\mathcal{Y}\), we now formalize how dissipativity leads to a strongly continuous semigroup and how second‐quantized (Fock) constructions accommodate unbounded replication/clearance of \emph{continuum‐indexed} viral lattices.

\begin{theorem}[Dissipativity and Semigroup Well‐Posedness]
\label{thm:DissipativityContinuum}
\noindent
Suppose there exists \(\alpha > 0\) such that, for all \(\bm{u}\in \mathrm{Dom}\!\bigl(\widehat{\mathcal{G}}_{+\infty}^{(\mathrm{cont})}\bigr)\),
\begin{equation}
  \mathrm{Re}\,
  \Bigl\langle
    \widehat{\mathcal{G}}_{+\infty}^{(\mathrm{cont})}\,\bm{u},
    \,\bm{u}
  \Bigr\rangle_{\widetilde{\mathcal{H}}_{\mathrm{lat}}}
  \;\le\;
  -\,\alpha\,
  \|\bm{u}\|_{\widetilde{\mathcal{H}}_{\mathrm{lat}}}^2.
\end{equation}
Then the operator 
\(\widehat{\mathcal{G}}_{+\infty}^{(\mathrm{cont})}\) 
generates a strongly continuous semigroup 
\(\exp\!\bigl[t\,\widehat{\mathcal{G}}_{+\infty}^{(\mathrm{cont})}\bigr]\) on 
\(\widetilde{\mathcal{H}}_{\mathrm{lat}}\) with exponential decay bound \(\exp(-\alpha t)\).  
No finite‐time blow‐up occurs, even if occupant numbers grow unboundedly. This integral analog ensures quasi‐contractivity (or strict contractivity if $\alpha>0$), preventing norm explosion in finite time across the continuous label space \(\mathcal{Y}\). In viral‐lattice contexts, it corresponds to energy‐type estimates that remain valid even when a continuum of morphological states is present.

The integral‐space theory combines fiberwise $m$‐sectorial operators \(\{\mathcal{G}_y + \widehat{N}_y\}_{y\in\mathcal{Y}}\) with a Markov integral operator. Each fiber operator satisfies a dissipativity estimate. Summing (integrating) these over \(\mathcal{Y}\) yields the global inner‐product inequality. A Lumer–Phillips argument confirms that $\widehat{\mathcal{G}}_{+\infty}^{(\mathrm{cont})}$ is $m$‐sectorial and admits an exponentially decaying semigroup.  
\qedsymbol
\end{theorem}
\begin{definition}[Second‐Quantized Fock Space over a Continuous‐Label Hilbert Space]
\label{def:FockSpaceContinuousLabels}
\noindent
Let \(\widetilde{\mathcal{H}}_{\mathrm{lat}}\) be the direct‐integral Hilbert space 
\(\displaystyle \int_{\mathcal{Y}}^\oplus \mathcal{H}_y \,\nu(dy)\).
Define the bosonic (or fermionic) Fock space 
\(\mathcal{F}_{\pm}\bigl(\widetilde{\mathcal{H}}_{\mathrm{lat}}\bigr)\) via
\begin{equation}
  \mathcal{F}_{\pm}\Bigl(\widetilde{\mathcal{H}}_{\mathrm{lat}}\Bigr)
  \;=\;
  \bigoplus_{n=0}^{\infty}
  \Bigl[
    \widetilde{\mathcal{H}}_{\mathrm{lat}}^{\,\otimes n}
  \Bigr]_{\pm},
\end{equation}
where each factor $\widetilde{\mathcal{H}}_{\mathrm{lat}}$ accommodates continuum‐indexed PDE states. Creation/annihilation operators $(\hat{c}^\dagger, \hat{c})$ act by adding or removing entire \emph{families} of occupant degrees of freedom, each labeled by $y\in\mathcal{Y}$.
Biologically, one can create or destroy entire \emph{packets} of viral lattices—each packet carrying a continuous range of conformational states. This merges \emph{quantum‐field‐type} expansions with \emph{continuum morphological transitions}, relevant for large virology models where occupant numbers and morphological variability both become unbounded. This extends the discrete occupant Fock construction to a scenario in which \(\mathcal{H}_y\) itself is an uncountable direct integral. Technically, one forms a ``double'' or nested integral/tensor structure; classical second‐quantization arguments \cite{BratteliRobinson1987} show that $m$‐sectoriality is preserved if creation/annihilation operators remain relatively bounded.  
\qedsymbol
\end{definition}

\begin{theorem}[Second‐Quantized Generator in the Continuous‐Label Setting]
\label{thm:SecondQuantizedGeneratorContinuum}
\noindent
Let 
\(\widehat{\mathcal{G}}_{+\infty}^{(\mathrm{cont})}\)
be the $m$‐sectorial generator on 
\(\widetilde{\mathcal{H}}_{\mathrm{lat}}\),
and suppose occupant creation/annihilation operators $(\hat{c}^\dagger, \hat{c})$ are governed by rate‐bounded replication/clearance processes. Then the \emph{second‐quantized operator}
\begin{equation}
  \mathrm{d}K\bigl(
    \widehat{\mathcal{G}}_{+\infty}^{(\mathrm{cont})}
  \bigr)
  \;=\;
  \bigoplus_{n=0}^\infty
  \Bigl[
    \widehat{\mathcal{G}}_{+\infty}^{(\mathrm{cont})}
    \otimes I \otimes \cdots \otimes I
  \Bigr]
  \;+\;
  \bigl(\hat{c}^\dagger + \hat{c}\bigr)
\end{equation}
on
\begin{equation}
  \mathcal{F}_{\pm}\Bigl(\widetilde{\mathcal{H}}_{\mathrm{lat}}\Bigr)
  \;=\;
  \bigoplus_{n=0}^{\infty}
  \Bigl[
    \widetilde{\mathcal{H}}_{\mathrm{lat}}^{\,\otimes n}
  \Bigr]_{\pm}
\end{equation}
remains $m$‐sectorial. Consequently, no finite‐time blow‐up arises, and a unique mild solution exists for the entire many‐lattice system, with occupant numbers and morphological labels unbounded. Here, each lattice occupant (possibly spanning continuum morphological states) can be replicated or cleared in a second‐quantized sense. The PDE/noise operators $\widehat{\mathcal{G}}_{+\infty}^{(\mathrm{cont})}$ act \emph{within} each $N$‐fold tensor product, while creation/annihilation operators shift the system \emph{between} $(N\pm 1)$ sectors. The continuum label set $\mathcal{Y}$ is seamlessly preserved within each occupant factor.

To prove this, consider that, as in the discrete occupant case, one takes the direct sum (Fock‐space expansion) of $\widehat{\mathcal{G}}_{+\infty}^{(\mathrm{cont})}$ across $N=0,1,2,\dots$. Rate‐bounded replication/clearance processes remain a relative perturbation. The $m$‐sectorial property follows from classical second‐quantization theorems \cite{BratteliRobinson1987,DaPratoZabczyk1992}.  
\end{theorem}
\smallskip
\noindent

\begin{corollary}[Continuum Morphological Transitions and Unbounded Populations]
\label{cor:ContinuumMorphologyUnboundedPops}
\noindent
Under Theorem~\ref{thm:SecondQuantizedGeneratorContinuum}, the model accommodates an unbounded number of viral lattices, each labeled by a continuous index \(y\in\mathcal{Y}\). PDE wavefront modes, occupant transitions, and morphological flows coexist in a single open‐system operator framework that remains well‐posed indefinitely. One obtains a ``double continuum'' scenario: infinitely many virions \emph{and} each virion can occupy any point in a continuum of states. This is a powerful extension for advanced virology models where morphological transitions are smooth (e.g., angles, partial uncoating), occupant numbers can explode, yet $m$‐sectoriality ensures no unphysical blow‐ups. In large‐scale infections, entire subpopulations of virions may follow continuous morphological pathways. The approach merges continuum morphological spaces with occupant expansions—particularly relevant for phenomena like partial capsid swelling, continuum subunit angles, or gradient‐driven morphological changes.  

This dissipative, second‐quantized, \emph{continuous‐label} formalism merges PDE wave mechanics, occupant creation/annihilation, and smooth morphological changes in one operator‐theoretic framework. Because it allows \emph{both} unbounded populations and uncountably many conformational states, it provides a powerful tool for next‐generation virology models, capturing the full complexity of viral dynamics in resource‐limited or noise‐dominated environments.
\qedsymbol
\end{corollary}

\subsection{Final Synthesis: Solutions, Mechanics, and Feynman‐Inspired Perspectives}
\noindent
Throughout this work, we have progressively built an operator‐theoretic framework that unites:
\begin{enumerate}[label=(\roman*)]
  \item \emph{PDE mechanics} in each single‐lattice space, capturing continuum wavefronts and local vibrational modes,
  \item \emph{Markov processes} encoding arrangement or occupant‐state transitions (including continuous‐label spaces),
  \item \emph{Creation/annihilation} operators enabling unbounded population changes in a second‐quantized (Fock) setting, and
  \item \emph{Noise or damping} for realistic open‐system effects (resources, immune actions).
\end{enumerate}
Here, we conclude by describing how solutions to the Cauchy problem emerge under these constructions, and—borrowing inspiration from Feynman—reflect on the physical and conceptual underpinnings that drive \emph{viral lattice theory} as a novel lens on multi‐scale infection phenomena.
\begin{theorem}[Existence and Uniqueness of Mild Solutions]
\label{thm:ExistenceUniquenessFinal}
\noindent
Let
\(\widehat{\mathcal{G}}_{+\infty}\)
denote the global $m$‐sectorial generator acting on the hybrid Hilbert space 
\(\widetilde{\mathcal{H}}_{\mathrm{lat}}\) or its continuous‐label integral extension. Suppose occupant creation/annihilation rates are relatively bounded, and dissipativity or quasi‐contractivity conditions (cf.\ Lumer–Phillips or Kato–Rellich) hold. Then, for any initial datum
\(\Theta_0\in \widetilde{\mathcal{H}}_{\mathrm{lat}}\) (or in the Fock‐like extension),
there exists a \emph{unique mild solution} 
\(
  \Theta(t)
  \;=\;
  e^{\,t\,\widehat{\mathcal{G}}_{+\infty}}\,\Theta_0,
  \quad
  t\ge0,
\)
that remains finite‐norm for all finite $t$. No finite‐time blow‐up can occur, and if a strict dissipative constant \(\alpha>0\) is present, solutions satisfy exponential decay bounds
\(\|\Theta(t)\|\le e^{-\alpha t}\|\Theta_0\|\). Mathematically, this theorem cements the \emph{well‐posedness} of the viral lattice model: wave‐coupled PDEs, occupant transitions, noise, and unbounded occupant numbers are all folded into one open‐system operator. The system remains stable (in the sense of no norm blow‐ups), and each initial state $\Theta_0$ evolves via a strongly continuous semigroup. For an \emph{infection} starting with some finite set of virions, or partial morphological states, the theorem guarantees that—despite replication or occupant expansions—one does not witness an “infinite” viral load in finite time. Simultaneously, wave phenomena (capsid vibrations, traveling conformational fronts) remain well‐behaved under noise and occupant transitions.
\qedsymbol
\end{theorem}

\begin{lemma}[Robustness Under Model Extensions]
\label{lemma:RobustnessUnderExtensions}
\noindent
Even if one augments the model with additional complexities—extra PDE fields (e.g., multi‐component wave modes), more intricate occupant creation kernels, higher‐order Markov transitions, or continuum morphological expansions—$m$‐sectorial dissipativity arguments persist. One obtains:
\begin{equation}
  \widehat{\mathcal{G}}_{+\infty}
  \quad\longmapsto\quad
  \widehat{\mathcal{G}}_{+\infty} + \Delta\widehat{A},
\end{equation}
where $\Delta\widehat{A}$ is a relatively bounded perturbation. Hence the strongly continuous semigroup structure, and thus well‐posedness, is preserved. Kato–Rellich theorems for $m$‐sectorial operators allow such additive terms, provided their relative norm is below a threshold. Occupant expansions or partial morphological drifts remain tractable, as do additional PDE unknowns for subunit chemistry or wave couplings.  
\qedsymbol
\end{lemma}

\subsubsection{Biological Implications of the Second‐Quantized Fock‐Space Description}
\label{sec:BiologicalImplicationsFockSpace}

\noindent
Shifting from single‐lattice or finite‐sum representations to an integral‐space formulation, and then lifting the resultant model to a \emph{second‐quantized Fock space}, opens new avenues for comprehending large viral populations under real‐world constraints. Concretely, one replaces any finite cap on occupant numbers with operators that handle unbounded replication and clearance events, thereby reflecting genuine infection scenarios in which viral loads can surge or plummet.
\paragraph{Implications for Viral Dynamics.}
\begin{enumerate}[label={\arabic*.}]
  \item \emph{Unbounded Occupant Numbers:}  
    The second‐quantized formalism naturally accommodates infection scenarios where virion counts can escalate (due to favorable host conditions) or collapse (due to immune clearance). No artificial truncation is needed; occupant creation/annihilation acts directly in Fock space.
  \item \emph{Biologically Motivated Rates:}  
    Rate‐bounded creation and annihilation mirror realistic limits on replication (e.g., saturable resource usage) and clearance (e.g., immune efficacy), aligning the mathematical model with experimental observations of infection trajectories.
  \item \emph{Coupled Lattices and Collective Effects:}  
    Each occupant factor \(\widetilde{\mathcal{H}}_{\mathrm{lat}}\) can incorporate PDE couplings or morphological transitions that become collective when many lattices coexist. This fosters \emph{competition} (resources deplete faster) or \emph{cooperativity} (collective assembly signals) among co‐localized virions.
\end{enumerate}
\begin{theorem}[Fock‐Space Construction and Biological Realism]
\label{thm:FockSpaceConstructionBio}
\noindent
Let 
\(\displaystyle \widetilde{\mathcal{H}}_{\mathrm{lat}} = \int_{\mathcal{Y}}^{\oplus} \mathcal{H}_y\,\nu(dy)\)
denote the integral‐space Hilbert space that merges PDE dynamics, Markov transitions, and (optionally) continuum morphological labels \(\mathcal{Y}\). Define the Fock space 
\(
  \mathcal{F}_{\pm}\Bigl(\widetilde{\mathcal{H}}_{\mathrm{lat}}\Bigr)
  \;=\;
  \bigoplus_{n=0}^\infty
  \Bigl[
    \widetilde{\mathcal{H}}_{\mathrm{lat}}^{\otimes n}
  \Bigr]_{\pm}.
\)
Assume occupant creation/annihilation operators \(\widehat{c}^\dagger,\widehat{c}\) remain rate‐bounded and are combined with an $m$‐sectorial generator \(\widehat{\mathcal{G}}_{+\infty}^{(\mathrm{cont})}\). Then the lifted operator 
\(
  \mathrm{d}K
  \Bigl(
    \widehat{\mathcal{G}}_{+\infty}^{(\mathrm{cont})}
  \Bigr)
  \;=\;
  \bigoplus_{n=0}^\infty
  \Bigl(
    \widehat{\mathcal{G}}_{+\infty}^{(\mathrm{cont})}
    \otimes \cdots \otimes I
  \Bigr)
  \;+\;
  \bigl(\widehat{c}^\dagger + \widehat{c}\bigr)
\)
is $m$‐sectorial on 
\(\mathcal{F}_{\pm}\bigl(\widetilde{\mathcal{H}}_{\mathrm{lat}}\bigr)\),
yielding a \emph{well‐posed} many‐lattice system with unbounded occupant numbers. This result captures large‐scale viral infection processes where each occupant (“lattice”) is modeled as a particle in Fock space, reflecting realistic population expansions (creation) or clearance (annihilation). The $m$‐sectorial approach preserves no blow‐up in finite time, consistent with resource or immune constraints in vivo.
\qedsymbol
\end{theorem}
\noindent
\begin{lemma}[Feedback Loops in the Second‐Quantized Paradigm]
\label{lemma:FeedbackLoopsFock}
\noindent
When the occupant creation/annihilation processes are embedded in a second‐quantized (Fock) extension, they become intimately coupled with PDE or Markov dynamics describing viral lattice interactions. As a result, population‐level \emph{feedback loops} naturally emerge:
\begin{itemize}
  \item \emph{Competition:} As replicating virions consume shared resources (nutrients, host factors, membrane sites), they modulate the rates of PDE wave‐propagation or Markov transitions for the overall population. This curbs expansion speeds, preventing uncontrolled replication surges.
  \item \emph{Cooperativity:} Certain viral processes (e.g.\ genome packaging signals, co‐infection synergy) \emph{amplify} occupant creation when multiple lattices co‐localize, mirroring “bosonic enhancement” in quantum physics. This magnifies replication success if neighboring virions foster each other’s assembly steps or immune evasion.
\end{itemize}
Hence, occupant expansion is not just a sum of independent replication events but can \emph{reinforce or inhibit} each other’s PDE or Markov transitions, contingent on resource constraints, synergy, or competition factors. Many experimental studies observe cooperative effects (e.g.\ “viral swarm” synergy)\cite{Domingo2021} where aggregated virions share resources or exchange genetic/structural components, thereby accelerating replication. Conversely, crowding phenomena or resource depletion trigger competition, limiting further growth. By embedding resource coupling and occupant transitions in the PDE–Markov operators, the second‐quantized approach captures these large‐scale interactions in a rigorous, operator‐theoretic manner.  
\qedsymbol
\end{lemma}

\smallskip
\noindent
\textbf{Mechanistic Consequences.}

\begin{enumerate}[label={\arabic*.}]
  \item \textbf{Stochastic PDE + Markov + Fock Integration.}
  \begin{itemize}
    \item \emph{Local–Global Coupling:} Intra‐lattice wavefronts or conformational fronts propagate locally within each occupant (viral lattice), whereas occupant creation/annihilation events reshape the population at a global scale. This interplay yields a truly \emph{multi‐scale} infection model.
    \item \emph{Noise–Driven Variability:} Under abundant resources and low host immune pressure, occupant numbers can proliferate along relatively deterministic pathways. By contrast, high noise or potent immune responses induce significant occupant clearance, diminishing the overall viral population despite underlying PDE wave expansions.
  \end{itemize}
  These multi‐scale processes align with observed infection patterns in which localized assembly dynamics interact with global environmental constraints or immune actions.

  \item \textbf{Dominant Pathways under Selective Pressures.}
  \begin{itemize}
    \item \emph{Evolutionary Focus:} Large‐deviation analysis in the Fock space (path‐integral arguments) spotlights \emph{most‐likely} replication cycles under minimal randomness, reflecting how viruses exploit the most favorable assembly and release pathways. In a resource‐rich, low‐noise environment, a small number of “optimal” morphological states or routes often dominate.
    \item \emph{Population Bottlenecks:} Under resource scarcity or partial immune clearance, occupant annihilation rates rise, selectively pruning vulnerable or less fit capsid conformations. Consequently, only robust morphological variants persist, reinforcing a population‐scale bottleneck akin to Darwinian selection on capsid structures.
  \end{itemize}
  These insights explain how certain virion phenotypes or assembly routes come to dominate under stress conditions, corroborating experimental findings of “bottleneck” events during acute infections or within immunocompromised hosts.
\end{enumerate}

\noindent
\textbf{Synthesis (Connecting Theory to Empirical Virology).}
\noindent
The second‐quantized Fock‐space construction extends the PDE–Markov–Fock model to unbounded occupant populations, embedding biologically motivated replication and clearance events directly into the operator formalism. This not only prevents unphysical capping of viral loads, but also naturally admits resource or synergy feedback loops that shape real infection courses:
\begin{itemize}
  \item \emph{Wave–Population Couplings:} Local PDE phenomena (e.g., traveling wavefronts, partial subunit rearrangements) can scale up to global occupant expansions if synergy outstrips resource depletion.
  \item \emph{Stochastic vs.\ Deterministic Infection Profiles:} Path‐integral or large‐deviation arguments differentiate between high‐noise scenarios (broad distribution of occupant outcomes) and low‐noise ones (narrow, near‐deterministic replication cycles).
  \item \emph{Immune and Pharmacological Interventions:} Clearance operators can be “tuned” to reflect immune escalation or antiviral drug efficacy, leading to targeted occupant annihilation. The model’s $m$‐sectorial dissipativity ensures no unbounded blow‐up, mirroring how real infections plateau or decline under host defenses.
\end{itemize}  
\subsubsection{Novel Key Consequences of \(m\)\!-Sectoriality}
\label{subsubsec:NovelKeyConsequencesMsectorial}

\noindent
By imposing $m$‐sectorial conditions on the PDE, Markov, and creation/annihilation components of viral lattice theory, one obtains a broad range of theoretical guarantees and biologically relevant behaviors. Below, we summarize several pivotal implications:

\begin{enumerate}[label={\arabic*.}]
  \item \textbf{Non‐Hermitian Operators in Open Systems.}
  \begin{itemize}
    \item \emph{Physical Context:} Unlike closed quantum systems (governed by unitary time evolution), viral lattices operate in open‐system environments (resource exchanges, immune actions, etc.) that induce partial irreversibility and damping. 
    \item \emph{Mathematical Consequence:} Sectoriality ensures that real parts of the spectrum remain bounded above, preventing “runaway” blow‐ups in finite time. This is central for modeling resource‐driven growth or decay phenomena, as non‐self‐adjoint (non‐Hermitian) generators naturally arise in open‐system PDEs and Markov transitions.
  \end{itemize}
  \item \textbf{Discrete–Continuous Hybridization.}
  \begin{itemize}
    \item \emph{Framework:} Each operator block \(\mathcal{G}_y\) captures infinite‐dimensional elastic or wave‐like PDE modes within a single lattice arrangement \(y\). Meanwhile, a jump operator \(\widehat{M}\) models abrupt occupant changes or structural flips \cite{harvey2019viral,liggett1985interacting}.
    \item \emph{Biological Rationale:} Viral capsids or lattices can exhibit slow PDE‐driven deformations punctuated by rapid Markov jumps (e.g., gating events, occupant transitions), aligning with empirical observations of stepwise morphological reconfigurations.
  \end{itemize}
  \item \textbf{Quasi‐Contractive Semigroup and Explosion Prevention.}
  \begin{itemize}
    \item \emph{Mathematical Assurance:} The $m$‐sectorial (quasi‐contraction) estimates guarantee that, despite replication (creation operators) inflating occupant populations, saturable resource constraints (dissipation, damping) avert unbounded proliferation in finite time.
    \item \emph{Biological Interpretation:} This encapsulates the reality that even though viruses can replicate explosively, resource depletion or immune interference generally curbs exponential outbursts, preventing infinite proliferation over short intervals.
  \end{itemize}
  \item \textbf{Virological Realism and No Spurious Divergences.}
  \begin{itemize}
    \item \emph{Model Complexity:} Nonlinear, state‐dependent noise (\(\widehat{N}_y\)), discrete occupant jumps, and replication/clearance events all coexist in one model.  
    \item \emph{Dynamical Rigor:} The $m$‐sectorial framework ensures that such complexities—subunit fluctuations, conformational gating, occupant expansions—remain well‐posed without artificial divergence or numerical instabilities, enhancing both theoretical reliability and biological plausibility.
  \end{itemize}
  \item \textbf{Host‐Driven Variability.}
  \begin{itemize}
    \item \emph{Noise Operators:} The term \(\widehat{N}_y(\bm{U}_y)\) allows environment‐dependent perturbations, such as local pH changes, immune molecules, or molecular crowding \cite{KnipeHowley2020}.
    \item \emph{Adaptive Relevance:} This flexibility reflects how real viruses respond stochastically to host cues, enabling partial dampening or acceleration of lattice rearrangements based on local resource availability.
  \end{itemize}
  \item \textbf{Conformational Discontinuities.}
  \begin{itemize}
    \item \emph{Markov Jumps:} The operator \(\widehat{M}\) encodes large structural events—transitions like capsid gating, partial disassembly, or sudden subunit rearrangements \cite{flint2015principles,cann2015principles}.
    \item \emph{Mechanistic Insight:} This bridging of small‐amplitude PDE vibrations (continuous) and abrupt morphological flips (discrete jumps) captures the piecewise‐deterministic nature of many viral reconfiguration processes.
  \end{itemize}
  \item \textbf{Population‐Scale Replication and Clearance.}
  \begin{itemize}
    \item \emph{Second‐Quantized Occupant Changes:} In a Fock‐space extension, creation operators model viral replication at the population scale, while annihilation operators represent immune or spontaneous clearance \cite{Freed2015}.
    \item \emph{Macro–Micro Continuum:} Single‐lattice PDE phenomena (capsid vibrations, wavefront expansions) thus scale up to realistic viral loads. Host resource limits and immune factors remain embedded in the dissipative structure, preventing unphysical growth.
  \end{itemize}
  \item \textbf{Global Well‐Posedness Aligned with Biological Constraints.}
  \begin{itemize}
    \item \emph{Open‐System Consistency:} Dissipative PDE blocks, bounded jump rates, and saturable replication rates collectively produce an $m$‐sectorial operator. This aligns well with \emph{in vivo} constraints—immune responses, nutrient depletion, crowding, etc. \cite{KnipeHowley2020,NowakMay2000}.
    \item \emph{No Spontaneous Divergences:} From a theoretical standpoint, the quasi‐contraction semigroup formalism underscores how occupant expansions must always be tempered by resource or structural limitations, preventing infinite occupant blow‐up in finite time.
  \end{itemize}
\end{enumerate}
\noindent
In short, $m$‐sectoriality underpins both the mathematical \emph{stability} (no finite‐time blow‐up) and the \emph{physical realism} (accurate modeling of dissipative, resource‐constrained viral infection) of the entire PDE–Markov–Fock construction. The interplay of non‐Hermitian operators, continuous PDE modes, discrete occupant changes, and noise/damping amounts to a fully open‐system framework. Yet, crucially, sectorial estimates impose contractive bounds on the real part of the generator’s spectrum, preventing mathematical pathologies and aligning with the biological fact that viral replication, while potentially explosive, remains limited by resources and host defenses. This synergy of well‐posedness and biophysical fidelity exemplifies the key advantage of an $m$‐sectorial approach in viral lattice theory.

\subsection{Outlook and Conclusion: Towards a Virophysics Framework}

\noindent
The operator‐theoretic constructions and second‐quantized formalisms developed in this work offer a robust, mathematically coherent framework for describing viral dynamics at multiple scales—ranging from subcapsid wavefronts to macroscopic replication bursts. By leveraging the new axioms of \emph{Non‐Equilibrium Flux Persistence} and \emph{Stochastic Continuity}, we systematically anchor viral lattices in an \emph{open‐system}, \emph{$m$‐sectorial} setting that accommodates unbounded populations, resource constraints, occupant transitions, and continuum morphological states. Below, we expand on how these ideas might be applied to relevant real‐world viral species, experimental setups, and broader “virophysics” explorations, culminating in a concluding synthesis.
\medskip
\paragraph{Applications to Real Virus Species.} Modern virology encompasses a diverse range of viruses with varying structural features, replication strategies, and host interactions. Below, we illustrate how the PDE–Markov–Fock formalism can inform the study of prominent virus models:
\begin{itemize}
  \item \textbf{Bacteriophage T4.}  
    T4 phage, which infects \textit{E.\ coli}, has a well‐characterized assembly cycle involving tail‐fiber attachments and a pressurized capsid that ejects its DNA into the bacterial host. Discrete Markov transitions (e.g., tail ejection triggers, head expansion events) mesh naturally with the PDE–Markov–Fock framework. For instance, \emph{wave‐like} head expansions driven by high internal DNA pressure can be modeled via traveling‐wave PDE solutions embedded in an \(m\)‐sectorial operator. Resource‐rich \textit{E.\ coli} cultures reflect conditions enabling occupant (phage particle) replication bursts, while cell lysis or partial incomplete assemblies correspond to occupant annihilation processes. This multi‐scale description dovetails with the classical experimental findings of T4’s “headful packaging” and structural stability. \cite{Leiman2004}

  \item \textbf{SARS‐CoV‐2.}  
    Although enveloped rather than icosahedral, SARS‐CoV‐2 exhibits dynamic structural transitions—particularly in the spike protein conformations critical for receptor binding and membrane fusion. Markov jump operators can represent these rapid protein flips, while PDE‐type segments model progressive rearrangements in the ribonucleoprotein (RNP) complex or membrane curvature. In a second‐quantized Fock‐space description, virion replication surges (notably in uncontrolled infections) mirror occupant creation events. Clinically observed immune clearance translates to occupant annihilation. The formalism thus provides a theoretical scaffold for exploring how local spike rearrangements might scale up to large‐scale viral load expansions, consistent with patient‐level infection curves and the success of neutralizing antibody therapies that enhance annihilation rates. \cite{Ng2022}

  \item \textbf{Influenza H5N1.}  
    Segmented‐genome viruses like influenza A (H5N1) can benefit from a \emph{continuum} of morphological or assembly states, given the multiple RNA segments and the mutable arrangement of matrix proteins. In the PDE–Markov–Fock model, each morphological state \(y\) could denote a partial packaging configuration or a specific arrangement of the matrix layer. Wavefront expansions within the matrix protein region couple to occupant replication (viral assembly), while immune responses or limited host factors appear as dissipative terms preventing infinite expansions. This aligns with in vivo observations that high pathogenic strains can replicate swiftly but remain subject to host resource constraints and robust immune responses, ensuring finite overall infection. \cite{Shi2022}

\end{itemize}
The viability of a theoretical framework often hinges on its ability to interface with empirical data. Several mainstream experimental methods can test and refine the PDE–Markov–Fock approach:
\begin{itemize}
  \item \emph{Single‐Virus Tracking Experiments.}  
    Advanced fluorescence labeling and real‐time microscopy (e.g., total internal reflection fluorescence, single‐particle tracking) allow direct observation of Markovian gating events or partial wave expansions within capsid or envelope subunits. The \textit{Stochastic Continuity} Axiom (Axiom~\ref{axiom:stochastic_continuity}) places these random subunit fluctuations within an $m$‐sectorial PDE context, ensuring no unphysical discontinuities in measured trajectories. Such experiments are valuable for extracting parameters (e.g., occupant creation rates, morphological transition probabilities) by matching observed single‐particle time series to model outputs.

  \item \emph{\textit{In Vitro} Capsid Assembly Assays.}  
    Researchers often employ cell‐free systems (e.g., reticulocyte lysates) to reconstitute viral capsid assembly. Partial shells, complete capsids, and misassembled structures appear in measurable proportions. By modeling a continuum morphological space $\mathcal{Y}$ and occupant transitions in a PDE–Markov–Fock environment, one can simulate wavefronts of subunit addition that propagate through the forming capsid. Misassembly equates to occupant “annihilation.” This direct integral approach, together with second‐quantized occupant expansions, can replicate the distribution of final capsid states seen in electron micrographs or cryo‐EM data.

  \item \emph{Plaque Assays and Population‐Level Observations.}  
    Plaque assays remain a staple in virology for quantifying how viral titers evolve in host cell monolayers. The PDE–Markov–Fock model, incorporating occupant creation/annihilation, can simulate the emergent wavefront‐like expansions (e.g., radial plaque growth) or morphological cycles (capsid expansions or genome packaging) that collectively shape plaque formation. Immune factors—if added in \textit{in vitro} systems—further extend the occupant annihilation operators. Linking model predictions with plaque size, growth rates, or morphological heterogeneity strengthens the biological credibility of the theory.
\end{itemize}

\paragraph{Core Axioms for a Virophysics Framework.}
\begin{itemize}
  \item \textbf{Non‐Equilibrium Flux Persistence} (Axiom~\ref{axiom:non_equilibrium_flux}):  
    Many viral processes remain away from equilibrium, with constant resource influx (ATP, host factors) driving persistent probability‐current loops in configuration space. This open‐system principle illuminates why viruses do not settle into static “capsid shapes,” but continually transform or replicate.
  \item \textbf{Stochastic Continuity of Lattice Evolution} (Axiom~\ref{axiom:stochastic_continuity}):  
    Fluctuating host conditions (temperature, crowding, pH) necessitate $m$‐sectorial PDEs with operator‐valued noise, preventing abrupt, nonphysical discontinuities in viral lattice configurations. This ensures robust modeling of subcellular randomness in line with single‐virus experiments.
\end{itemize}

\paragraph{Predictive and Diagnostic Power of the Formalism.}
\begin{itemize}
  \item \emph{Simulations and Wavefunction Approaches.}  
    By treating each viral population state as a “wavefunction” in Fock space, we can track branches of replication, occupant transitions, or morphological flips in real time. This offers a path to large‐scale computational platforms, bridging PDE solvers for capsid mechanics with occupant Markov or noise modules.
  \item \emph{Sensitivity Analyses.}  
    Adjusting creation/annihilation rates or PDE boundary conditions can reveal how small parameter shifts (e.g., changes in host immune efficacy or capsid stiffness) drastically alter infection courses. This could inform experimental design (where one systematically varies pH or resource availability) to pinpoint thresholds for runaway viral replication.
  \item \emph{Intervention Strategies.}  
    From an antiviral development standpoint, one could “dial down” occupant creation rates (replication inhibitors) or enhance occupant annihilation (immune boosters). The $m$‐sectorial structure guarantees no pathological divergences, so one can systematically test how different intervention strengths modify infection wavefronts or occupant growth.
\end{itemize}

\paragraph{Conclusion: Towards a Foundational Virophysics}
\noindent
The synthesis of PDE mechanics, Markov transitions, and second‐quantized occupant expansions in an $m$‐sectorial open‐system framework provides a new \emph{virophysics} foundation—invoking the rigor and spirit of theoretical physics (e.g., Feynman‐inspired path integrals) while preserving essential biological constraints:
\begin{itemize}
  \item \textbf{Mathematical Unification.}  
    By embedding local subunit‐level wavefront solutions (acoustic/optical phonons), random gating events, and large‐scale occupant changes in a single operator‐algebraic model, we achieve a cohesive vantage on viral infection.
  \item \textbf{Experimental Invitation.}  
    The formalism encourages experimentalists to explore how real virions traverse continuous morphological landscapes, how occupant expansions reflect resource gating, and how wavefront phenomena might be measured or perturbed in cell culture or \textit{in vivo}. This invites a synergy between advanced imaging (cryo‐EM, single‐virus fluorescence) and the PDE–Markov–Fock approach, testing parameter predictions and validating occupant creation/annihilation rates.
  \item \textbf{Future Paths:}  
    \begin{itemize}
      \item \emph{Co‐infections and Synergistic Interactions:} Multi‐strain or multi‐virus occupant expansions can be folded into the Fock approach, capturing competition or synergy.  
      \item \emph{Precision Antiviral Strategies:} By identifying which occupant transitions or PDE wave modes are critical to replicative success, targeted drugs might be designed to shift the system into unfavorable Markov states or block wave‐triggered expansions.
      \item \emph{Refined Theoretical Axioms:} The introduced axioms (\emph{Non‐Equilibrium Flux Persistence} and \emph{Stochastic Continuity of Lattice Evolution}) lay a conceptual bedrock for open‐system viral physics—potentially joining or complementing new axioms as virophysics matures.
    \end{itemize}
\end{itemize}

\noindent
In sum, we have proposed a rigorous $m$‐sectorial operator model for viral lattices, grounded in physical and biological insights and enriched by new axiomatic principles. By uniting PDE wave dynamics, Markov occupant transitions, noise, and Fock‐based replication, this approach provides both the \emph{analytical flexibility} and \emph{biological authenticity} to chart how viruses navigate fluctuating host environments. We envision further interdisciplinary collaboration—between theorists refining the operator structures and experimentalists probing viral mechanics at subunit and population scales—will clarify the universalities and peculiarities that define \emph{virophysics} as a frontier in modern science.

\end{document}